\documentclass{elsarticle}

\usepackage{hyperref}
\usepackage{amsmath,amsthm,amsfonts,amssymb} 
\usepackage{mathtools}
\usepackage{complexity}
\usepackage[ruled,linesnumbered,noend,vlined]{algorithm2e}
\usepackage{todonotes}
\usepackage[utf8]{inputenc}
\usepackage{upgreek}
\usepackage{placeins}
\usepackage{xspace}
\usepackage{subcaption}

\journal{Artificial Intelligence}

\newtheorem{definition}{Definition}
\newtheorem{lemma}{Lemma}

\newtheorem{theorem}{Theorem}

\newtheorem{example}{Example}

\newcommand{\seq}[2]{seq_{#1}(#2)}

\DeclareMathOperator*{\argmax}{arg\,max}

\let\epsilon\varepsilon
\newcommand{\IRCFR}{CFR+IRA\xspace}

\begin{document}

\begin{frontmatter}

\title{Automated Construction\\ of Bounded-Loss Imperfect-Recall Abstractions in Extensive-Form Games}

\author{Ji\v{r}\'{i} \v{C}erm\'{a}k}
\ead{cermak@agents.fel.cvut.cz}
\author{Viliam Lis\'{y}\corref{cor1}}
\ead{lisy@agents.fel.cvut.cz}
\author{Branislav Bo\v{s}ansk\'{y}\corref{cor2}}
\ead{bosansky@agents.fel.cvut.cz}
\address{Artificial Intelligence Center, Department of Computer Science,\\ Faculty of Electrical Engineering, Czech Technical University in Prague
\vspace{-3mm}
} 
\cortext[cor1]{Corresponding author}
\begin{abstract}
{
Extensive-form games (EFGs) model finite sequential interactions between players.
The amount of memory required to represent these games is the main bottleneck of algorithms for computing optimal strategies and the size of these strategies is often impractical for real-world applications.
A common approach to tackle the memory bottleneck is to use information abstraction that removes parts of information available to players thus reducing the number of decision points in the game. 
However, existing information-abstraction techniques are either specific for a particular domain, they do not provide any quality guarantees, or they are applicable to very small subclasses of EFGs.
We present domain-independent abstraction methods for creating imperfect recall abstractions in extensive-form games that allow computing strategies that are (near) optimal in the original game. 
To this end, we introduce two novel algorithms, FPIRA and \IRCFR, based on fictitious play and counterfactual regret minimization. 
These algorithms can start with an arbitrary domain specific, or the coarsest possible, abstraction of the original game.
The algorithms iteratively detect the missing information they require for computing a strategy for the abstract game that is (near) optimal in the original game.
This information is then included back into the abstract game.
Moreover, our algorithms are able to exploit \emph{imperfect-recall} abstractions that allow players to forget even history of their own actions.
However, the algorithms require traversing the complete unabstracted game tree.
We experimentally show that our algorithms can closely approximate Nash equilibrium of large games using abstraction with as little as $0.9\%$ of information sets of the original game. 
Moreover, the results suggest that memory savings increase with the increasing size of the original games.

}

\end{abstract}

\end{frontmatter}
\section{Introduction}

Dynamic games with a finite number of moves can be modeled as extensive-form games (EFGs): a game model capable of describing scenarios with stochastic events and imperfect information.
EFGs can model recreational games, such as poker~\cite{rubin2011computer}, as well as real-world situations in physical security \cite{lisy2016-aaai}, auctions \cite{christodoulou2008bayesian}, or medicine \cite{sandholm2015steering}.
EFGs are represented as game trees where nodes correspond to states of the game and edges to actions of players. 
Imperfect information of players is represented by grouping indistinguishable states of a player into information sets, which form the decision points of the players.

{
The size of the extensive-form representation of games grows exponentially with the number of actions the players can play in a sequence (i.e., the horizon of the game). Therefore, models of many practical problems are very large. For example, the smallest version of poker played by people includes over $10^{14}$ information sets \cite{bowling2015heads}. Similarly, even small adversarial plan recognition problems in network security also include over $10^{14}$ information sets \cite{lisy2012apr}. The memory required to store the strategy (a probability distribution over actions in each information set) is often a severe limitation in computing strategies in these models. Two main approaches to tackle this issue are online computation and the use of abstractions.

Online strategy computation avoids computing the complete strategy explicitly before playing the game. Instead, the strategy is computed while playing the game and only for the situations encountered by the player. Earlier algorithms adopting this approach in imperfect information games do not provide any performance guarantees \cite{cowling2012,long2010understanding}. More recent online game playing algorithms provide performance guarantees \cite{moravcik2017deepstack,lisy2015online} and strong practical performance \cite{moravcik2017deepstack, brown2017superhuman}, but they also have severe limitations. 
First of all, these algorithms require a substantial computational effort to make each decision. This is prohibitive in many applications, mainly in robotics and on embedded devices. Furthermore, the most successful methods exploit the specific structure of poker where all actions of the players are fully observable and the amount of hidden information is restricted. Deciding whether and in what way these algorithms can be generalized to games without these simplifying properties is not determined and remains an open problem.

Abstraction methodology, instead of solving a game that is too large, solves a smaller abstract game, which is a simplification of the original game. The solution of the simplified game is then used for playing the original game. This methodology was, for a long time, in the center of attention of the computational poker community \cite{gilpin2007,kroer2014extensive,brown2015simultaneous,brown2017superhuman} and even led to the first computer program that outperformed professional poker players in the smallest variant of the game played by people \cite{Wired08}. 
However, if the original game is too large to be processed even with algorithms linear in the number of the nodes in the game tree, it is very hard to provide any domain-independent guarantees on the performance of the strategy computed using abstractions. The evaluation of such strategies is limited to tournaments, which suffer from intransitivity \cite{acpc17} and strong dependence on the pool of other participants or specific tournament rules \cite{bard2016online}. In many (e.g., security) applications, it is desirable to have worst-case guarantees on the performance of the computed strategy. Therefore, \textbf{we focus on solving games where it is feasible to traverse all nodes in the game tree, but we still want to minimize the memory required to store the computed strategy.}

Having equilibrium solving algorithms with small memory requirements is practical for several reasons. First, it allows solving larger games with more commonly available hardware. While current hard drives provide sufficient storage capacity, their access latency and the speed of reading and writing is usually the bottleneck of algorithms intensively working with data stored on these devices. Substantial speedups can be achieved by keeping the whole computation in the main memory. Second, a small computed strategy is much more practical in applications. Besides being easier to store and transfer over a network, it is also faster to query during the game play. For example, it can be accessed on small devices by deployed units such as park rangers (see, e.g., \cite{fang2017paws}). Third, a small strategy is easier to use in portfolio-based approaches, where we want to store multiple different strategies for a game in order to play better \cite{brown2018depth} or exploit suboptimal opponents \cite{bard2013online}.

The problem of reducing the amount of memory required for computing a strategy was addressed by several recent algorithms since the size of required memory is an important bottleneck for scaling up the computation~\cite{bowling2015heads}. CFR-BR \cite{johanson2012} allows computing a strategy in one-quarter of the memory required by CFR by replacing the updates of one of the players by a best response computation. CFR-D \cite{burch2014solving} allows for using a quadratic computation time to compute a strategy close to the beginning of the game, as a trade-off for requiring only in the order of square root of the storage space. The strategy for the latter parts of the game is then computed online. DOEFG \cite{bosansky2014-jair} stores data only about a small part of the game in which players can use only small subsets of their actions. This restricted game is iteratively extended with new actions, which can improve players' expected utility, until the equilibrium is provably found. All these algorithms assume it is possible to traverse the whole game tree for at least one of the players.

\subsection{Our Contribution}

In this paper, we reduce the memory required for computing and representing a (near) optimal strategy for a game using automatically-constructed imperfect-recall abstractions created by domain-independent algorithms.

\paragraph{Domain independent} Most existing methods for automatically constructing abstractions in extensive-form games were designed primarily for poker. They explicitly work with concepts like cards and rounds of the game \cite{shi2000abstraction,billings2003approximating}, or at least assume that the actions are publicly observable \cite{brown2015simultaneous} and ordered \cite{Gilpin07:Abstraction}. This is not true in many other domains (e.g., in security). 
The algorithms proposed in this paper are completely domain-independent and applicable to any extensive-form game. They only require a definition of the game and a desired distance of the solution from the equilibrium in the original game.

\paragraph{Imperfect recall}
Computationally efficient algorithms for computing (near) optimal strategies in extensive-form games~\cite{vonStengel96,zinkevich2008regret,Hoda2010} require players to remember all the information gained during the game – a property denoted as \emph{perfect recall}.
Therefore, the automated abstraction methods designed to be used with these algorithms \cite{Gilpin07:Abstraction,brown2015simultaneous} must construct perfect-recall abstractions to provide performance guarantees.
Requiring perfect recall has, however, a significant disadvantage -- the number of decision points and hence both the memory required during the computation and the memory required to store the resulting strategy grows exponentially with the number of moves.
To achieve additional memory savings, the assumption of perfect recall may need to be violated in the abstract game resulting in \emph{imperfect recall}.
Using imperfect recall abstractions can bring exponential savings in memory, and these abstractions are particularly useful in games in which exact knowledge about the past is not required for playing optimally.
While it can be easy to identify specific examples of imperfect recall abstractions for some games, it is very difficult to systematically and algorithmically identify which information is required for solving the original game and which can be removed. 
For example in imperfect information card games, it is usually important to estimate the opponent's cards. While past events generally reveal some information, it is not clear which exact event is relevant or not. 
\\

The only method for automatically constructing imperfect-recall abstractions with qualitative bounds is presented in \cite{kroer2016imperfect}. This existing method considers only a very restricted class of imperfect-recall abstractions. In short, the information sets can be merged only if they satisfy strict properties on the history of actions and there is a mapping between the applicable actions in these information sets such that future courses of the game and possible rewards are similar (see \cite{kroer2016imperfect} for all the details).
In this work, we take a different approach and instead of constraining which information sets can be merged, we design algorithms that start with a very coarse abstraction and similarly to \emph{inflation} operation~\cite{dalkey1953equivalence} refine information sets where necessary.
Our approach, however, does not require any specific structure of the abstract game or refined information sets.
We introduce two domain-independent algorithms, which can start with an arbitrary imperfect recall abstraction of the solved two-player zero-sum perfect recall EFG. The algorithms simultaneously solve the abstracted game, detect the missing information causing problems and refine the abstraction to include this information. This process is repeated until provable convergence to the desired approximation of the Nash equilibrium of the \textbf{original game}. 

Our algorithms can be initialized by an arbitrary abstraction since the choice of the initial abstraction does not affect the convergence guarantees of the algorithms. Hence, for example, in poker, we can use the existing state-of-the-art abstractions used by the top poker bots. Even though these abstractions have no guarantees that they allow solving the original poker to optimality, our algorithms will further refine these abstractions where necessary and provide the desired approximation of the Nash equilibrium in the original game. If there is no suitable abstraction available for the solved game, the algorithms can start with a simple coarse imperfect recall abstraction (we provide a domain-independent algorithm for constructing such abstraction) and again update the abstraction until it allows approximation of the Nash equilibrium of the original game to the desired precision.   
}

\subsection{Outline of the Proposed Algorithms}

The first algorithm is Fictitious Play for Imperfect Recall Abstractions (\textsc{FPIRA})\footnote{A part of this work appeared in \cite{cermak2017ijcai}. Here, we provide an improved version of the FPIRA which can use a significantly smaller initial abstraction. Additionally, we significantly extend the experimental evaluation of the algorithm.}. FPIRA is based on Fictitious Play (FP, \cite{brown1949some}). As a part of the contribution, we discuss the problems of applying FP to imperfect recall abstractions and how to resolve them. We then demonstrate how to detect the parts of the abstraction that need to be refined to enable convergence to the Nash equilibrium of the original game. We base this detection on the difference between the quality of the strategies expected from running FP directly on the original two-player zero-sum EFG with perfect recall and the result obtained from applying it to the abstraction. Finally, we prove that the guarantee of convergence of FP to the Nash equilibrium of the original two-player zero-sum EFG with perfect recall directly translates to the guarantee of convergence of FPIRA to the Nash equilibrium of this game. 

The second algorithm is CFR+ for Imperfect Recall Abstractions (\IRCFR). 
\IRCFR replaces FP by the CFR+ algorithm \cite{tammelin2014cfr+} since CFR+ is known to have a significantly faster empirical convergence to the Nash equilibrium. As a part of the contribution, we describe problems of applying CFR+ directly to imperfect recall abstractions and how to resolve them. To update the abstraction, we compare the expected theoretical convergence of CFR+ in the original game and the convergence achieved in the abstraction. The abstraction is refined when the observed convergence is slower than the theoretical guarantee provided by CFR+ in the original game with perfect recall. We provide a bound on the average external regret of \IRCFR and hence show that \IRCFR is guaranteed to converge to the Nash equilibrium of the original two-player zero-sum EFG with perfect recall. Finally, we provide an efficient heuristic for the abstraction update and demonstrate that it significantly improves the speed of convergence to the Nash equilibrium of the original EFG.

Both algorithms are conceptually similar to the Double Oracle algorithm (DOEFG, \cite{bosansky2014}) since they create a smaller version of the original game and repeatedly refine it until the desired approximation of the Nash equilibrium of the original game is found. Our algorithms, however, use imperfect recall information abstractions during the computation, while DOEFG uses a restricted perfect recall game, where the players are allowed to play only a subset of their actions. Hence, the algorithms introduced in this article exploit a completely different type of sparseness than DOEFG. 

In the experimental evaluation, we compare the memory requirements and runtime of \IRCFR, FPIRA, and DOEFG. We demonstrate that \IRCFR requires at least an order of magnitude less memory than DOEFG and FPIRA to solve a diverse set of domains. Hence it is the most suitable algorithm for (approximately) solving games with limiting memory requirements. We show that even if \IRCFR is initialized with a trivial automatically built abstraction, it requires building of information abstractions with as few as $0.9\%$ of information sets of the original game to find the desired approximation of the Nash equilibrium of the original game. Moreover, the results suggest that the relative size of the abstraction built by \IRCFR will further decrease as the size of the solved game increases.
{From the runtime perspective, we demonstrate that the \IRCFR may converge similarly fast to CFR+ applied directly to the original game.}

\subsection{Related Work on Abstractions}
Here, we provide an overview of the related work concerning the use of abstractions to solve large EFGs. 
There are two distinct approaches to abstracting an EFG: \emph{information abstractions} and \emph{action abstractions}.

The information abstractions reduce the size of the original large extensive-form game by removing information available to players; hence, merging their information sets.
Since the players have to play an identical strategy in the merged information sets, the size of the strategy representation in the abstracted game can be significantly smaller than in the original game.
The abstracted game is then solved, and the small resulting strategies are used in the original game.
Most of the work on information abstractions was driven by research in the poker domain.
Information abstractions were initially created using domain-dependent knowledge \cite{shi2000abstraction,billings2003approximating}. Algorithms for creating information abstractions followed \cite{gilpin2006competitive,Gilpin07:Abstraction} and led to the development of bots with increasing quality of play in poker.
This work culminated in lossless information abstractions (i.e., abstractions where the strategies obtained by solving the abstracted game form optimal strategies in the original game) which allowed solving the Rhode Island Hold'em, a poker game with $3.1\cdot 10^9$ nodes in the game tree \cite{gilpin2007}{, after reducing the number of sequences in the sequence form representation of the game to approximately 1.4\% of their original number.}
When moving to larger games, lossless abstractions were found to be too restrictive to offer sufficient reductions in the size of the abstracted game. Hence the focus switched to lossy abstractions. A mathematical framework that can be used to create perfect recall information abstractions with bounds on solution quality was introduced \cite{kroer2014extensive}. Additionally, both lossless and lossy imperfect recall abstractions were provided in \cite{lanctot2012,kroer2016imperfect}. The authors show that running the CFR algorithm on this class of imperfect recall abstractions leads to a bounded regret in the full game. 
Even though these restrictions simplify solving of the abstracted game, they prevent us from creating sufficiently small and useful abstracted games and thus fully exploit the possibilities of imperfect recall. Existing methods for using imperfect recall abstractions without severe limitations cannot provide any guarantees of the quality of computed strategies \cite{waugh2009}, or assume that the abstraction is given and require computationally complex approximation of a bilinear program to solve it \cite{cermak2018approximating}.

Another form of abstractions are action abstractions. Instead of merging information sets, this abstraction methodology restricts the actions available to the players in the original game. Similarly to information abstractions, most of the work on action abstraction was driven by the poker domain where there is a prohibitive amount of actions available to the players (e.g., when betting the players can choose any value up to the number of chips they have available). The first automated techniques iteratively adjust the bet sizing in no-limit hold'em \cite{hawkin2011automated,hawkin2012using}. The algorithm for simultaneous action-abstraction finding and game-solving was introduced in \cite{brown2015simultaneous}. This approach is similar to the algorithms presented in this work since it starts with a coarse abstraction and iteratively refines it until guaranteed convergence to the Nash equilibrium of the solved game.

The combination of information abstractions and action abstractions was at the heart of a number of algorithms for playing heads-up no-limit Texas Hold'em poker (see, e.g., \cite{gilpin2008heads}). Furthermore, the algorithm Libratus \cite{brown2017superhuman} using both mentioned abstractions led to a superhuman performance in heads-up no-limit Texas Hold'em.

In other domains, abstractions with bounded error were introduced in Patrolling Security Games \cite{basilico2011automated}. Additionally, the first general framework for lossy game abstraction with bounds was provided for stochastic games \cite{sandholm2012lossy}.

\section{Extensive-Form Games}
A two-player extensive-form game (EFG) is a tuple $G=(\mathcal{N},\mathcal{H},\mathcal{Z},\mathcal{A},u,{\cal C},{\cal I})$, which is commonly visualized as a game tree (see Figure~\ref{fig:efg_example}).

\begin{figure}[t]
\centering
\includegraphics[height=9em]{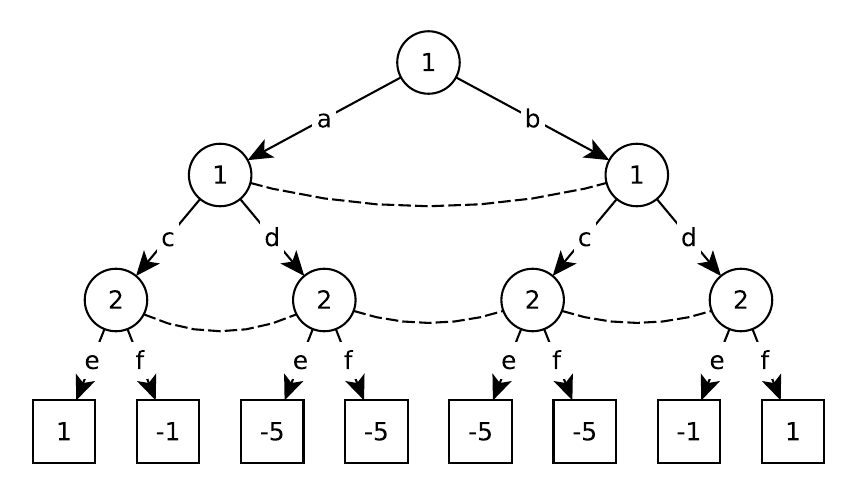}
\caption{An imperfect recall game. Circle nodes represent the states of the game, labels in the circles show which player acts in that node (player 1, player 2 or chance player $N$), dashed lines represent indistinguishable states and box nodes are the terminal states with utility value for player 1 (the game is zero-sum, hence player 1 maximizes the utility, player 2 minimizes it). {Imperfect recall is caused by player 1 forgetting their first action.}}\label{fig:efg_example}
\end{figure}

$\mathcal{N} = \{1, 2\}$ is a set of players, by $i$ we refer to one of the players, and by $-i$ to his opponent. Additionally, the chance player $N$ represents the stochastic environment of the game.
$\mathcal{A}$ denotes the set of all actions labeling the edges of the game tree.
$\mathcal{H}$ is a finite set of \emph{histories} of actions taken by all players and the chance player from the root of the game. 
Each history corresponds to a \emph{node} in the game tree; hence, we use the terms history and node interchangeably.
$\mathcal{Z} \subseteq \mathcal{H}$  is the set of all \emph{terminal states} of the game corresponding to the leaves of the game tree. 
For each $z \in \mathcal{Z}$ and $i \in \mathcal{N}$ we define a \emph{utility function} $u_i: \mathcal{Z} \rightarrow \mathbb{R}$. If $u_i(z) = -u_{-i}(z)$ for all $z \in \mathcal{Z}$, we say that the game is zero-sum. Chance player selects actions based on a fixed probability distribution known to all players. 
Function ${\cal C} : \mathcal{H} \rightarrow [0,1]$ is the probability of reaching $h$ obtained as the product of probabilities of actions of chance player preceding $h$. We further overload ${\cal C}$ and use it to denote the probability $\mathcal{C}(a)$ that action $a$ of chance player is taken.
Imperfect observation of player $i$ is modeled via \emph{information sets} ${\cal I}_i$ that form a partition over $h \in \mathcal{H}$ where $i$ takes action.
Player~$i$ cannot distinguish between nodes in any $I \in \mathcal{I}_i$. We represent the information sets as nodes connected by dashed lines in the examples.
$\mathcal{A}(I)$ denotes actions available in each $h \in I$.
The action $a$ uniquely identifies the information set where it is available, i.e., for all distinct $I, I' \in \mathcal{I}\ \forall a \in \mathcal{A}(I)\ \forall a' \in \mathcal{A}(I')\ a \neq a'$.
An ordered list of all actions of player~$i$ from the root to node $h$ is referred to as a \emph{sequence}, $\sigma_i = \seq{i}{h}$. $\Sigma_i$ is a set of all sequences of player~$i$. 
We use $\seq{i}{I}$ as a set of all sequences of player~$i$ leading to $I$.
A game has \emph{perfect recall} iff $\forall i \in {\cal N}\ \forall I \in {\cal I}_i$, for all $h, h' \in I$ holds that $seq_i(h) = seq_i(h')$. If there exists at least one information set where this does not hold (denoted as \emph{imperfect recall information set}), the game has \emph{imperfect recall}.

\subsection{Information Abstraction}
By \emph{information abstraction} of a game $G = ({\cal{N}}, {\cal H}, {\cal Z}, u, {\cal I}, \mathcal{A})$ we denote a game $G^x = ({\cal{N}}, {\cal H}, {\cal Z}, u, {\cal I}^x, \mathcal{A}^x)$ (when iteratively building the abstraction, we willl use $x$ to refer to the abstraction in a specific iteration). $G^x$ differs from $G$ in the structure of information sets and hence also in the action labeling. Each $I \in \mathcal{I}^x$ groups one or more information sets in $\mathcal{I}$. $\mathcal{A}^x$ changes $\mathcal{A}$ so that each action $a$ again uniquely identifies the information set in $\mathcal{I}^x$ where it is available. Furthermore, $\forall I \in \mathcal{I}^x \forall h, h' \in I \ \mathcal{A}^x(h) = \mathcal{A}^x(h')$.

To formally describe the information abstraction, we define mappings $\Phi_x: \mathcal{I} \rightarrow \mathcal{I}^x$, which for each $I \in \mathcal{I}$ returns the information set containing $I$ in $G^x$ and $\Phi_x^{-1}: \mathcal{I}^x \rightarrow \wp(\mathcal{I})$, the inverse of $\Phi_x$. By $\Xi_x : \mathcal{A} \rightarrow \mathcal{A}^x$ and $\Xi^{-1}_x : \mathcal{A}^x \rightarrow \wp(\mathcal{A})$ we denote the mapping of actions from $G$ to $G^x$ and vice versa.

We say that $I \in \mathcal{I}^x$ is an \emph{abstracted information set} if $|\Phi_x^{-1}(I)| > 1$. By $\widetilde{\mathcal{I}}^x \subseteq \mathcal{I}^x$ we denote the set of all abstracted information sets in $G^x$.

\begin{example}
\begin{figure}
\centering
\includegraphics[width = 3.5cm]{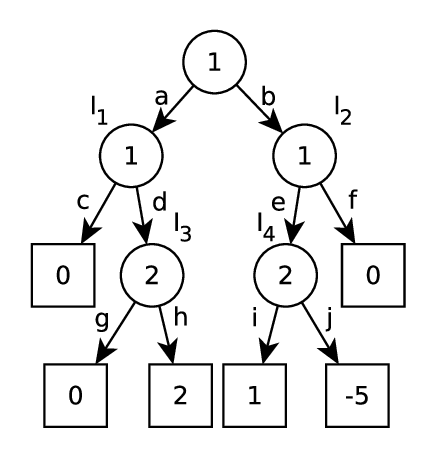}\includegraphics[width = 3.5cm]{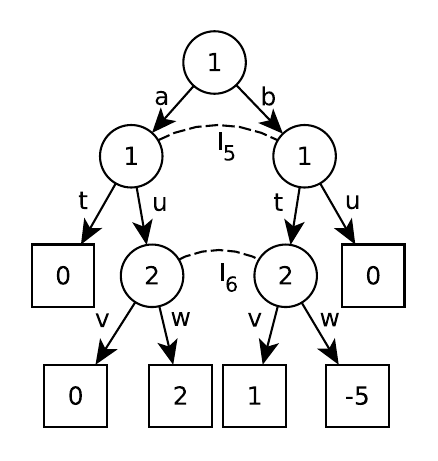}
\caption{(a) Extensive-form game $G$. (b) Imperfect recall information abstraction of $G$.}
\label{fig:abstr_example}
\end{figure}
In Figure \ref{fig:abstr_example}(a) we show an extensive-form game $G$ and in Figure \ref{fig:abstr_example}(b) its imperfect recall abstraction $G^x$. In this case, 

{\small\begin{gather*}\Phi_x: \{I_1 \rightarrow I_5, I_2 \rightarrow I_5, I_3 \rightarrow I_6, I_4 \rightarrow I_6\},\\ 
\Phi^{-1}_x: \{I_5 \rightarrow \{I_1, I_2\}, I_6 \rightarrow \{I_3, I_4\}\},\\ 
\Xi_x: \{c \rightarrow t, e \rightarrow t, d \rightarrow u, f \rightarrow u, g \rightarrow v, i \rightarrow v, h \rightarrow w, j \rightarrow w\},\\
\Xi^{-1}_x: \{t \rightarrow \{c, e\}, u \rightarrow \{d, f\}, v \rightarrow \{g, i\}, w \rightarrow \{h, j\}\},\\
\widetilde{\mathcal{I}}^x = \{I_5, I_6\}.
\end{gather*}}
   
\end{example}

\subsection{Strategies in Imperfect Recall Games}
There are several representations of strategies in EFGs.
A \emph{pure strategy} $s_i$ for player $i$ is a mapping assigning $\forall I \in {\cal I}_i$ an element of $\mathcal{A}(I)$. ${\cal S}_i$ is a set of all pure strategies for player $i$. A \emph{mixed strategy} $m_i$ is a probability distribution over ${\cal S}_i$. The set of all mixed strategies of $i$ is denoted as ${\cal M}_i$. %
\emph{Behavioral strategy} $b_i$ assigns a probability distribution over $\mathcal{A}(I)$ for each $I \in \mathcal{I}_i$. ${\cal B}_i$ is a set of all behavioral strategies for $i$, ${\cal B}^p_i \subseteq {\cal B}_i$ is the set of deterministic behavioral strategies for $i$. 
A \emph{strategy profile} is a set of strategies, one strategy for each player.  

\begin{definition}
A pair of strategies $x_i, y_i$ of player $i$ with arbitrary representation is \emph{realization equivalent} if $\forall z \in \mathcal{Z}: \pi_i^{x_i}(z) = \pi_i^{y_{i}}(z)$, where $\pi_i^{x_{i}}(z)$ is a probability that $z$ is reached due to strategy $x_i$ of player $i$ when the rest of the players play to reach $z$.
\end{definition}
We overload the notation and use $u_i$ as the expected utility of $i$ when the players play according to pure (mixed, behavioral) strategies.

Behavioral strategies and mixed strategies have the same expressive power in perfect recall games, but it can differ in imperfect recall games~\cite{Kuhn1953} (see \cite{cermak2018approximating} for more detailed discussion). 

Moreover, the size of these representations differs significantly. Mixed strategies of player $i$ state probability distribution over $\mathcal{S}_i$, where $|\mathcal{S}_i| \in \mathcal{O}(2^{|\mathcal{Z}|})$, behavioral strategies create probability distribution over the set of actions (therefore, its size is proportional to the number of information sets, which can be exponentially smaller than $|\mathcal{Z}|$). Hence, behavioral strategies are more memory efficient strategy representation. Additionally, when used in combination with information abstractions, behavioral strategies directly benefit from the reduced number of information sets in the abstracted game.

Finally, we define the Nash equilibrium, $\epsilon$-Nash equilibrium and the exploitability of a strategy.
\begin{definition}
\label{def:NE_behav}
We say that strategy profile $b = \{b^*_i, b^*_{-i}\}$ is a \emph{Nash equilibrium} (NE) in behavioral strategies iff $\forall i \in \mathcal{N}\ \forall b'_i \in {\cal B}_i: u_i(b^*_i, b^*_{-i}) \geq u_i(b'_i, b^*_{-i})$.
\end{definition}
Informally, a strategy profile is a NE if and only if no player wants to deviate to a different strategy.

\begin{definition}
\label{def:eNE_behav}
We say that strategy profile $b = \{b^*_i, b^*_{-i}\}$ is an $\epsilon$\emph{-Nash equilibrium} ($\epsilon$-NE) in behavioral strategies iff $\forall i \in \mathcal{N}\ \forall b'_i \in {\cal B}_i: u_i(b^*_i, b^*_{-i}) \geq u_i(b'_i, b^*_{-i}) - \epsilon$.
\end{definition}
Informally, a strategy profile is a $\epsilon$-NE if and only if no player can gain more than $\epsilon$ by deviating to a different strategy.

\begin{definition}
We define the \emph{exploitability} of a strategy $b_i$ as 

{\small
$$\max_{b'_i \in \mathcal{B}_i} \min_{b'_{-i} \in \mathcal{B}_{-i}} u_i(b'_i, b'_{-i}) - \min_{b'_{-i} \in \mathcal{B}_{-i}}u_i(b_i, b'_{-i}).$$}
\end{definition}
Informally, the exploitability of a strategy of player $i$ corresponds to the highest loss the player $i$ can suffer for not playing the strategy maximizing his worst case expected outcome.

\section{Iterative Algorithms for Solving EFGs}
\label{sec:cfr}
In this section we describe Fictitious Play (FP,  \cite{brown1949some}), we follow with the description of ideas behind external regret, Counterfactual Regret Minimization (CFR, \cite{zinkevich2008regret}) and its variant CFR+~\cite{tammelin2014cfr+,tammelin2015solving}.

\subsection{Fictitious Play}
Fictitious play (FP) is an iterative algorithm originaly defined for normal-form games \cite{brown1949some}. It keeps track of average mixed strategies $\bar{m}^T_i, \bar{m}^T_{-i}$ of both players. Players take turn updating their average strategy as follows. In iteration $T$, player $i$ computes $s^T_i \in BR(\bar{m}^{T-1}_{-i})$. He then updates his average strategy $\bar{m}^T_i = \frac{T_i - 1}{T_i}\bar{m}^{T - 1}_i  + \frac{1}{T_i}s^T_i$ ($T_i$ is the number of updates performed by $i$ plus 1). In two-player zero-sum games $\bar{m}^T_i, \bar{m}^T_{-i}$ converge to a NE \cite{robinson1951iterative}. Furthermore, there is a long-standing conjecture \cite{Karlin2003mathematical,Daskalakis2014counter} that the convergence rate of FP is $O(T^{-\frac{1}{2}})$, the same order as the convergence rate of CFR (though the empirical convergence of CFR and CFR+ tends to be better).

When applying FP to behavioral strategies in perfect recall zero-sum EFG $G$, one must compute the average behavioral strategy $\bar{b}^t_i$ such that it is realization equivalent to $\bar{m}^t_i$ obtained when solving the normal form game corresponding to $G$ for all $t$ and all $i \in \mathcal{N}$ to keep the convergence guarantees.
To update the behavioral strategy in such a way we use the following Lemma \cite{heinrich2015fictitious}.
\begin{lemma}
Let $b_i$, $b_i'$ be two behavioral strategies and $m_i$, $m_i'$ two mixed strategies realization equivalent to $b_i$, $b_i'$, and $\lambda_1, \lambda_2 \in [0, 1]$, $\lambda_1 + \lambda_2 = 1$. Then 
{\small
$$b_i''(I) = b_i(I) + \frac{\lambda_2 \pi_i^{b_i'}(I)}{\lambda_1 \pi_i^{b_i}(I) + \lambda_2 \pi_i^{b_i'}(I)}(b_i'(I) - b_i(I)), \forall I \in \mathcal{I}_i, $$}
defines a behavioral strategy $b_i''$ realization equivalent to the mixed strategy $m_i'' = \lambda_1 m_i + \lambda_2 m_i'$.\label{lemma:strat_update}
\end{lemma}

\subsection{External Regret} 
Given a sequence of behavioral strategy profiles $b^1, ..., b^T$, the external regret for player $i$, defined as

{\small
\begin{equation}
R_i^T = \max_{b_i' \in {\cal B}_i}\sum_{t = 1}^T(u_i(b_i', b_{-i}^t) - u_i(b^t_i, b_{-i}^t)), \label{eq:ext_regret}
\end{equation}
}
is the amount of additional expected utility player $i$ could have gained if he played the best possible strategy across all time steps $t \in \{1, ..., T\}$ compared to the expected utility he got from playing $b_i^t$ in every $t$.
An algorithm is a no-regret algorithm for player $i$, if the average external regret approaches zero; i.e., 

{\small $$\lim_{T \rightarrow \infty}\frac{R_i^{T}}{T}= 0.$$}

\subsection{Counterfactual Regret Minimization}
Let $b^t_{I \rightarrow a}$ be the strategy profile $b^t$ except for $I$, where $a \in \mathcal{A}(I)$ is played. Let $\pi^b(h)$ be the probability that $h$ will be reached when players play according to the strategy profile $b$, with $\pi^b_i(h) = \prod_{a \in seq_i(h)}b(a)$  being the contribution of player $i$ and $\pi^b_{-i, c}(h) = \mathcal{C}(h)\prod_{a \in seq_{-i}(h)}b(a)$ the contribution of $-i$ and chance. $z[I]$ stands for the state $h$ which is the predecessor of $z$ in $I$. Let $\pi^b(h, h')$ be the probability that $h'$ will be reached from $h$ when players play according to $b$ and $\mathcal{Z}_I \subseteq \mathcal{Z}$ a set of leaves reachable from all $h \in I$.
Finally, let the counterfactual value of $i$ in information set $I$ when players play according to the strategy profile $b$ be

{\small
\begin{equation}
v_i(b, I) = \sum_{z \in \mathcal{Z}_I}u_i(z)\pi^b_{-i, c}(z[I])\pi^b(z[I], z).\nonumber
\end{equation}  
}
Counterfactual regret is defined for each iteration $T$, player $i$, information set $I \in \mathcal{I}_i$ and action $a \in \mathcal{A}(I)$ as 
\begin{equation}
R_i^T(I, a) = \sum_{t = 1}^T\left[ v_i(b^t_{I \rightarrow a}, I) - v_i(b^t, I)\right].\label{eq:RIA_update}
\end{equation}

Let  $\left(x\right)^+$ stand for $\max(x, 0)$. The strategy $b^T_i$ for player $i$ in iteration $T$ in the standard CFR algorithm (sometimes termed vanilla CFR) is computed from counterfactual regrets using the \emph{regret-matching} as follows

{\small
\begin{equation}
b^{T}_i(I, a) = \begin{cases}
\frac{\left(R^{T-1}_i(I, a)\right)^+}{\sum_{a'\in \mathcal{A}(I)} \left( R_i^{T-1}(I, a')\right)^+} \text{, if }\sum_{a'\in \mathcal{A}(I)} \left(R_i^{T-1}(I, a')\right)^+ > 0\\
\frac{1}{\mathcal{A}(I)}\text{, otherwise.} \label{eq:regret-matching}
\end{cases}
\end{equation}}

The immediate counterfactual regret is defined as 
{\small
$$R^T_{i, imm}(I) = \frac{1}{T}\max_{a \in A(I)}R_i^T(I, a).$$}
In games having perfect recall, minimizing the immediate counterfactual regret in every information set minimizes the average external
regret. This holds because perfect recall implies that 

{\small
\begin{equation}
\frac{R_i^T}{T} \leq\sum_{I\in {\cal I}_i}\left(R_{i, imm}^{T}(I)\right)^+,\label{eq:bound1}
\end{equation}
}

i.e., the external regret is bounded by the sum of positive parts of immediate counterfactual regrets \cite{zinkevich2008regret}.

Let $\Delta_I = \max_{z \in \mathcal{Z}_I}u_i(z) - \min_{z \in \mathcal{Z}_I}u_i(z)$ and $u_{max} = \max_{i \in \mathcal{N}}\max_{z \in \mathcal{Z}} u_i(z)$.
When player $i$ plays according to eq. \eqref{eq:regret-matching} in $I$ during iterations \{1, ..., T\}, then

{\small
\begin{equation}
R_{i, imm}^T(I)    \leq \frac{\Delta_I\sqrt{|\mathcal{A}(I)|}}{\sqrt{T}}.\label{eq:cfr_is_regret_bound}
\end{equation}}

Thus, from eq. \eqref{eq:bound1}
{\small
\begin{equation}
\frac{R_i^T}{T} \leq \frac{\Delta|{\cal I}_i|\sqrt{|A_{max}|}}{\sqrt{T}},\label{eq:bound2}
\end{equation}
} where $|A_{max}| = \max_{I\in {\cal I}_i}|\mathcal{A}(I)|$ and $\Delta = 2u_{max}$ \cite{zinkevich2008regret}.

Furthermore, let $\bar{b}_i^T$ be the average strategy for $i$ defined as 

{\small
\begin{equation}
\bar{b}_i^T(I, a) = \frac{\sum_{t = 1}^T\pi_i^{b_i^t}(I)b_i^t(I, a)}{\sum_{t = 1}^T \pi_i^{b_i^t}(I)}, \quad\forall I \in \mathcal{I}_i, \forall a \in \mathcal{A}(I).
\end{equation}
}

If $\frac{R^T_i}{T} < \epsilon$ for all $i \in \mathcal{N}$ in a two-player zero-sum EFG, the strategy profile $\bar{b}^T = (\bar{b}^T_1, \bar{b}^T_2)$ forms a $2\epsilon$-Nash equilibrium~\cite{zinkevich2008regret}.

\subsection{CFR+}
\label{sec:cfr+}
The CFR+~\cite{tammelin2014cfr+,tammelin2015solving} replaces the regret-matching shown in equation \eqref{eq:regret-matching} with \emph{regret-matching}$^+$. To do that, the algorithm maintains alternative regret values $Q_i^T(I, a)$ in each iteration $T$ and for each player $i$ defined as

{\small
\begin{align}
Q_i^0(I, a) &= 0\\
Q_i^T(I, a) &= \left( Q_i^{T-1}(I, a) + R_i^T(I, a) - R_i^{T-1}(I, a)\right)^+.
\end{align}}

The update of $Q_i^T(I, a)$ is identical to the update of $R_i^T(I, a)$ in eq. \eqref{eq:RIA_update} except for negative values which are replaced by 0. This prevents negative counterfactual regrets from accumulating. Since the strategy is then computed using the regret-matching$^+$ defined as
\begin{equation}
b^{T}_i(I, a) = \begin{cases}
\frac{Q_i^{T-1}(I, a)}{\sum_{a'\in \mathcal{A}(I)} Q_i^{T-1}(I, a')}\text{, if }\sum_{a'\in \mathcal{A}(I)} Q_i^{T-1}(I, a') > 0\\
\frac{1}{\mathcal{A}(I)}\text{, otherwise,} 
\end{cases}
\end{equation}
any future positive regret changes will immediatelly affect the resulting strategy instead of canceling out with the accumulated negative regret values.

Additionally, CFR+ uses alternating updates of the regrets, i.e., during one iteration of the algorithm only regrets of one player are updated, and the players take turn.  

Finally, as described in \cite{tammelin2015solving}, the update of the average strategy starts only after a fixed number of iterations denoted as $d$. 
The update is then weighted by the current iteration, formally for each $T > d$

{\small
\begin{equation}
\bar{b}_i^T(I, a) = \frac{2\sum_{t = d}^Tt \cdot\pi_i^{b_i^t}(I)\cdot b_i^t(I, a)}{\left((T - d)^2 + T - d\right)\sum_{t = d}^T \pi_i^{b_i^t}(I)}, \quad\forall I \in \mathcal{I}_i, \forall a \in \mathcal{A}(I).\label{eq:cfr+stratupdate}
\end{equation}
}

Intuitively, the average strategy update puts a higher weight on later strategies, as they are expected to perform better.

The bound on average external regret $\frac{R^T_i}{T}$ from eq. \eqref{eq:bound2} also holds when using regret matching$^+$. Hence the average strategy profile computed by CFR+ converges to the Nash equilibrium of the solved two-player zero-sum EFG~\cite{tammelin2015solving}. Additionally, the empirical convergence of CFR+ is significantly better compared to CFR~\cite{tammelin2015solving}.
\section{Algorithms for Constructing and Solving Imperfect Recall Abstractions}

In this section, we present the main algorithmic results of this paper.
We first discuss the initial imperfect recall abstraction of a given two-player zero-sum EFG $G$ which forms a starting point of the algorithms. Note that in this section we focus on the scenario where no initial abstraction is given, and the algorithms need to build the initial coarse imperfect recall abstraction automatically. We follow with the description of the two algorithms which iteratively solve and refine this abstraction until they reach the desired approximation of the Nash equilibrium of the original unabstracted game $G$. In Section \ref{sec:FPIRA} we present the FP based approach denoted as Fictitious Play for Imperfect Recall Abstractions (FPIRA). In Section~\ref{sec:IRCFR+} we show the approach using a modification of CFR+ to iteratively solve and refine this abstraction. We denote the algorithm CFR+ for Imperfect Recall Abstractions (\IRCFR). We provide proofs of convergence of both algorithms to the NE of the original unabstracted game and discussion of memory requirements and runtime performance of both algorithms.

\subsection{Abstraction}
As discussed before, the algorithms presented in this section can start from an arbitrary initial imperfect recall abstraction. In Section \ref{sec:init_abstr}, we demonstrate how to create a coarse imperfect recall abstraction which serves as a starting point of the algorithms if no initial abstraction is given.

We refer to the initial abstraction of the solved two-player zero-sum EFG $G$ as $G^1 = ({\cal{N}}, {\cal H}, {\cal Z}, u, {\cal I}^1, \mathcal{A}^1)$. In every iteration $t$, the algorithms operate with possibly more refined abstraction with respect to $G^1$, denoted as $G^t =  ({\cal{N}}, {\cal H}, {\cal Z}, u, {\cal I}^t, \mathcal{A}^t)$.

\subsubsection{Initial Abstraction}
\label{sec:init_abstr}
Given a game $G$, the initial abstraction $G^1$ is created in the following way:
Each $I \in \mathcal{I}_i^1$ is formed as the largest set of information sets $\mathcal{I}_I$ of $G$, so that $\forall I', I'' \in  \mathcal{I}_I\ |seq_i(I')| = |seq_i(I'')| \wedge |\mathcal{A}(I')| = |\mathcal{A}(I'')|$. Furthermore, $\bigcup_{I \in \mathcal{I}^1} \mathcal{I}_I = \mathcal{I}$.  
 Informally, for each $i \in \mathcal{N}$, the algorithm groups together information sets of $i$ with the same length of the sequence of $i$ leading there and with the same number of actions available.
Additionally, when creating some abstracted information set $I$ by grouping all information sets in $\mathcal{I}_I$, we need to specify the mapping $\Xi^{-1}_1(a)$ for all $a \in \mathcal{A}^1$. When creating the initial abstraction, the algorithm uses the order of actions given by the domain description to create $\mathcal{A}^1$ (i.e., the first action available in each $I' \in \mathcal{I}_I$ is mapped to the first action in $\mathcal{A}^1(I)$, etc.).

\subsection{The FPIRA Algorithm}\label{sec:FPIRA}
Let us now describe Fictitious Play for Imperfect Recall Abstractions (\textsc{FPIRA}). We first give a high-level idea behind \textsc{FPIRA}. Next, we provide the pseudocode with the description of all steps and prove its convergence in two-player zero-sum EFGs.  Finally, we discuss the memory requirements and runtime of FPIRA.

Given a perfect recall game $G$, \textsc{FPIRA} creates a coarse imperfect recall abstraction $G^1$ of $G$ as described in Section \ref{sec:init_abstr}. The algorithm then follows the FP procedure. It keeps track of average strategies of both players in the information set structure of the abstraction and updates the strategies in every iteration based on the best responses to the average strategies. Note that the best responses are computed directly in $G$ (see Section \ref{sec:experiments} for empirical evidence that these best responses are small), {hence it requires computation time proportional to the full unabstracted game.}
To ensure the convergence to Nash equilibrium of $G$, \textsc{FPIRA} refines the information set structure of the abstraction in every iteration to make sure that the strategy update does not lead to more exploitable average strategies in the following iterations compared to the strategy update made directly in $G$.
\begin{algorithm}[h]
\small
\DontPrintSemicolon
\SetKwInOut{Input}{input}
\SetKwInOut{Output}{output}
\SetKwInOut{Parameter}{parameters}
\SetKwProg{Function}{function}{}{}

\SetKwFunction{BuildAbstraction}{InitAbstraction}
\SetKwFunction{ActingPlayer}{ActingPlayer}
\SetKwFunction{UpdateStrategy}{UpdateStrategy}
\SetKwFunction{ComputeDelta}{ComputeDelta}
\SetKwFunction{RefineBR}{RefineForBR}
\SetKwFunction{Refine}{Refine}
\SetKwFunction{PureStrat}{PureStrat}
\SetKwFunction{BR}{BR}
\Input{$G$, $\epsilon$}
\Output{$\bar{b}_i^T$, $\bar{b}_{-i}^T$, $G^T$}
\BlankLine
$G^1 \gets$ \BuildAbstraction{$G$} \label{alg:irfp:abstr_init}\;
$\bar{b}^0_1 \gets $ \PureStrat{$G^1$}, $\bar{b}^0_2 \gets$ \PureStrat{$G^1$}\label{alg:irfp:strat_init}\;
$t \gets 1$\;
\While{$u_1($\BR{$G, \bar{b}^{t-1}_2$}$, \bar{b}^{t-1}_2) - u_1(\bar{b}^{t-1}_1,$ \BR{$G, \bar{b}^{t-1}_1$}$)  > \varepsilon$}{\label{alg:irfp:loop}
  $i \gets \ActingPlayer{t}$\;
  $b_i^t \gets$ \BR{$G, \bar{b}^{t-1}_{-i}$}\label{alg:irfp:br}\; 
  $G^t \gets $ \RefineBR{$G^t$, $b_i^t$, $\bar{b}_i^t$}\label{alg:irfp:br_refine}\;
  $\hat{b}^t_i \gets $ \UpdateStrategy{$G^{t}, \bar{b}^{t-1}_i, b_i^t$}\label{alg:irfp:hat_b1}\;
  $\widetilde{b}^t_i \gets $ \UpdateStrategy{$G, \bar{b}^{t-1}_i, b_i^t$}\label{alg:irfp:hat_b2}\;
  \If{\ComputeDelta{$G$, $\hat{b}^t_i$,  $\widetilde{b}^t_i$} $ > 0$\label{alg:irfp:delta_test}}{
       $G^{t+1} \gets $ \Refine{$G^{t}$}, $\bar{b}^t_i \gets \widetilde{b}^t_i$\label{alg:irfp:refine}\;
  }\Else{
      $G^{t+1} \gets G^{t}$, $\bar{b}^t_i \gets \hat{b}^t_i$\label{alg:irfp:update_only}\;
  }
  $t \gets t + 1$\;
}
\caption{\textsc{FPIRA} algorithm}
\label{alg:irfp}
\end{algorithm}

In Algorithm \ref{alg:irfp} we present the pseudocode of \textsc{FPIRA}. \textsc{FPIRA} is given the original perfect recall game $G = ({\cal{N}}, {\cal H}, {\cal Z}, u, {\cal I},  \mathcal{A})$ and a desired precision of NE approximation $\epsilon$. \textsc{FPIRA} first creates the coarse imperfect recall abstraction $G^1$  of $G$ (line \ref{alg:irfp:abstr_init}) as described in Section \ref{sec:init_abstr}. Next, it initializes the strategies of both players to an arbitrary pure strategy in $G^1$ (line \ref{alg:irfp:strat_init}). \textsc{FPIRA} then iterativelly solves and updates $G^1$ until convergence to $\epsilon-$Nash equilibrium of $G$. In every iteration it updates the average strategy of one of the players and if needed the information set structure of the abstraction (the game used in iteration $t$ is denoted as $G^t$). In every iteration player $i$ first computes the best response $b^t_i$ to $\bar{b}^{t-1}_{-i}$ in $G$ (line \ref{alg:irfp:br}, Section \ref{sec:br}).  Since $b^t_i$ is computed in $G$, \textsc{FPIRA} first needs to make sure that the structure of information sets in $G^t$ allows $b^t_i$ to be played. If not, $G^t$ is updated as described in Section \ref{subsec:G_update}, Case 1 (line \ref{alg:irfp:br_refine}). Next, \textsc{FPIRA} computes $\hat{b}^t_i$ as the strategy resulting from the update in $G^t$ (line \ref{alg:irfp:hat_b1}) and $\widetilde{b}^t_i$ as the strategy resulting from the update in original game $G$ (line \ref{alg:irfp:hat_b2}). \textsc{FPIRA} then checks whether the update in $G^t$ changes the expected values of the pure strategies of $-i$ compared to the update in $G$ using $\hat{b}^t_i$ and $\widetilde{b}^t_i$ (line \ref{alg:irfp:delta_test}, Section \ref{subsec:G_update} Case 2). If yes, \textsc{FPIRA} refines the information set structure of $G^t$, creating $G^{t+1}$ such that when updating the average strategies in $G^{t+1}$ no error in expected values of pure strategies of $-i$ is created (Section \ref{subsec:G_update}, Case 2). FPIRA then sets $\bar{b}^t_i= \widetilde{b}^t_i$  (line \ref{alg:irfp:refine}), and continues using $G^{t+1}$. If there is no need to update the structure of $G^t$, \textsc{FPIRA} sets $G^{t+1} = G^t, \bar{b}^t_i = \hat{b}^t_i$ and continues with the next iteration.

\subsubsection{Best Response Computation}\label{sec:br}
\begin{algorithm}[h!]
\small
\DontPrintSemicolon
\SetKwProg{Function}{function}{}{}
\SetKwFunction{IsGameEnd}{IsGameEnd}
\SetKwFunction{GetPlayerToMove}{GetPlayerToMove}
\SetKwFunction{GetInformationSetFor}{GetInformationSetFor}
\SetKwFunction{ComputePartialDelta}{ComputePartialDelta}
\SetKwFunction{GetValueFromCache}{GetValueFromCache}
\SetKwFunction{StoreToCache}{SaveToCache}
\SetKwFunction{BR}{BR}
\SetKwFunction{CleanUnreachable}{CleanUnreachable}

\Function{\BR{$G$, $\bar{b}_{-i}$}}{
 $b^{BR}_i \gets \emptyset$\;
 \BR{$G.root$, $G$,  $\bar{b}_{-i}$, $b^{BR}_i$}\;
 \Return{$b^{BR}_i$}\;
}
\BlankLine
\Function{\BR{$h$, $G$, $\bar{b}_{-i}$, $b^{BR}_i$}}{
\If{$\pi_i^{\bar{b}_{-i}}(h) = 0$}{\label{alg:br:opp_prunning_start}
\Return{$0$}\label{alg:br:opp_prunning_end}\;
}
\If{\IsGameEnd{$h$}}{
\Return{$\mathcal{C}(h)\cdot\pi_i^{\bar{b}_{-i}}(h)\cdot u_{i}(h)$}\;
}
$v \gets $ \GetValueFromCache{$h$}\label{alg:br:cache_return_start}\;
\If{$v \neq null$}{
\Return{$v$}\label{alg:br:cache_return_end}\;
}
\If{\GetPlayerToMove{$h$} $ = i$}{
$I \gets $ \GetInformationSetFor{$h, \mathcal{I}$}\label{alg:br:findIS}\;
$a_{max} \gets \argmax_{a \in \mathcal{A}(I)}\sum_{h' \in I}$ \BR{$h' \cdot a$, $G$, $\bar{b}_{-i}$, $b^{BR}_i$}\label{alg:br_redundancy}\;
For each $h' \in I$ store $v_{h'} = $ \BR{$h' \cdot a_{max}$, $G$, $\bar{b}_{-i}$, $b^{BR}_i$} to cache\label{alg:br:cache}\;
$v \gets$ \BR{$h\cdot a_{max}$, $G$, $\bar{b}_{-i}$, $b^{BR}_i$}\label{alg:br:v_storage}\;
$b^{BR}_i \gets b^{BR}_i  \cup \{I \rightarrow a_{max}\}$\label{alg:br:strategy_storage}\;
\CleanUnreachable{$b^{BR}_i$}\label{alg:br:strategy_cleanup}\;
}\Else{
$v \gets \sum_{a \in \mathcal{A}(h)}$\BR{$h\cdot a$, $G$, $\bar{b}_{-i}$, $b^{BR}_i$}\label{alg:br:rest}\;
}
\Return{$v$}\;
}
\caption{Best response computation.}
\label{alg:br}
\end{algorithm}

In Algorithm \ref{alg:br} we present the pseudocode for computing the best response $b^{BR}_i$ of $i$ against $\bar{b}_{-i}$ in $G$. The algorithm recursively traverses the parts of the game tree reachable by $\bar{b}_{-i}$ and computes the best action to be played in each $I \in \mathcal{I}_i$ encountered. More formally, when the best response computation reaches state $h$, where $i$ plays, it first finds the information set $I$ such that $h \in I$ (line \ref{alg:br:findIS}). The algorithm then finds the action $a_{max}$ which maximizes the sum of expected values of $i$, when $i$ plays the best response to $\bar{b}_{-i}$, over all $h \in I$ (line \ref{alg:br_redundancy}). The expected value of playing $a_{max}$ in $h$ is prepared to be propagated up (line \ref{alg:br:v_storage}) and the prescription of playing $a_{max}$ in $I$ is stored to the best response (line \ref{alg:br:strategy_storage}).
Notice, that to eliminate revisiting already traversed parts of the game tree on lines \ref{alg:br_redundancy} and \ref{alg:br:v_storage}, we use cache. The cache stores the values computed by the $\mathsf{BR}$ for each $h \in \mathcal{H}$ visited during the computation (line \ref{alg:br:cache}) and provides this value if $h$ is revisited (lines \ref{alg:br:cache_return_start} to \ref{alg:br:cache_return_end}). Finally, since we are searching for a pure best response $b^{BR}_i$, there is no need to store the behavior in the parts of the tree unreachable when playing $b^{BR}_i$. Hence to reduce the memory needed to store $b^{BR}_i$, we delete the prescription in all $I \in \mathcal{I}_i$ which cannot be reached due to player $i$ playing $a_{max}$ (line \ref{alg:br:strategy_cleanup}).

In case of nodes of player $-i$ and chance, the algorithm simply propagates up the values of successors (line \ref{alg:br:rest}).

\textbf{Pruning.}
The implementation of the best response used in FPIRA incorporates pruning based on the lower bound and upper bound in each node $h$ of the game tree. The lower bound represents the lowest value that needs to be achieved in $h$ to make sure that there is a chance that searching the subtree of $h$ influences the resulting best response and the upper bound represents the estimate of how much can the best responding player gain by visiting $h$. For a more detailed discussion of the pruning see \cite{bosansky2014}, Section 4.2. 

\subsubsection{Updating $G^t$}
\label{subsec:G_update}
There are two reasons for splitting some $I \in \widetilde{\mathcal{I}}^t_i$ in iteration $t$ where player $i$ computes the best response. First, the best response computed in $G$ prescribes more than one action in $I$. Second, $I$ causes the expected value of some pure strategy of $-i$ to be different against the average strategy of $i$ computed in $G^t$ compared to the expected value against the average strategy computed in $G$. This can happen since $I$ is an abstracted information set, and hence updating the average strategy in $I$ in $G^t$ changes the behavior in multiple information sets in $G$ (see Example \ref{ex:fpira} for more details). 
 
 \begin{algorithm}[h!]
\small
\DontPrintSemicolon
\SetKwProg{Function}{function}{}{}
\SetKwFunction{IsGameEnd}{IsGameEnd}
\SetKwFunction{GetPlayerToMove}{GetPlayerToMove}
\SetKwFunction{GetInformationSetFor}{GetInformationSetFor}
\SetKwFunction{ComputePartialDelta}{ComputePartialDelta}
\SetKwFunction{GetValueFromCache}{GetValueFromCache}
\SetKwFunction{StoreToCache}{StoreToCache}
\SetKwFunction{RefineForBR}{RefineForBR}
\SetKwFunction{CreateNewIS}{CreateNewIS}

\Function{\RefineForBR{$G^t$, $b^t_{i}$, $\bar{b}^t_{i}$}}{
 \For{$I \in \widetilde{\mathcal{I}}_i^t$}{\label{alg:refineBR:iter1}
     \If{$\pi_i^{b^t_{i}}(I) > 0$}{\label{alg:refineBR:iter2}
         $\mathcal{I}'' \gets \emptyset$\;
         \For{$a \in \mathcal{A}^t(I)$}{\label{alg:refineBR:starthatIa}
             $\hat{I}_a \gets \emptyset$\;
         }
         \For{$I' \in \phi_t^{-1}(I)$}{
             \If{$\exists a \in \mathcal{A}(I')$ such that $b^t_{i}(I, a) = 1$}{    
                 $a' \gets a \in \mathcal{A}(I')$ such that $b^t_{i}(I, a) = 1$\;
                $a \gets \Xi_t(a')$\;                 
                 $\hat{I}_a \gets \hat{I}_a \cup I'$\label{alg:refineBR:endhatIa}\;
             }\Else{
                 $\mathcal{I}'' \gets \mathcal{I}'' \cup I'$\label{alg:refineBR:I''}\;
             }
         }
         $\hat{\mathcal{I}} \gets \emptyset$\label{alg:refineBR:starthatI}\;
         \For{$a \in \mathcal{A}^t(I)$}{
             \If{$\hat{I}_a \neq \emptyset$}{
                 $\hat{\mathcal{I}} \gets \hat{\mathcal{I}} \cup \{\hat{I}_a\}$\label{alg:refineBR:endhatI}\;
             }
         }
         \If{$|\hat{\mathcal{I}}| > 1$}{\label{alg:refineBR:split_check}
      \For{$I' \in \hat{\mathcal{I}} \cup \mathcal{I}''$}{\label{alg:refineBR:strat_update_start}
        $\bar{b}^t_{i}(I') \gets \bar{b}^t_{i}(I)$\;
      }     
      $\bar{b}^t_{i} \gets \bar{b}^t_{i} \setminus  \bar{b}^t_{i}(I)$\label{alg:refineBR:strat_update_end}\;
             $\mathcal{I}^t \gets \mathcal{I}^t \setminus I$\label{alg:refineBR:deleteI}\;
             \For{$\hat{I}_a \in \hat{\mathcal{I}}$}{\label{alg:refineBR:addhatIstart}
        $\mathcal{I}^t \gets \mathcal{I}^t\ \cup $ \CreateNewIS{$\hat{I}_a$}\label{alg:refineBR:addhatIend}\;
      }
      \If{$\mathcal{I}'' \neq \emptyset$}{\label{alg:refineBR:addI''start}
                  $\mathcal{I}^t \gets \mathcal{I}^t\ \cup $ \CreateNewIS{$\mathcal{I}''$}\label{alg:refineBR:addI''end}\;
             }
         }        
     }
 }
 \Return{$G^t$}
}
\caption{Abstraction update to accommodate best response.}
\label{alg:refineBR}
\end{algorithm}
 
\vspace{0.3em}
\noindent\textbf{Case 1:}
Here we discuss the abstraction update which guarantees that the best response $b^t_i$ computed in $G$ is applicable in the resulting abstraction. 
This abstraction update first detects every $I \in \widetilde{\mathcal{I}}^t_i$ where the $b^t_i$ prescribes more than one action. It then splits each such $I$ by grouping the information sets in $\phi_t^{-1}(I)$ to the largest possible subsets where $b^t_i$ prescribes the same action. One additional information set is created containing all $I' \in \phi_t^{-1}(I)$ which are not reachable when playing $b^t_i$.

In Algorithm \ref{alg:refineBR} we show the pseudocode for this abstraction update. The algorithm iterates over all abstracted information sets $I \in \widetilde{\mathcal{I}}_i^t$ that can be visited when playing $b_i^t$ (lines \ref{alg:refineBR:iter1} and  \ref{alg:refineBR:iter2}). Each such $I$ is first divided to a sets of information sets $\hat{\mathcal{I}}$ and $\mathcal{I}''$. Let $\hat{I}_a \subseteq \Phi_t^{-1}(I)$ be a union of all $I' \in \Phi_t^{-1}(I)$  where $b^t_i(I', a') = 1$ for $\Xi_t(a') = a$ (lines \ref{alg:refineBR:starthatIa} to \ref{alg:refineBR:endhatIa}). $\hat{\mathcal{I}}$ is a union of $\hat{I}_a$ for all $a \in \mathcal{A}^t(I)$ (lines \ref{alg:refineBR:starthatI} to \ref{alg:refineBR:endhatI}). $\mathcal{I}''$ contains all $I' \in  \Phi_t^{-1}(I)$ for which $b_i^t$ does not prescribe any action (line \ref{alg:refineBR:I''}). If there is more than one element in $\hat{\mathcal{I}}$ (line \ref{alg:refineBR:split_check}), the average strategy before the strategy update in all information sets which are about to be created is set to the strategy previously played in $I$ (lines \ref{alg:refineBR:strat_update_start} to \ref{alg:refineBR:strat_update_end}). Next, the algorithm removes $I$ from $G^t$ (line \ref{alg:refineBR:deleteI}). Finally, it creates new informations set for each $\hat{I}_a \in \hat{\mathcal{I}}$, that contains all $I' \in \hat{I}_a$ and an information set containing all $I' \in \mathcal{I}''$ (lines \ref{alg:refineBR:addhatIstart} to \ref{alg:refineBR:addI''end}).

 \begin{algorithm}[h!]
\small
\DontPrintSemicolon
\SetKwProg{Function}{function}{}{}
\SetKwFunction{IsGameEnd}{IsGameEnd}
\SetKwFunction{GetPlayerToMove}{GetPlayerToMove}
\SetKwFunction{GetInformationSetFor}{GetInformationSetFor}
\SetKwFunction{ComputePartialDelta}{ComputePartialDelta}
\SetKwFunction{GetValueFromCache}{GetValueFromCache}
\SetKwFunction{StoreToCache}{StoreToCache}
\SetKwFunction{RefineForBR}{RefineForBR}
\SetKwFunction{CreateNewIS}{CreateNewIS}

\Function{\Refine{$G^t$, $b^t_{i}$, $\bar{b}^t_{i}$}}{
 \For{$I \in \widetilde{\mathcal{I}}_i^t$}{\label{alg:refine:iter1}
     \If{$\pi_i^{b^t_{i}}(I) > 0$}{\label{alg:refine:iter2}
         $\hat{\mathcal{I}} \gets \{I' \in \phi^{-1}_t(I) | \pi_i^{b^t_{i}}(I') > 0\}$\label{alg:refine:hatI}\;
         $\mathcal{I}'' \gets  \phi^{-1}_t(I) \setminus \hat{\mathcal{I}}$\label{alg:refine:I''}\;
         $\mathcal{I}^t \gets \mathcal{I}^t \setminus I$\;
    \For{$I' \in \hat{\mathcal{I}} \cup \mathcal{I}''$}{\label{alg:refine:strat_update_start}
      $\bar{b}^t_{i}(I') \gets \bar{b}^t_{i}(I)$\;
    }     
    $\bar{b}^t_{i} \gets \bar{b}^t_{i} \setminus  \bar{b}^t_{i}(I)$\label{alg:refine:strat_update_end}\;
         \If{$\mathcal{I}'' \neq \emptyset$}{\label{alg:refine:addI''start}
             $\mathcal{I}^t \gets \mathcal{I}^t\ \cup $ \CreateNewIS{$\mathcal{I}''$}\;
         }
         \For{$I' \in \hat{\mathcal{I}}$}{
             $\mathcal{I}^t \gets \mathcal{I}^t\ \cup $ \CreateNewIS{$I'$}\label{alg:refine:addhatIend}\;
         }
     }
 }
 \Return{$G^t$}\;
}
\caption{Abstraction update guaranteeing that the expected values of all pure strategies of $-i$ are the same against the strategy resulting from update in $G$ and $G^t$.}
\label{alg:refine}
\end{algorithm}

 \vspace{0.3em}
 
\noindent\textbf{Case 2:} 
Now we turn to the abstraction update guaranteeing that the expected values of all pure strategies of $-i$ against the average strategy of $i$ computed in $G^t$ are equal to the expected values against the average strategy computed in $G$. As a part of this abstraction update, FPIRA constructs the average strategy $\hat{b}^t_i$ resulting from the update in $G^t$ and $\widetilde{b}^t_i$ resulting from the update in $G$. It then checks whether there exists a pure strategy of player $-i$ which has different expected value against $\hat{b}^t_i$ and $\widetilde{b}^t_i$. If there exists such pure strategy, the abstraction is updated to guarantee that the average strategy update in the abstraction results in $\widetilde{b}^t_i$, and hence that there is no difference in the expected values of pure strategies of $-i$.

More formally, the algorithm first constructs the average behavioral strategy $\hat{b}^t_i$ in $G^t$ (line \ref{alg:irfp:hat_b1} in Algorithm \ref{alg:irfp}). This is done according to Lemma \ref{lemma:strat_update} from $\bar{b}^{t-1}_i$ with weight $\frac{t_i -1}{t_i}$ and  $b^t_i$ with weight $\frac{1}{t_i}$, where $t_i$ is the number of updates performed by $i$ so far, plus 1 for the initial strategy ($b^t_i$ is used with mappings $\Phi_t$ and $\Xi_t$). Next, \textsc{FPIRA} constructs $\widetilde{b}^t_i$ (line \ref{alg:irfp:hat_b2}) in the same way in the information set structure of $G$ ($\bar{b}^{t-1}_i$ is used with mappings $\Phi^{-1}_t$ and $\Xi^{-1}_t$).  \textsc{FPIRA} then computes 

{\small
$$
\Delta_i^t = \!\!\max_{b_{-i} \in \mathcal{B}^P_{-i}}\!\!| u_{-i}(\widetilde{b}^t_i, b_{-i}) - u_{-i}(\hat{b}^t_i, b_{-i})| ,
$$}

as described below (line \ref{alg:irfp:delta_test} in Algorithm \ref{alg:irfp}).
If $\Delta_i^t = 0$, all the pure strategies of $-i$ have the same expected value against both $\hat{b}^t_i$ and $\widetilde{b}^t_i$. In this case, \textsc{FPIRA} sets $G^{t+1} = G^{t}$, $\bar{b}_i^t = \hat{b}_i^t$ (line \ref{alg:irfp:update_only}). If $\Delta_i^t > 0$, the expected value of some pure strategy of $-i$ changed when updating the strategy in $G^t$, compared to the expected value it would get against the strategy updated in $G$.  \textsc{FPIRA} then creates $G^{t+1}$ according to Algorithm \ref{alg:refine} in the following way. 
Every abstracted information set $I \in \widetilde{\mathcal{I}}_i^t$ which is visited when playing $b_i^{t}$ (lines \ref{alg:refine:iter1} and \ref{alg:refine:iter2}) is first divided to sets of information sets $\hat{\mathcal{I}} \subseteq \Phi^{-1}_t(I)$ and $\mathcal{I}''$. $\hat{\mathcal{I}}$ contains all the $I' \in \Phi^{-1}_t(I)$  for which $\pi_i^{b_i^{t}}(I') > 0$ (line \ref{alg:refine:hatI}), $\mathcal{I}''$ contains the rest of $I' \in \Phi^{-1}_t(I)$ (line \ref{alg:refine:I''}). 
The average strategy before the strategy update in all information sets which are about to be created is set to the strategy previously played in $I$ (lines \ref{alg:refine:strat_update_start} to \ref{alg:refine:strat_update_end}).
The algorithm then creates new information set for each $\hat{I} \in \hat{\mathcal{I}}$, and an information set containing all $I' \in \mathcal{I}''$ (lines \ref{alg:refine:addI''start} to \ref{alg:refine:addhatIend}). 
The strategy resulting from update in $G$ is a valid strategy in $G^{t+1}$ after such update, hence $\bar{b}_i^t = \widetilde{b}_i^t$. 
Notice that by setting $\bar{b}_i^t = \widetilde{b}_i^t$, we made sure that $\Delta_i^t = 0$ since the update is now equal to the update that would occur in $G$. This, as we will show in Section \ref{subsec:proof}, is sufficient to guarantee the convergence of $\bar{b}^t_i, \bar{b}^t_{-i}$ to Nash equilibrium of~$G$.

\vspace{0.3em}

\noindent\textbf{Computing $\Delta^t_i$.}
Given $\hat{b}^t_i$ and $\widetilde{b}^t_i$, $\Delta^t_i$ can be computed as
{\small
$$
\Delta^t_i = \max_{b_{-i} \in \mathcal{B}^p_{-i}}\left|\sum_{z \in \mathcal{Z}}\mathcal{C}(z)\pi_{-i}^{b_{-i}}(z)\left[ \pi_i^{\widetilde{b}^t_i}(z) - \pi_i^{\hat{b}^t_i}(z)\right] u_{-i}(z)\right|.
$$}

Let

{\small
\begin{align} 
u'_{-i}(z) &=\mathcal{C}(z)\left[ \pi_i^{\widetilde{b}^t_i}(z) - \pi_i^{\hat{b}^t_i}(z)\right] u_{-i}(z), \quad\forall z \in \mathcal{Z}\label{eq:u'}\\
u''_{-i}(z) &= -\mathcal{C}(z)\left[ \pi_i^{\widetilde{b}^t_i}(z) - \pi_i^{\hat{b}^t_i}(z)\right] u_{-i}(z), \quad\forall z \in \mathcal{Z}.\label{eq:u''}
\end{align} }

\begin{algorithm}[h!]
\small
\DontPrintSemicolon
\SetKwProg{Function}{function}{}{}
\SetKwFunction{IsGameEnd}{IsGameEnd}
\SetKwFunction{GetPlayerToMove}{GetPlayerToMove}
\SetKwFunction{GetInformationSetFor}{GetInformationSetFor}
\SetKwFunction{ComputePartialDelta}{ComputePartialDelta}
\SetKwFunction{GetValueFromCache}{GetValueFromCache}
\SetKwFunction{StoreToCache}{SaveToCache}

\Function{\ComputeDelta{$G$, $\hat{b}^t_i$,  $\widetilde{b}^t_i$}}{
Let $u'_{-i}$ be the function from eq. \eqref{eq:u'}\;
Let $u''_{-i}$ be the function from eq. \eqref{eq:u''}\;
\Return{$\max ($\ComputePartialDelta{$G.root$, $u'_{-i}$, $\hat{b}^t_i$,  $\widetilde{b}^t_i$}, \ComputePartialDelta{$G.root$, $u''_{-i}$, $\hat{b}^t_i$,  $\widetilde{b}^t_i$}$)$}\;
}
\BlankLine
\Function{\ComputePartialDelta{$h$, $u'_{-i}$, $\hat{b}^t_i$,  $\widetilde{b}^t_i$}}{
\If{$\pi_i^{\hat{b}^t_i}(h) = 0 \wedge \pi_i^{\widetilde{b}^t_i}(h) = 0$}{
\Return{$0$}\;
}
\If{\IsGameEnd{$h$}}{
\Return{$u'_{-i}(h)$}\label{alg:delta_comp:leaves}\;
}
$v \gets $ \GetValueFromCache{$h$}\;
\If{$v \neq null$}{
\Return{$v$}\;
}
\If{\GetPlayerToMove{$h$} $ = -i$}{
$I \gets $ \GetInformationSetFor{$h, \mathcal{I}$}\;
$a_{max} \gets \argmax_{a \in \mathcal{A}(I)}\sum_{h' \in I}$ \ComputePartialDelta{$h' \cdot a$, $u'_{-i}$, $\hat{b}^t_i$,  $\widetilde{b}^t_i$}\label{alg:delta_comp_redundancy}\;
For each $h' \in I$ store $v_{h'} = $ \BR{$h' \cdot a_{max}$, $G$, $\bar{b}_{-i}$, $b^{BR}_i$} to cache\label{alg:delta_comp:cache}\;
$v \gets$ \ComputePartialDelta{$h\cdot a_{max}$, $u'_{-i}$, $\hat{b}^t_i$,  $\widetilde{b}^t_i$}\;
}\Else{
$v \gets \sum_{a \in \mathcal{A}(h)}$\ComputePartialDelta{$h\cdot a$, $u'_{-i}$, $\hat{b}^t_i$,  $\widetilde{b}^t_i$}\;
}
\Return{$v$}\;
}
\caption{$\Delta^t_i$  computation.}
\label{alg:delta_comp}
\end{algorithm}

$\Delta^t_i$ can be computed in $\mathcal{O}(|\mathcal{Z}|)$ by $\mathsf{ComputeDelta}$ depicted in Algorithm \ref{alg:delta_comp}. The computation consists of two calls of function $\mathsf{ComputePartialDelta}$ similar to the computation of the best response described in Section \ref{sec:br}. The only difference is the use of functions shown in eqs. \eqref{eq:u'} and \eqref{eq:u''} when evaluating the leaves (line \ref{alg:delta_comp:leaves}). $\mathsf{ComputePartialDelta}$ using $u'_{-i}$ searches for the largest positive difference in the utility of pure strategies of $-i$, while $\mathsf{ComputePartialDelta}$ using $u''_{-i}$ searches for the highest negative difference. The two calls are necessary, since $\mathsf{ComputePartialDelta}$ using $u'_{-i}$ cannot reliably detect negative differences in the utility, since it will always prefer choosing pure strategy with no difference in the utility over the pure strategy with the negative difference. 
Notice, that similarly to the best response computation, we use cache to eliminate redundant tree traversals caused by line \ref{alg:delta_comp_redundancy}. 

\begin{example}\label{ex:fpira}
\begin{figure}
\begin{subfigure}{0.22\textwidth}
\centering
\includegraphics[width=\textwidth]{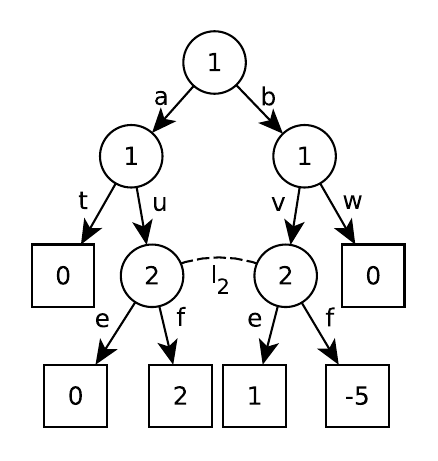}
~\\\vspace{-3mm}
$G$
\end{subfigure}
\begin{subfigure}{0.22\textwidth}
\centering
\includegraphics[width=\textwidth]{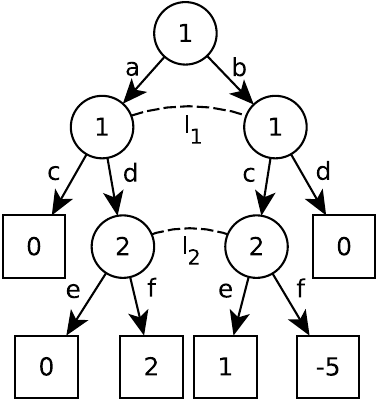}
~\\\vspace{-3mm}
$G^1$
\end{subfigure}
\begin{minipage}{0.55\textwidth}
\begin{subfigure}{0.32\textwidth}
\centering
\includegraphics[width=0.8\textwidth]{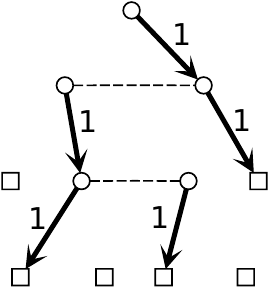}
\footnotesize{Initial $\bar{b}^0$}\\
\end{subfigure}
\begin{subfigure}{0.32\textwidth}
\centering
\includegraphics[width=0.8\textwidth]{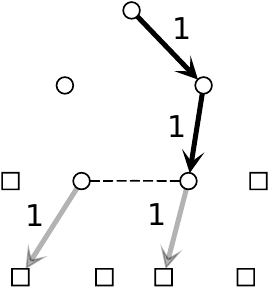}
\footnotesize{Iteration $1$}
\end{subfigure}
\begin{subfigure}{0.32\textwidth}
\centering
\includegraphics[width=0.8\textwidth]{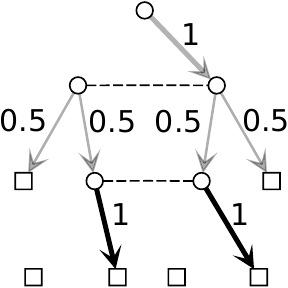}
\footnotesize{Iteration $2$}
\end{subfigure}
\begin{subfigure}{0.32\textwidth}
\centering
~\\\vspace{1mm}
\includegraphics[width=0.8\textwidth]{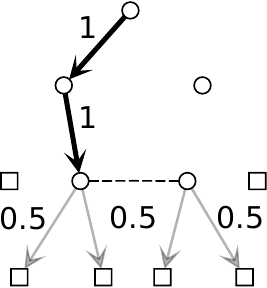}
\footnotesize{Iteration $3$}
\end{subfigure}
\begin{subfigure}{0.32\textwidth}
\centering
\includegraphics[width=0.8\textwidth]{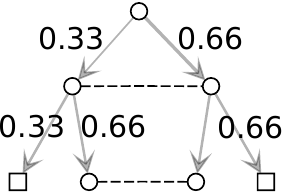}
\footnotesize{Abstract ($\hat{b}^3_1$)}
\end{subfigure}
\begin{subfigure}{0.32\textwidth}
\centering
\includegraphics[width=0.8\textwidth]{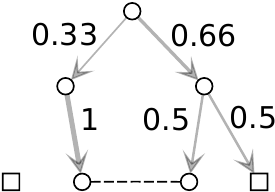}
\footnotesize{Original ($\tilde{b}^3_1$)}
\end{subfigure}
\end{minipage}
\caption{{Game $G$ for demonstration of \textsc{FPIRA} with initial abstraction $G^1$. Using the initial strategies $\bar{b}^0$, we show the average strategy for one player (grey) and the best response for the other (black) in individual iterations. In the third iteration, the average strategy in the abstract and the original game are strategically different.}}
\label{fig:example}
\end{figure}
Let us demonstrate several iterations of \textsc{FPIRA} algorithm.
Consider game $G$ and its imperfect recall abstraction $G^1$ from Figure \ref{fig:example}. 
The function $\Xi_1$  is $\Xi_1(t) = \Xi_1(v) = c,  \Xi_1(u) = \Xi_1(w) = d$, identity otherwise. Note that when we apply strategies from $G$ to $G^t$ and vice versa in iteration $t$, we assume that it is done with respect to $\Xi_t$ and $\Xi_t^{-1}$. Let us assume that  \textsc{FPIRA} first initializes the strategies to $\bar{b}_1^0(b) = \bar{b}_1^0(d) = 1, \bar{b}_2^0(e) = 1$, as shown in Figure~\ref{fig:example}. 

\vspace{0.3em}
 
\noindent\textbf{Iteration 1:} Player 1 starts in iteration 1. \textsc{FPIRA} computes $b^1_1  \in BR(\bar{b}_2^0)$ in $G$, resulting in $b^1_1(b)  = b^1_1(v) = 1$. Next, \textsc{FPIRA} checks whether $b^1_1$ is playable in $G^1$. Since there is no information set in $G^1$ for which $b^1_1$ assigns more than one action, we do not need to update $G^1$ in any way. We follow by computing $\hat{b}^1_1$ and $\widetilde{b}^1_1$ according to Lemma \ref{lemma:strat_update} with $\lambda_1 = \lambda_2 = 0.5$. In this case $\hat{b}^1_1(b) = \widetilde{b}^1_1(b) = 1, \hat{b}^1_1(c) = \widetilde{b}^1_1(v) = 0.5$. Since $\hat{b}^1_1$ and $\widetilde{b}^1_1$ are equal, w.r.t. $\Xi_1$, we know that $\Delta_i = 0$. Hence we let $G^2 = G^1$, $\bar{b}^1_1 = \hat{b}^1_1$ and $\Xi_2 =\Xi_1$. 

\vspace{0.3em}
 
\noindent\textbf{Iteration 2:} Player 2, whose information sets were not abstracted, continues in iteration 2. \textsc{FPIRA} computes the best response to $\bar{b}^1_1$, resulting in $b^2_2(f) = 1$. The algorithm then computes $\hat{b}^2_2$ and $\widetilde{b}^2_2$, resulting in $\hat{b}^2_2(e) = \widetilde{b}^2_2(e) = 0.5$. Hence, we let $G^3 = G^2$, $\bar{b}^2_2 = \hat{b}^2_2$ and $\Xi_3 =\Xi_2$.

\vspace{0.3em}
 
\noindent\textbf{Iteration 3:}
The best response in this iteration is $b^3_1(a) = b^3_1(u) = 1$, which is again playable in $G^3$, hence we do not need to update $G^3$ at this point. \textsc{FPIRA} computes $\hat{b}^3_1$ resulting in $\hat{b}^3_1(a) = \frac{1}{3}, \hat{b}^3_1(d) = \frac{2}{3}$, $\widetilde{b}^3_1$ is, on the other hand, $\widetilde{b}^3_1(a) = \frac{1}{3}$, $\widetilde{b}^3_1(u) = 1, \widetilde{b}^3_1(w) = 0.5$ (both according to Lemma \ref{lemma:strat_update} with $\lambda_1 = \frac{2}{3}, \lambda_2 = \frac{1}{3}$). In this case, $\Delta^3_1 = \frac{1}{3}$ since by playing $f$ player 2 gets $-\frac{2}{3}$ $\left(=\frac{1}{3}\cdot\frac{2}{3}\cdot 2+\frac{2}{3}\cdot\frac{1}{3}\cdot(-5)\right)$ against $\hat{b}^3_1$ compared to $-1$ $\left(=\frac{1}{3}\cdot 1 \cdot 2 + \frac{2}{3}\cdot\frac{1}{2}\cdot(-5)\right)$ against $\widetilde{b}^3_1$. Hence, the algorithm splits all imperfect recall information sets reachable when playing $b^3_1$, in this case $I_1$, as described in Section \ref{subsec:G_update}, Case 2, resulting in $G$. %
\end{example}
\subsubsection{Theoretical Properties}
\label{subsec:proof}
Here, we show that the convergence guarantees of FP in two-player zero-sum perfect recall game $G$ \cite{heinrich2015fictitious} directly apply to \textsc{FPIRA} solving $G$.
\begin{theorem}\label{thm:conv}
Let $G$ be a perfect recall two-player zero-sum EFG. 
Assume that initial strategies $\bar{b}^0_1, \bar{b}^0_2$ in \textsc{FPIRA} and initial strategies $\bar{b'}^0_1, \bar{b'}^0_2$ in the FP are realization equivalent, additionally assume that the same tie breaking rules are used when more than one best response is available in any iteration. 
The exploitability of $\bar{b}^t_i$ computed by \textsc{FPIRA} applied to $G$ is exactly equal to the exploitability of $\bar{b'}^t_i$, computed by FP applied to $G$ in all iterations $t$ and for all $i \in \mathcal{N}$.
\end{theorem}
\begin{proof}
The proof is done by induction. If 

{\small  
\begin{align}
\forall b_{-i} \in \mathcal{B}^p_{-i}: u_{-i}(\bar{b}^t_i,  b_{-i}) &= u_{-i}(\bar{b'}^t_i,  b_{-i}),\label{eq:ind_assumption}\\
\forall b_i \in \mathcal{B}^p_{i}: u_{i}(b_i, \bar{b}^t_{-i}) &= u_{i}(b_i, \bar{b'}^t_{-i})\label{eq:ind_assumption1},
\end{align} }

where $\mathcal{B}^p$ is the set of pure behavioral strategies in $G$, then

{\small $$b_{-i} \in \mathcal{B}^p_{-i}: u_{-i}(\bar{b}^{t+1}_i, b_{-i}) = u_{-i}(\bar{b'}^{t+1}_i, b_{-i}).$$ }

The initial step trivially holds from the assumption that initial strategies in \textsc{FPIRA} and initial strategies and in FP are realization equivalent.
Now let us show that the induction step holds. Let $b^t_{i}$ be the best response chosen in iteration $t$ in \textsc{FPIRA} and ${b'}^t_{i}$ be the best response chosen in $t$ in FP.
From \eqref{eq:ind_assumption1} and the use of the same tie breaking rule we know that $b^t_{i} = {b'}^t_{i}$.
From Lemma \ref{lemma:strat_update} we know that 

{\small 
\begin{align}
\forall b_{-i} \in \mathcal{B}^p_{-i}:\ &u_{-i}(\bar{b'}_i^{t+1}, b_{-i}) =\nonumber\\
&\frac{t_i}{t_i+1}u_{-i}(\bar{b'}^t_i, b_{-i}) + \frac{1}{t_i+1}u_{-i}({b'}^t_i, b_{-i}).\nonumber
\end{align}}

However, same holds also for $\bar{b}^{t+1}_i$ since \textsc{FPIRA} creates $G^{t+1}$ from $G^t$ so that $\Delta^t_i = 0$. Hence

{\small 
\begin{align}
\forall b_{-i} \in \mathcal{B}^p_{-i}:\ &u_{-i}(\bar{b}_i^{t+1}, b_{-i}) =\nonumber\\
&\frac{t_i}{t_i+1}u_{-i}(\bar{b}^t_i, b_{-i}) + \frac{1}{t_i+1}u_{-i}({b}^t_i, b_{-i}).\nonumber
\end{align}}

From \eqref{eq:ind_assumption} and from the equality $b^t_{i} = {b'}^t_{i}$ follows that  

{\small $$\forall b_{-i} \in \mathcal{B}^p_{-i} u_{-i}(\bar{b}_i^{t+1}, b_{-i}) = u_{-i}(\bar{b'}_i^{t+1}, b_{-i}),$$}

and therefore also 

{\small \begin{equation}
\max_{b_{-i} \in \mathcal{B}^p_{-i}} u_{-i}(\bar{b}_i^{t+1}, b_{-i}) = \max_{b_{-i} \in \mathcal{B}^p_{-i}} u_{-i}(\bar{b'}_i^{t+1}, b_{-i}).\nonumber
\end{equation}
}
\end{proof}

\subsubsection{Storing the Information Set Map}\label{sec:fpira_abstraction_storage}
In this section we discuss the memory requirements for storing the mapping of information sets $\mathcal{I}$ of $G$ to $\mathcal{I}^t$ of $G^t$ in FPIRA.

\textbf{Initial abstraction.} As described in Section \ref{sec:init_abstr}, the mapping between any $I \in \mathcal{I}_i$ and its abstracted information set in $G^1$ is perfectly defined by $|seq_i(I)|$ and $|\mathcal{A}(I)|$. Hence the mapping can always be determined for any given $I \in \mathcal{I}$ without using any additional memory.

\textbf{Case 1 Update.} When updating the abstraction resulting in $G^t$ according to Case 1 in Section \ref{subsec:G_update}, FPIRA can split $I \in \widetilde{\mathcal{I}}^{t-1}_i$ to multiple abstracted information sets. FPIRA is then forced to store the mapping for each newly created abstracted information set $I'$, for each $I'' \in \Phi^{-1}_t(I')$. This is necessary since $|seq_i(I'')|$ and $|\mathcal{A}(I'')|$ is no longer unique identifier of $I'$. In our implementation, we use a unique integer $j$ to represent the new mapping for $I'$, i.e., the algorithm stores $j$ for each $I'' \in \Phi^{-1}_t(I')$. We use $j$ corresponding to the number of newly created information sets during the run of FPIRA. Finally, let $\mathcal{I}^m_n$ be all the information sets in $\widetilde{\mathcal{I}}_i^t$ such that $\forall I \in \mathcal{I}^m_n |seq_i(I)| = m \wedge |\mathcal{A}(I)| = n$. We can always identify the mapping to one of $\mathcal{I}^m_n$ without using any memory by the sequence length and number of actions, as long as the rest of $I \in \mathcal{I}^m_n$ uses mapping with the unique integer. We keep track of $I \in \mathcal{I}^m_n$ with largest $|\Phi^{-1}_t(I)|$ in each $\mathcal{I}^m_n$ and use the sequence length and number of actions to represent the mapping for all $\Phi^{-1}_t(I')$ to minimize the memory requirements. In Section \ref{sec:experiments} we empirically demonstrate that the memory required to store the mapping is small. Notice that in all $I' \in \mathcal{I}^t \setminus \widetilde{\mathcal{I}}^t$, the mapping is defined by the domain description of $G$ and hence no memory is required.

\textbf{Case 2 Update.} When updating the abstraction resulting in $G^t$  according to Case 2 in Section \ref{subsec:G_update}, there is no additional memory required to store the updated mapping compared to the mapping used in $G^{t-1}$. This holds since this abstraction update only removes information set from $G$ from the abstracted information sets in $G^{t-1}$. And so we use the same mapping as in $G^{t-1}$ for abstracted information sets in $G^t$, and the information set structure provided by the domain description in the rest.

\subsubsection{Memory and Time Efficiency}\label{sec:fpira_memory}
\textsc{FPIRA} needs to store the average behavioral strategy for every action in every information set of the solved game, hence storing the average strategy in $G^t$ instead of $G$ results in significant memory savings directly proportional to the decrease of information set count. When the algorithm computes $\widetilde{b}_i^t$, it can temporarily refine the information set structure of $G^t$ only in the parts of the tree that can be visited when playing the pure best response $b^{t}_i$ according to $\mathcal{I}_i$ to avoid representing and storing $G$. 
Additional memory used to store the current abstraction mapping is discussed in Section \ref{sec:fpira_abstraction_storage}. %

When computing the best response (see Section \ref{sec:br}), the algorithm needs to store the best response strategy and the cache that is used to eliminate additional tree traversals.
FPIRA stores the behavior only in the parts of the game reachable due to actions of $i$ in $b^t_i$ (line \ref{alg:br:strategy_cleanup} in Algorithm \ref{alg:br}) and due to $\bar{b}^t_{-i}$ (lines \ref{alg:br:opp_prunning_start} and \ref{alg:br:opp_prunning_end} in Algorithm \ref{alg:br}). For this reason and since $i$ plays only 1 action in his information sets in $b_i^t$, there are typically large parts of the game tree where $b_i^t$ does not prescribe any behavior. %
The cache (used also in the computation of $\Delta^t_i$) stores one number for each state visited during the computation. We show the size of the cache in Section \ref{sec:experiments}. 
If necessary, the memory requirements of the cache can be reduced by limiting its size and hence balancing the memory required and the additional tree traversals performed. Additionally, efficient domain-specific implementations of best response (e.g., on poker \cite{johanson2011accelerating}) can be employed to further reduce the memory and time requirements. 

We empirically demonstrate the size of all the data structures stored during the run of FPIRA in Section \ref{sec:experiments}.

The iteration of \textsc{FPIRA} takes approximately three times the time needed to perform one iteration of FP in $G$, as it now consists of the standard best response computation in $G$, two modified best response computations to obtain $\Delta_i^t$ and two updates of average behavioral strategies (which are faster than the update in $G$ since the average strategy is smaller).

{
\paragraph{Scalability limitations}
The computation time of FPIRA is given by the time required for computing a best response. This computation is, in general, linear in the number of histories in the game. There may be domain specific implementations of best response, which use a compact representation of the game or very efficient pruning. In poker, domain-specific best response computation proposed in \cite{johanson2011accelerating} can compute the best response for Limit Texas Hold'em with approximately $10^{18}$ histories with a highly optimized computation distributed among 72 cores in a little over one day. Since FPIRA commonly needs over ten thousands of iterations to converge, it would have to be distributed among substantially more machines to solve a game of this size. If we consider the number of hidden card combinations to be a constant, even the optimized best response computation is still linear in the number of histories in the game. Therefore, on a common desktop machine with six cores, a domain-specific implementation would likely solve to reasonable precision games of $10^{13}$ histories within a week of computation.

\subsubsection{Possible Enhancements}

FPIRA is strict in splitting the information set immediately when computing the average strategy in the abstraction or enforcing the same action in the best response would cause any error. A natural extension is to allow small errors in the strategy without the splits. However, our initial experiments in this direction did not show a substantial decrease in the size of computed abstractions, and it makes the formal analysis of the algorithm substantially more difficult. It is not clear the error does not keep growing with increasing number of iterations. Therefore, we leave this research direction to future work.

}

\subsection{CFR+ for Imperfect Recall Abstractions}\label{sec:IRCFR+}
In this section, we describe CFR+ for Imperfect Recall Abstractions (\IRCFR). We first provide a high-level idea of \IRCFR, followed by detailed explanation of all its parts with pseudocodes and proof of its convergence to NE in two-player zero-sum EFGs. Finally, we discuss the memory requirements and runtime of \IRCFR.

Given a two-player zero-sum perfect recall EFG $G$, \IRCFR first creates a coarse imperfect recall abstraction of $G$ as described in Section \ref{sec:init_abstr}. The algorithm then iteratively solves the game using CFR+, {traversing the whole unabstracted game tree in each iteration}. All regrets and average strategies computed as a part of CFR+ are stored in the information set structure of the abstraction. To ensure the convergence to the Nash equilibrium of $G$, \IRCFR updates the structure of the abstraction in every iteration based on the differences between the results obtained by CFR+ in the abstraction and the expected behavior of CFR+ in $G$.

\begin{algorithm}[h!]
\small
\DontPrintSemicolon
\SetKwInOut{Input}{input}
\SetKwInOut{Output}{output}
\SetKwProg{Function}{function}{}{}

\SetKwFunction{InitAbstraction}{InitAbstraction}
\SetKwFunction{InitRegrets}{InitRegrets}
\SetKwFunction{UniformStrategy}{UniformStrategy}
\SetKwFunction{CFR}{ComputeRegrets}
\SetKwFunction{GetPlayer}{GetPlayer}
\SetKwFunction{IsTerminal}{IsTerminal}
\SetKwFunction{ChooseInformationSets}{SampleInformationSets}
\SetKwFunction{ComputePerfectRecallRegrets}{OrigRegrets}
\SetKwFunction{UpdateAbstractionHeuristic}{UpdateAbstractionForHeuristic}
\SetKwFunction{UpdateAbstractionBound}{UpdateAbstractionForBound}
\SetKwFunction{GetInformationSetFor}{GetInformationSetFor}
\SetKwFunction{GetRegretsFor}{GetRegretsFor}
\SetKwFunction{RegretMatching}{RegretMatching+}
\SetKwFunction{UpdateAverageStrategy}{UpdateAverageStrategy}
\SetKwFunction{RemoveNegativeRegrets}{RemoveNegativeRegrets}
\SetKwFunction{Zeros}{Zeros}
\Input{$G$, $k_h$, $k_b$, delay, $\varepsilon$}
\Output{$(\bar{b}_1 , \bar{b}_2)$ -- $\varepsilon$-Nash equilibrium of $G$}

\BlankLine

$G^1 \gets $ \InitAbstraction{G} \label{alg:cfr_abstr:init_abstr}\;
$r \gets$ \InitRegrets{$\mathcal{I}^1$}\;
$\bar{b}_1 \gets$ \UniformStrategy{$\mathcal{I}$}, $\bar{b}_2 \gets$ \UniformStrategy{$\mathcal{I}$}\;
$t \gets 1$, $t_{next} \gets 1$, $t_{last} \gets 0$, $j \gets 0$\;
$r_b \gets \emptyset$\;

\While{$u_1($\BR{$G, \bar{b}_2$}$, \bar{b}_2) - u_1(\bar{b}_1,$ \BR{$G, \bar{b}_1$}$) > \varepsilon$}{\label{alg:cfr_abstr:loop}
  $\hat{\mathcal{I}}_h \gets$ \ChooseInformationSets{$\widetilde{\mathcal{I}}^t_i$, $k_h$}\label{alg:cfr_abstr:choose_is_heur}\;
  $r_h \gets$ \InitRegrets{$\hat{\mathcal{I}}_h$}\;  
  \If{$t = t_{next}$}{\label{alg:cfr_abstr:bound_resample_cond}
    $\hat{\mathcal{I}}_b \gets$ \ChooseInformationSets{$\widetilde{\mathcal{I}}^t$, $k_b$}\label{alg:cfr_abstr:choose_is_bound}\;
    $r_b \gets$ \InitRegrets{$\hat{\mathcal{I}}_b$}\;  
    $t_{last} \gets t$\;    
    $t_{next} \gets t + 2^j$\label{alg:cfr_abstr:next_it_update}\;
    $j \gets j + 1$\;
  }
  \CFR{$\mathcal{I}^t$, $G^t.root$, \GetPlayer{$t$}, 1, 1, $r$, $r_h$, $r_b$, $\hat{\mathcal{I}}_h$, $\hat{\mathcal{I}}_b$}\label{alg:cfr_abstr:cfr+}\; 
  \RemoveNegativeRegrets{$r$}\label{alg:cfr_abstr:neg_r}\;
  \If{$t > delay$}{
    \UpdateAverageStrategy{$r$, $t$}\label{alg:cfr_abstr:avg_strat}\;
  } 
    $G^{t+1} \gets$ \UpdateAbstractionHeuristic{$G^t$, $\hat{\mathcal{I}}_h$, $r_h$, $r$, $\bar{b}_1$, $\bar{b}_2$, $t$}\label{alg:cfr_abstr:abstr_update_heur}\;
    delete $r_h$\;
    \If{$t \neq t_{last}$}{   
      $G^{t+1} \gets$ \UpdateAbstractionBound{$G^{t+1}$, $\hat{\mathcal{I}}_b$, $r$, $r_b$, $\bar{b}_1$, $\bar{b}_2$, $t$, $t_{last}$}\label{alg:cfr_abstr:abstr_update_bound}\;
      }
  $t \gets t+1$\;
}

\caption{CFR+ for Imperfect Recall Abstractions}
\label{alg:cfr_abstr}
\end{algorithm}

In Algorithm \ref{alg:cfr_abstr} we provide the pseudocode of the \IRCFR. \IRCFR is given the original perfect recall game $G$, the desired precision of approximation of the NE $\epsilon$ and the limits $k_h$ and $k_b$ on the memory that can be used to update the abstraction. Finally it is given the $delay$ which represents the number of initial iterations for which the algorithm does not update the average strategy. The algorithm starts by creating a coarse imperfect recall abstraction $G^1$ of the given game $G$ (line \ref{alg:cfr_abstr:init_abstr}, see Section \ref{sec:init_abstr}). It stores the regrets $r$ and average strategies $\bar{b}_1, \bar{b}_2$ for each information set of the current abstraction. The algorithm then simultaneously solves the abstraction and updates its structure until the $\bar{b}_1, \bar{b}_2$ form an $\varepsilon$-Nash equilibrium of $G$ (line \ref{alg:cfr_abstr:loop}). The players take turn updating their strategies and regrets. In every iteration, the algorithm updates the regrets and the current strategy for the acting player $i$ according to CFR+ update (line \ref{alg:cfr_abstr:cfr+}, see Section \ref{sec:cfr+} and Algorithm \ref{alg:CFR+_func}). As a part of the CFR+ update the algorithm removes the negative regrets (line \ref{alg:cfr_abstr:neg_r}) and updates the average strategy (line \ref{alg:cfr_abstr:avg_strat}). Note that we follow CFR+ as described in Section \ref{sec:cfr+} and so the update of the average strategy starts only after a fixed number of iterations denoted as $delay$. The average strategy is then updated according to eq. \eqref{eq:cfr+stratupdate}.

The algorithm continues with the update of the current abstraction $G^t$.
There are two procedures for updating the abstraction. 

First, the abstraction is updated to guarantee the convergence of the algorithm to the Nash equilibrium (see Section \ref{sec:regr_bound_update} for more details). As a part of this abstraction update, \IRCFR samples a subset $\hat{\mathcal{I}}_b$ of information sets of $G$ (line \ref{alg:cfr_abstr:choose_is_bound}). It then checks the immediate regret in all $I \in \hat{\mathcal{I}}_b$ for a given number of iterations before again resampling $\hat{\mathcal{I}}_b$. During these iterations, it updates the abstraction so that any $I \in \hat{\mathcal{I}}_b$, where the immediate regret decreases slower than a given function, is removed from its abstracted information set.

Second, the abstraction is updated using a heuristic update which significantly improves the empirical convergence of the algorithm (see Section \ref{sec:heur_update} for more details). The heuristic update samples a subset $\hat{\mathcal{I}}_h$ of information sets of $i$ in $G$ in every iteration of \IRCFR (line \ref{alg:cfr_abstr:choose_is_heur}). 
It then keeps track of regrets in all $I \in \hat{\mathcal{I}}_h$ in this iteration. Finally, it uses these regrets to update the abstraction so that only information sets with similar regrets remain grouped.

\begin{algorithm}[h!]
\small
\DontPrintSemicolon
\SetKwInOut{Input}{input}
\SetKwInOut{Output}{output}
\SetKwProg{Function}{function}{}{}

\SetKwFunction{InitAbstraction}{InitAbstraction}
\SetKwFunction{InitRegrets}{InitRegrets}
\SetKwFunction{UniformStrategy}{UniformStrategy}
\SetKwFunction{CFR}{ComputeRegrets}
\SetKwFunction{GetPlayer}{GetPlayer}
\SetKwFunction{IsTerminal}{IsTerminal}
\SetKwFunction{ChooseInformationSets}{ChooseInformationSets}
\SetKwFunction{ComputePerfectRecallRegrets}{PRRegrets}
\SetKwFunction{UpdateAbstraction}{UpdateAbstraction}
\SetKwFunction{GetInformationSetFor}{GetInformationSetFor}
\SetKwFunction{GetRegretsFor}{GetRegretsFor}
\SetKwFunction{GetPlayerToMove}{GetPlayerToMove}
\SetKwFunction{RegretMatching}{RegretMatching+}
\SetKwFunction{UpdateAverageStrategy}{UpdateAverageStrategy}
\SetKwFunction{RemoveNegativeRegrets}{RemoveNegativeRegrets}
\SetKwFunction{Zeros}{Zeros}
\Function{\CFR{$\mathcal{I}$, ${h}$, ${i}$, $\pi_1$, $\pi_2$, ${r}$, ${r}_h$, ${r}_b$, $\hat{\mathcal{I}}_h$, $\hat{\mathcal{I}}_b$}}{
    \If{\IsTerminal{${h}$}}{
          \Return{$u_i({h})$}\;
      }
      \If{\GetPlayerToMove{${h}$}$ = c$}{
        \If{${i}$ $= 1$}{
            \Return $\sum_{a \in \mathcal{A}(h)} b_c(a)\cdot$ \CFR{$\mathcal{I}$, ${h}\cdot a$, ${i}$, $\pi_1$, $b_c(a) \cdot \pi_2$, ${r}$, ${r}_h$, ${r}_b$, $\hat{\mathcal{I}}_h$, $\hat{\mathcal{I}}_b$}\;
        }
        \Return $\sum_{a \in \mathcal{A}(h)} b_c(a)\cdot$ \CFR{$\mathcal{I}$, ${h}\cdot a$, ${i}$, $b_c(a) \cdot \pi_1$, $\pi_2$, ${r}$, ${r}_h$, ${r}_b$, $\hat{\mathcal{I}}_h$, $\hat{\mathcal{I}}_b$}\;
      }
      ${I}$ $\gets$ \GetInformationSetFor{${h}$, $\mathcal{I}$}\; 
      ${r}_I \gets$ \GetRegretsFor{${I}$, ${r}$}\; 
      ${b}^t \gets$ \RegretMatching{${r}_I$}\;
      $v_{b^t} \gets 0$\;
      $v \gets$ \Zeros{$|\mathcal{A}(h)|$}\;
      \For{$a \in \mathcal{A}(h)$}{
        \If{${i}$ $= 1$}{
            $v[a] \gets$ \CFR{$\mathcal{I}$, ${h}\cdot a$, ${i}$, $b^t(a)\cdot \pi_1$, $\pi_2$, ${r}$, ${r}_h$, ${r}_b$, $\hat{\mathcal{I}}_h$, $\hat{\mathcal{I}}_b$}\;
        }\Else{
            $v[a] \gets$ \CFR{$\mathcal{I}$, ${h}\cdot a$, ${i}$, $\pi_1$, $b^t(a)\cdot \pi_2$, ${r}$, ${r}_h$, ${r}_b$, $\hat{\mathcal{I}}_h$, $\hat{\mathcal{I}}_b$}\;
        }
        $v_{b^t} \gets v_{b^t} + b^t(a)\cdot v[a]$\;
    }
    \If{\GetPlayerToMove{${h}$}$ = {i}$}{
        \For{$a \in \mathcal{A}(h)$}{
            $r_I[a] \gets r_I[a] + \pi_{-i}\cdot (v[a] - v_{b^t})$\;
        }
    
        \If{$I \in \hat{\mathcal{I}}_h$}{\label{alg:CFR+_func:heur_update_start}
              ${r}^h_I \gets$ \GetRegretsFor{${I}$, ${r}_h$}\; 
            \For{$a \in \mathcal{A}(h)$}{
                $r^h_I[a] \gets r^h_I[a] + \pi_{-i}\cdot (v[a] - v_{b^t})$\label{alg:CFR+_func:heur_update_end}\;
            }
        }
        \If{$I \in \hat{\mathcal{I}}_b$}{\label{alg:CFR+_func:bound_update_start}
              ${r}^b_I \gets$ \GetRegretsFor{${I}$, ${r}_b$}\; 
            \For{$a \in \mathcal{A}(h)$}{
                $r^b_I[a] \gets r^b_I[a] + \pi_{-i}\cdot (v[a] - v_{b^t})$\label{alg:CFR+_func:bound_update_end}\;
            }
        }
    }
    \Return{$v_{b^t}$}
}
\caption{Regret update}
\label{alg:CFR+_func}
\end{algorithm}

\begin{algorithm}[h!]
\small
\DontPrintSemicolon
\SetKwInOut{Input}{input}
\SetKwInOut{Output}{output}
\SetKwProg{Function}{function}{}{}
\SetKw{Break}{break}
\SetKwFunction{InitAbstraction}{InitAbstraction}
\SetKwFunction{InitRegrets}{InitRegrets}
\SetKwFunction{UniformStrategy}{UniformStrategy}
\SetKwFunction{CFR}{CFR+}
\SetKwFunction{GetPlayer}{GetPlayer}
\SetKwFunction{IsTerminal}{IsTerminal}
\SetKwFunction{ChooseInformationSets}{SampleInformationSets}
\SetKwFunction{ComputePerfectRecallRegrets}{PRRegrets}
\SetKwFunction{UpdateAbstraction}{UpdateAbstraction}
\SetKwFunction{GetInformationSetFor}{GetInformationSetFor}
\SetKwFunction{GetRegretsFor}{GetRegretsFor}
\SetKwFunction{RegretMatching}{RegretMatching+}
\SetKwFunction{UpdateAverageStrategy}{UpdateAverageStrategy}
\SetKwFunction{RemoveNegativeRegrets}{RemoveNegativeRegrets}
\SetKwFunction{Zeros}{Zeros}
\SetKwFunction{Intersection}{Intersection}
\SetKwFunction{OriginalSets}{UnabstractedSets}
\SetKwFunction{GetRandomAbstractedSet}{GetRandomAbstractedSet}
\SetKwFunction{NoMoreAbstractedSets}{AllSetsSampled}
\SetKwFunction{RandomSubset}{RandomSubset}
\Function{\ChooseInformationSets{$\widetilde{\mathcal{I}}$, $k$}}{
    $\hat{\mathcal{I}} \gets \emptyset$\;
    \While{$|\hat{\mathcal{I}}|<k$}{
        \If{\NoMoreAbstractedSets{$\widetilde{\mathcal{I}}$, $\hat{\mathcal{I}}$}}{
             \Break\;
        }
        ${I} \gets $ \GetRandomAbstractedSet{$\widetilde{\mathcal{I}}$}\label{alg:choose_IS:sampling}\;
        $\mathcal{I}'' \gets \Phi_t^{-1}({I})$\;
        \If{$|\mathcal{I}''| + |\hat{\mathcal{I}}| \leq k$} {
            $\hat{\mathcal{I}} \gets \hat{\mathcal{I}} \cup \mathcal{I}''$\;
        }\Else{
            $\hat{\mathcal{I}} \gets \hat{\mathcal{I}}\ \cup\ $ random subset of $\mathcal{I}''$ with size $k - |\hat{\mathcal{I}}|$\;
        }
    }
    \Return $\hat{\mathcal{I}}$
}
\caption{Sampling of unabstracted sets for unabstracted regret storage.}
\label{alg:choose_IS}
\end{algorithm}

\subsubsection{Regret Bound Update}
\label{sec:regr_bound_update}
\begin{algorithm}[h!]
\small
\DontPrintSemicolon
\SetKwInOut{Input}{input}
\SetKwInOut{Output}{output}
\SetKwProg{Function}{function}{}{}

\SetKwFunction{InitAbstraction}{InitAbstraction}
\SetKwFunction{InitRegrets}{InitRegrets}
\SetKwFunction{UniformStrategy}{UniformStrategy}
\SetKwFunction{CFR}{CFR+}
\SetKwFunction{GetPlayer}{GetPlayer}
\SetKwFunction{IsTerminal}{IsTerminal}
\SetKwFunction{ChooseInformationSets}{ChooseInformationSets}
\SetKwFunction{ComputePerfectRecallRegrets}{PRRegrets}
\SetKwFunction{UpdateAbstractionBound}{UpdateAbstractionForBound}
\SetKwFunction{GetInformationSetFor}{GetInformationSetFor}
\SetKwFunction{GetRegretsFor}{GetRegretsFor}
\SetKwFunction{GetStrategyFor}{GetStrategyFor}
\SetKwFunction{RegretMatching}{RegretMatching+}
\SetKwFunction{UpdateAverageStrategy}{UpdateAverageStrategy}
\SetKwFunction{RemoveNegativeRegrets}{RemoveNegativeRegrets}
\SetKwFunction{Zeros}{Zeros}
\SetKwFunction{Intersection}{Intersection}
\SetKwFunction{OriginalSets}{UnabstractedSets}
\SetKwFunction{CreateNewIS}{CreateNewIS}
\SetKwFunction{InitRegret}{InitRegret}
\Function{\UpdateAbstractionBound{${G}^t$, $\hat{\mathcal{I}}_b$, ${r}$, ${r}_b$, ${\bar{b}}_1$, ${\bar{b}}_2$, ${t}$, ${t}_{last}$}}{
    \For{${I} \in \widetilde{\mathcal{I}^t}$}{\label{alg:abstr_update_bound:iter1}
        $\mathcal{I}'' \gets \Phi_t^{-1}({I})$\;
            \For{${I}'' \in\ $\Intersection{$\mathcal{I}''$, $\hat{\mathcal{I}}_b$}}    {        
                ${r}^b_{{I}''} \gets $ \GetRegretsFor{${I}''$, ${r}_b$}\;
                \If{$\frac{max_{a \in \mathcal{A}(I'')}{r}^b_{{I}''}[a]}{t - t_{last}} > L^{t_{last}}_{I''}(t)$}{\label{alg:abstr_update_bound:condition}
                    ${I}_{new} \gets $ \CreateNewIS{$\mathcal{I}'' \setminus {I}''$}\;\label{alg:abstr_update_bound:disconnect_start}
                    ${G}^t.\mathcal{I} \gets {G}^t.\mathcal{I} \setminus {I}$\;
                    ${G}^t.\mathcal{I} \gets {G}^t.\mathcal{I} \cup \{$\CreateNewIS{${I}''$}, ${I}_{new}\}$\;
                    ${r}_{I} \gets $ \GetRegretsFor{${I}$, ${r}$}\;
          ${r} \gets {r} \setminus {r}_{I}$ \;
                    ${r}_{I_{new}} \gets $ \InitRegret{${I}_{new}$}\;
                    ${r} \gets {r} \cup {r}_{I_{new}}$\;
                    ${r}_{I''} \gets $ \InitRegret{${I}''$}\;
                    ${r} \gets {r} \cup {r}_{I''}$\;
          $\hat{\mathcal{I}}_b \gets \hat{\mathcal{I}}_b \setminus I''$\;
          ${r}_b \gets {r}_b \setminus {r}^b_{{I}''}$\; 
                    \If{\GetPlayer{${t}$} = 1}{
                        ${\bar{b}}_1 \gets {\bar{b}}_1 \setminus$ \GetStrategyFor{${I}$, ${\bar{b}}_1$}\;        
                    }\Else{
                        ${\bar{b}}_2 \gets {\bar{b}}_2 \setminus$ \GetStrategyFor{${I}$, ${\bar{b}}_2$}\label{alg:abstr_update_bound:disconnect_end}\;        
                    }
                }
            }
        
    }
    \Return{${G}^t$}\;
}
\caption{Abstraction update for regret bound}
\label{alg:abstr_update_bound}
\end{algorithm}
In this section we present more detailed description of the update of the abstraction based on the regret bound.

Let $\mathcal{T}^T = (T_1, ..., T_n)$ be a sequence of iterations, where every $T_{j+1} - T_{j} = 2^{j}$ for all $j \in \{1, ..., n - 1\}$ (the elements of $\mathcal{T}$ are computed on line \ref{alg:cfr_abstr:next_it_update} in Algorithm \ref{alg:cfr_abstr}).
 As a part of the abstraction update, the algorithm samples a subset $\hat{\mathcal{I}}_b \subseteq \mathcal{I}$ of information sets of $G$ in predetermined iterations specified by elements of $\mathcal{T}$.
The subset $\hat{\mathcal{I}}_b$ is sampled on line \ref{alg:cfr_abstr:choose_is_bound} in Algorithm \ref{alg:cfr_abstr} according to Algorithm \ref{alg:choose_IS}. The sampling of  $\hat{\mathcal{I}}_b$ is done so that $\forall I \in \hat{\mathcal{I}}_b\ \Phi_t(I) \in \widetilde{\mathcal{I}}^t$ and 
 
 {\small
$$|\hat{\mathcal{I}}_b| = \min\left(k_b, \sum_{I \in \widetilde{\mathcal{I}}^t}|\Phi_t^{-1}(I)|\right).$$}
I.e., the size of $\hat{\mathcal{I}}_b$ is limited by the parameter $k_b$ and the actual number of information sets in $\mathcal{I}$ that are still mapped to some abstracted information set in iteration $t$. 
Additionally, sampling of information sets on line \ref{alg:choose_IS:sampling} in Algorithm \ref{alg:choose_IS} is performed so that the probability of adding any $I \in \mathcal{I}$ such that $\Phi_t(I) \in \widetilde{\mathcal{I}}^t$ to $\hat{\mathcal{I}}_b$ is equal to 

{\small
\begin{equation}
\frac{1}{\sum_{I' \in \widetilde{\mathcal{I}}^t} |\Phi^{-1}_t(I')|}. 
\end{equation}}

After sampling $\hat{\mathcal{I}}_b$ in some $T_j \in \mathcal{T}$, the algorithm keeps track of the regrets $r_b$ accumulated in each $I \in \hat{\mathcal{I}}_b$ for $T_{j+1} - T_j$ iterations during the CFR+ update (line \ref{alg:CFR+_func:bound_update_start} to \ref{alg:CFR+_func:bound_update_end} in Algorithm \ref{alg:CFR+_func}) before again resampling the $\hat{\mathcal{I}}_b$.

Let $L_I: \{T_j, ..., T_{j+1}\} \rightarrow \mathbb{R}$ be any function for which

{\small
$$L^{T_j}_I(t) \leq \frac{\Delta_I\sqrt{|\mathcal{A}(I)|}}{\sqrt{t - T_j}}, \quad\forall t \in \{T_j+1, ..., T_{j+1}\}.$$}
The actual abstraction update in each iteration $T' \in \{T_j+1, ... T_{j+1}\}$ is done according to Algorithm \ref{alg:abstr_update_bound} in the following way. The algorithm iterates over $I \in \widetilde{\mathcal{I}}^t$ in $G^t$ (line \ref{alg:abstr_update_bound:iter1}). For all $I'' \in \Phi^{-1}_t(I) \cap \hat{\mathcal{I}}_h$
the algorithm checks the immediate regret 

{\small
$$R^{T'}_{T_j, imm}(I) = \frac{1}{T'-T_j}\max_{a \in A(I)}\sum_{t = T_j}^{T'}\left[ v_i(b^t_{I \rightarrow a}, I) - v_i(b^t, I)\right] = \frac{\max_{a \in A(I)}r^b_I[a]}{T'-T_j}.$$ }

If $R^{T'}_{T_j, imm}(I) > L^{T_j}_I(T')$
for some  $T' \in \{T_j+1, ..., T_{j+1}\}$ (line \ref{alg:abstr_update_bound:condition} in Algorithm \ref{alg:abstr_update_bound}), the algorithm disconnects $I$ in iteration $T'$ from its abstracted information set in $G^{T'}$, resets its regrets to 0 and removes it from $\hat{\mathcal{I}}_b$ (lines \ref{alg:abstr_update_bound:disconnect_start} to \ref{alg:abstr_update_bound:disconnect_end} in Algorithm \ref{alg:abstr_update_bound}).
If $I$ is disconnected, the average strategy in $I$ in all $T > T'$ is computed $\forall a \in \mathcal{A}(I)$ as  

{\small
\begin{equation}\label{eq:avg_strat_update}
\bar{b}_i^T(I, a) = \begin{cases}
\frac{1}{|\mathcal{A}(I)|} \text{, \quad if }T = T'+1,\\
\frac{2\sum_{t' = T'+1}^Tt' \cdot\pi_i^{b_i^{t'}}(I)\cdot b_i^{t'}(I, a)}{\left((T - T' - 1)^2 + T - T' - 1\right)\sum_{t' = T'+1}^T \pi_i^{b_i^{t'}}(I)} \text{, \quad otherwise.} 
\end{cases}
\end{equation}
}

The average strategy update in eq. \eqref{eq:avg_strat_update} corresponds to the average strategy update described in Section \ref{sec:cfr+}, i.e., it does a weighted average of the $b_i^t$ strategies over iterations $\{T'+1, ..., T\}$ with weight corresponding to the iteration. 

\subsubsection{Heuristic Update}
\label{sec:heur_update}

\begin{algorithm}[h!]
\small
\DontPrintSemicolon
\SetKwInOut{Input}{input}
\SetKwInOut{Output}{output}
\SetKwProg{Function}{function}{}{}

\SetKwFunction{InitAbstraction}{InitAbstraction}
\SetKwFunction{InitRegrets}{InitRegrets}
\SetKwFunction{UniformStrategy}{UniformStrategy}
\SetKwFunction{CFR}{CFR+}
\SetKwFunction{GetPlayer}{GetPlayer}
\SetKwFunction{IsTerminal}{IsTerminal}
\SetKwFunction{ChooseInformationSets}{ChooseInformationSets}
\SetKwFunction{ComputePerfectRecallRegrets}{PRRegrets}
\SetKwFunction{UpdateAbstractionHeuristic}{UpdateAbstractionForHeuristic}
\SetKwFunction{GetInformationSetFor}{GetInformationSetFor}
\SetKwFunction{GetRegretsFor}{GetRegretsFor}
\SetKwFunction{GetStrategyFor}{GetStrategyFor}
\SetKwFunction{RegretMatching}{RegretMatching+}
\SetKwFunction{UpdateAverageStrategy}{UpdateAverageStrategy}
\SetKwFunction{RemoveNegativeRegrets}{RemoveNegativeRegrets}
\SetKwFunction{Zeros}{Zeros}
\SetKwFunction{Intersection}{Intersection}
\SetKwFunction{OriginalSets}{UnabstractedSets}
\SetKwFunction{CreateNewIS}{CreateNewIS}
\SetKwFunction{InitRegret}{InitRegret}
\Function{\UpdateAbstractionHeuristic{${G}^t$, $\hat{\mathcal{I}}_h$, ${r}_h$, ${r}$, ${\bar{b}}_1$, ${\bar{b}}_2$, ${t}$}}{
    \For{${I} \in \widetilde{\mathcal{I}}^t_i$}{\label{alg:abstr_update_heur:iter1}
        $\mathcal{I}'' \gets \Phi_t^{-1}({I})$\;
        \If{\Intersection{$\mathcal{I}''$, $\hat{\mathcal{I}}_h$} $\neq \emptyset$}{\label{alg:abstr_update_heur:cond1}
            \For{$I'' \in\ $\Intersection{$\mathcal{I}''$, $\hat{\mathcal{I}}_h$}}    {    
        ${r^h_{I''}} \gets $ \GetRegretsFor{$I'', {r}_h$}\;    
                $\gamma_{I''} \gets $ set of indices of $\{a \in \mathcal{A}(I'') | {r}^h_{I''}[a] \in \left[ \max_{a' \in \mathcal{A}(I'')}\left({r}^h_{I''}[a'] - \frac{1}{5\sqrt{t}},\max_{a' \in \mathcal{A}(I'')}{r}^h_{I''}[a']\right)\right] \}$\label{alg:abstr_update_heur:action_idcs}\;
            }
            Split $\mathcal{I}''$ to subsets $\mathcal{I}''_1, ..., \mathcal{I}''_k$ \label{alg:abstr_update_heur:split_start}\;
            Each $\mathcal{I}''_j \in \{\mathcal{I}''_1, ..., \mathcal{I}''_k\}$ contains all the $I'' \in \mathcal{I}''$ with the same $\gamma_{I''}$\label{alg:abstr_update_heur:split_end}\;
            $\mathcal{I}''_{max} \gets \argmax_{\mathcal{I}''_j \in \{\mathcal{I}''_1, ..., \mathcal{I}''_k\}}|\mathcal{I}''_j |$\label{alg:abstr_update_heur:I_max}\;
            $\mathcal{I}''_f \gets \mathcal{I}''_{max} \cup \mathcal{I}'' \setminus \hat{\mathcal{I}}_h$\label{alg:abstr_update_heur:join_rest}\;
            ${G}^t.\mathcal{I} \gets {G}^t.\mathcal{I}\setminus I$\label{alg:abstr_update_heur:replace_start}\;
            ${r} \gets {r} \setminus$ \GetRegretsFor{${I}$, ${r}$}\;
            \For{$\mathcal{I}''_j \in \left(\{\mathcal{I}''_1, ..., \mathcal{I}''_k\} \cup \mathcal{I}''_f  \right) \setminus \mathcal{I}''_{max}$}{
                ${I}_{new} \gets $ \CreateNewIS{$\mathcal{I}''_j$}\;
                ${G}^t.\mathcal{I} \gets {G}^t.\mathcal{I} \cup  {I}_{new}$\;
                ${r}_{I_{new}} \gets $ \InitRegret{${I}_{new}$}\;
                ${r} \gets {r} \cup {r}_{I_{new}}$\;                
            }
      \If{\GetPlayer{${t}$} = 1}{
          ${\bar{b}}_1 \gets {\bar{b}}_1 \setminus $ \GetStrategyFor{${I}$, ${\bar{b}}_1$}\;        
        }\Else{
          ${\bar{b}}_2 \gets {\bar{b}}_2 \setminus $ \GetStrategyFor{${I}$, ${\bar{b}}_2$}\label{alg:abstr_update_heur:replace_end}\;  
      }
        }
    }
    \Return{${G}^t$}\;
}
\caption{Abstraction update for heuristic}
\label{alg:abstr_update_heur}
\end{algorithm}

In this section, we focus on the description of the heuristic update of the abstraction. 

Let $i$ be the player who's regrets are updated in iteration $t$. As a part of the abstraction update, the algorithm samples $\hat{\mathcal{I}}_h \subseteq \mathcal{I}_i$ of information sets of $G$ in $t$. $\hat{\mathcal{I}}_h$ is sampled on line \ref{alg:cfr_abstr:choose_is_heur} in Algorithm \ref{alg:cfr_abstr} according to Algorithm \ref{alg:choose_IS}. In this case, the sampling in Algorithm \ref{alg:choose_IS} is done so that $\forall I \in \hat{\mathcal{I}}_h\ \Phi_t(I) \in \widetilde{\mathcal{I}}_i^t$ and 

{\small
$$|\hat{\mathcal{I}}_h| = \min\left(k_h, \sum_{I \in \widetilde{\mathcal{I}}_i^t}|\Phi_t^{-1}(I)|\right).$$}
I.e., the size of $\hat{\mathcal{I}}_h$ is limited by the parameter $k_h$ and the actual number of information sets in $\mathcal{I}_i$ that are still mapped to some abstracted information set in iteration $t$.
Additionally, sampling of information sets on line \ref{alg:choose_IS:sampling} in Algorithm \ref{alg:choose_IS} is performed so that the probability of adding any $I \in \mathcal{I}$ such that $\Phi_t(I) \in \widetilde{\mathcal{I}}^t$ to $\hat{\mathcal{I}}_h$ is equal to 

{\small
\begin{equation}
\frac{1}{\sum_{I' \in \widetilde{\mathcal{I}}_i} |\Phi^{-1}_t(I')|}. 
\end{equation}}

The algorithm keeps track of the regrets $r_h$ in each $I \in \hat{\mathcal{I}}_h$ during the CFR+ update in iteration $t$ (lines \ref{alg:CFR+_func:heur_update_start} to \ref{alg:CFR+_func:heur_update_end} in Algorithm \ref{alg:CFR+_func}). 

The actual abstraction update is done in a following way (Algorithm \ref{alg:abstr_update_heur}).
The algorithm iterates over $I \in \widetilde{\mathcal{I}}_i^t$ in $G^t$ which contain some of the $I' \in \hat{\mathcal{I}}_h$ (lines \ref{alg:abstr_update_heur:iter1} and \ref{alg:abstr_update_heur:cond1}). For all $I'' \in \Phi^{-1}_t(I) \cap \hat{\mathcal{I}}_h$ it creates a set of action indices $\gamma_{I''}$ corresponding to actions with regret in $r_h$ at most $\frac{1}{5\sqrt{t}}$ distant from the maximum regret for $I''$ in $r_h$ (line \ref{alg:abstr_update_heur:action_idcs}). The abstraction update then splits the set $\Phi^{-1}_t(I) \cap \hat{\mathcal{I}}_h$ to largest subsets $\mathcal{I}''_1, ..., \mathcal{I}''_k$ such that $\forall  \mathcal{I}''_j \in \{\mathcal{I}''_1, ..., \mathcal{I}''_k\} \forall I''_1, I''_2 \in \mathcal{I}''_j\ \gamma_{I''_1} = \gamma_{I''_1}$ (lines \ref{alg:abstr_update_heur:split_start} to \ref{alg:abstr_update_heur:split_end}). Next, the algorithm selects $\mathcal{I}''_{max}$, the largest element of $\{\mathcal{I}''_1, ..., \mathcal{I}''_k\}$ (line \ref{alg:abstr_update_heur:I_max}) and adds all the $I'' \in \Phi^{-1}_t(I)$ which are not in $\hat{\mathcal{I}}_h$ to $\mathcal{I}''_{max}$, creating $\mathcal{I}''_f$ (line \ref{alg:abstr_update_heur:join_rest}). This is done to avoid unnecessary splits caused by not tracking regrets in $\Phi^{-1}_t(I) \setminus \hat{\mathcal{I}}_h$.
Finally, the abstracted set $I$ is replaced in $G^t$ by the set of new information sets $\mathcal{I}^t_n$ created from $\left(\{\mathcal{I}''_1, ..., \mathcal{I}''_k\} \cup \mathcal{I}''_f  \right) \setminus \mathcal{I}''_{max}$. The regrets in each $I_n \in \mathcal{I}^t_n$ are set to 0 in $r$ and the average strategies are discarded (lines \ref{alg:abstr_update_heur:replace_start} to  \ref{alg:abstr_update_heur:replace_end}). 
Finally, let $T$ be an iteration such that $T > t$. Assuming that $I_n \in \mathcal{I}^t_n$ was not split further during iterations $\{t+1, ..., T\}$, the average strategy in $I_n$ is computed $\forall a \in \mathcal{A}(I_n)$ as
{\small
\begin{equation}\label{eq:avg_strat_update_heur}
\bar{b}_i^T(I, a) = \begin{cases}
\frac{1}{|\mathcal{A}(I)|} \text{, \quad if }T = t+1,\\
\frac{2\sum_{t' = t+1}^Tt' \cdot\pi_i^{b_i^{t'}}(I)\cdot b_i^{t'}(I, a)}{\left((T - t - 1)^2 + T - t - 1\right)\sum_{t' = t+1}^T \pi_i^{b_i^{t'}}(I)} \text{, \quad otherwise.} 
\end{cases}
\end{equation}
}

The average strategy update in eq. \eqref{eq:avg_strat_update_heur} corresponds to the average strategy update described in Section \ref{sec:cfr+}, i.e., it does a weighted average of the $b_i^t$ strategies over iterations $\{t+1, ..., T\}$ with weight corresponding to the iteration.

\subsubsection{Theoretical Properties}
In this section we present the bound on the average external regret of the \IRCFR algorithm. We first derive the bound for the case where the algorithm uses only the regret bound abstraction update described in Section \ref{sec:regr_bound_update}. Since \IRCFR randomly samples information sets during the abstraction update, we provide a probabilistic bound on the average external regret of the algorithm.
We then show that the regret bound still holds when also using the heuristic update described in Section \ref{sec:heur_update}. Finally, we discuss why it is insufficient to use only the heuristic abstraction update.

Given iteration $T$, let $\mathcal{T}^T = (T_1, ..., T_k)$ be a subsequence of $\mathcal{T}$ such that $T_k$ is the largest element in $\mathcal{T}$ for which $T_k < T$. $\uptau^T$ is the sequence of  iteration counts corresponding to $\mathcal{T}^T$. Let $\uptau^T_{\bar{R}_i}$ be a sequence containing all the iteration counts $\uptau_j \in \uptau^T$, where for the corresponding $T_j, T_{j+1}$ holds that 

{\small 
\begin{equation}
L_I^{T_{j}}(T_{j+1}) < \frac{\bar{R}_i}{|\mathcal{I}_i|}, \quad \forall I \in \mathcal{I}_i,\label{eq:l_bound}
\end{equation}}

\noindent for a given regret $\bar{R}_i$. 
Next, we define $p_r(T, \bar{R}_i)$, as 

{\small
\begin{equation}
p_r(T, \bar{R}_i) = 
\begin{cases}
0, \text{if } |\uptau^T_{\bar{R}_i}| < |\mathcal{I}|,\\
1 - \left(1 - \left(\frac{k_b}{|\mathcal{I}|}\right)^{|\mathcal{I}|}\right)^{\binom{|\uptau^T_{\bar{R}_i}|}{|\mathcal{I}|}}\text{, otherwise.}\label{eq:pr}
\end{cases}
\end{equation}}

\begin{lemma}\label{lemma:p_r}
$p_r(T, \bar{R}_i) $ is the lower bound on the probability, that the abstraction in iteration $T$ allows representation of average strategy with average external regret $\bar{R}_i$ for player $i$.
\end{lemma}
\begin{proof}
$p_r$ is computed for the worst case where all the $I \in \mathcal{I}$ are in some abstracted information set in $G^1$, and where it is necessary to reconstruct the complete original game by removing all $I \in \mathcal{I}$ from their abstracted information sets one by one in a fixed order to allow representation of the average strategy with average external regret $\bar{R}_i$ for player $i$. 
The rest of the proof is conducted in the following way: First, we show that the iteration counts in $\uptau^T_{\bar{R}_i}$ are large enough to guarantee that if there is an abstracted information set preventing representation of average strategies with average external regret bellow $\bar{R}_i$, it will be split. We then provide the probability that the information set structure in a given iteration allows \IRCFR to compute strategies with average external regret $\bar{R}_i$ for player $i$ as a function of the number of iteration counts in $\uptau_{\bar{R}_i}^T$.

We know that 

{\small
\begin{equation}
\bar{R}_i^T \leq\sum_{I\in {\cal I}_i}\left(R_{i, imm}^{T}(I)\right)^+ \leq |\mathcal{I}_i|\max_{I\in {\cal I}_i}\left(R_{i, imm}^{T}(I)\right)^+, \quad\forall T. \label{eq:pr_proof_bound}
\end{equation}}

Let us assume that the information set structure of the current abstraction $G^T$ in iteration $T$ does not allow representation of average strategy with average external regret $\bar{R}_i$ and that we do not update $G^T$ any further. Then, from eq. \eqref{eq:pr_proof_bound}, there must exist $I \in \mathcal{I}_i$ such that for each iterations $T', T''$, $T'' > T'$

{\small
\begin{equation}
\left(R_{T', imm}^{T''}(I)\right)^+ > \frac{\bar{R}_i}{|\mathcal{I}_i|}. \label{eq:pr_proof_bound1}
\end{equation}}

To remove an information set $I$ from its abstracted information set as a part of the regret bound abstraction update during the sequence of iterations $(T_j, ..., T_{j+1})$, there must exists $T' \in (T_j, ..., T_{j+1})$ such that $R^{T'}_{T_j, imm}(I) > L^{T_j}_I(T')$. Therefore, from eq. \eqref{eq:pr_proof_bound1} follows, that to guarantee that $I$ preventing convergence is removed from its abstracted information set during $(T_j,..., T_{j+1})$, it needs to hold that $L_I^{T_{j}}(T_{j+1}) < \frac{\bar{R}_i}{|\mathcal{I}_i|}$. Hence, from eq. \eqref{eq:l_bound} in definition of $\uptau^T_{\bar{R}_i}$ follows that $\uptau^T_{\bar{R}_i}$ contains only the iteration counts that guarantee that such information set is removed from its abstracted information set.

Now we turn to the formula in eq. \eqref{eq:pr}.
 The first case of the piecewise function in eq. \eqref{eq:pr} handles the situation where there are not enough iteration counts $\uptau_j \in \uptau_{\bar{R}_i}^T$ large enough to guarantee the required $|\mathcal{I}|$ splits. The second case computes the probability that given $|\uptau^T_{\bar{R}_i}|$ samples we correctly sample the required sequence of the length $|\mathcal{I}|$. The second case has the following intuition. $\binom{|\uptau^T_{\bar{R}_i}|}{|\mathcal{I}|}$ is the number of possibilities how to choose a subsequence of the length $|\mathcal{I}|$ from a sequence of length $|\uptau^T_{\bar{R}_i}|$. $\left(\frac{k_b}{|\mathcal{I}|}\right)^{|\mathcal{I}|}$ is the probability that a specific sequence of all information sets, i.e., a sequence of length $|\mathcal{I}|$, is sampled when we sample $k_b$ from $|\mathcal{I}|$ elements with a uniform probability. Hence, $\left(1 - \left(\frac{k_b}{|\mathcal{I}|}\right)^{|\mathcal{I}|}\right)^{\binom{|\uptau^T_{\bar{R}_i}|}{|\mathcal{I}|}}$ is the probability that the sequence of length $|\mathcal{I}|$ is not sampled in $\binom{|\uptau^T_{\bar{R}_i}|}{|\mathcal{I}|}$ attempts.

Since $p_r(T, \bar{R}_i)$ is computed assuming that there is the worst case number of splits necessary and that it takes the maximum possible number of iterations to split each information set, it is a lower bound on the actual probability that the abstraction in iteration $T$ allows representation of average strategy with average external regret $\bar{R}_i$ for player $i$.
\end{proof}

\begin{lemma}\label{lemma:ircfr_bound}
The average external regret of the \IRCFR is bounded in the following way

{\small
\begin{equation}
\frac{R_i^T}{T}  \leq B_i(T) = \frac{1}{T}\left(\Delta |\mathcal{I}_i|\sqrt{A_{max}}\sqrt{T} +  \Delta |\mathcal{I}_i| T_{lim}\right), \label{eq:regr}
\end{equation}}

with probability at least $p_r(T_{lim}, B_i(T))$.
\end{lemma}
\begin{proof}
$B_i(T)$ decomposes the bound on the average external regret to two parts. First, in iterations $\{1, ..., T_{lim}\}$ it assumes that the structure of the abstraction prevents the algorithm from convergence. Second, it assumes that the abstraction is updated so that it allows convergence of CFR+ in iterations $\{T_{lim} + 1, ..., T\}$. Hence, the regret in all $I \in \mathcal{I}_i$ has the following property:

{\small
\begin{equation}
\max_{a\in \mathcal{A}(I)}R_i^T(I, a) \leq \Delta T_{lim} + \Delta \sqrt{|\mathcal{A}(I)|} \sqrt{T}.\label{eq:ircfr_proof_I}
\end{equation}}

This bound holds since $\Delta \sqrt{|\mathcal{A}(I)|} \sqrt{T}$ is the bound on $\max_{a\in \mathcal{A}(I)}R_i^T(I, a)$ in $I$ provided by regret matching$^+$. Additionally, since the regret matching$^+$ in $I$ in each $t \in \{1, ..., T_{lim}\}$ uses regrets computed for information set $\phi_t(I)$ and not directly for $I$, it can lead to arbitrarily bad outcomes with respect to the utility structure of the solved game (as we assume that the structure of the abstraction prevents convergence of the algorithm in these iterations). The $\Delta T_{lim}$ corresponds to the worst case regret that can be accumulated for each action in $I$ during the first $T_{lim}$ iterations.
From eqs. \eqref{eq:ircfr_proof_I} and \eqref{eq:bound1} follows that 

{\small
\begin{equation}
\frac{R_i^T}{T}  \leq B_i(T) = \frac{1}{T}\left(\Delta |\mathcal{I}_i|\sqrt{A_{max}}\sqrt{T} +  \Delta |\mathcal{I}_i| T_{lim}\right).\nonumber
\end{equation}}

Since we assume that the abstraction in iteration $T_{lim}$ of \IRCFR allows computing the regret $B_i(T)$, this bound holds with probability at least $p_r(T_{lim}, B_i(T))$ (Lemma \ref{lemma:p_r}).
 \end{proof}
 
Finally, we provide the bound on the external regret of \IRCFR as a function of the probability that the bound holds.

\begin{theorem}\label{thm:bound}
Let 

{\small
\begin{equation}
\alpha(\delta) =  \sqrt[3]{2^{\sqrt[|\mathcal{I}|]{|\mathcal{I}|!\log_{1-\left(\frac{k_b}{|\mathcal{I}|}\right) ^{|\mathcal{I}|}}(\delta)}+ |\mathcal{I}| + 3}  A_{max}}.
\end{equation}}

In each $T \geq (\alpha(\delta) + 1)^3$, the average external regret of \IRCFR is bounded in the following way

{\small
\begin{equation}
\frac{R_i^T}{T} \leq \Delta |\mathcal{I}_i|\frac{\sqrt{A_{max}}}{\sqrt{T}} + \Delta |\mathcal{I}_i|\frac{\alpha(\delta)}{\sqrt[3]{T}} + \Delta |\mathcal{I}_i|\frac{1}{T}\in \mathcal{O}\left(\frac{1}{\sqrt[3]{T}}\right).\label{eq:thm_bound}
\end{equation}
}

with probability $1 - \delta$.
\end{theorem}
\begin{proof}
The bound in eq. \eqref{eq:thm_bound} is created by substituting 

{\small
\begin{equation}
T_{lim} =  \raisebox{1.8ex}{\Bigg\lceil}\sqrt[3]{2^{\sqrt[|\mathcal{I}|]{|\mathcal{I}|!\log_{1-\left(\frac{k_b}{|\mathcal{I}|}\right) ^{|\mathcal{I}|}}(\delta)}+ |\mathcal{I}| + 3} A_{max}}\cdot T^{\frac{2}{3}}\raisebox{1.8ex}{\Bigg\rceil},
\end{equation}}

to the bound from Lemma \ref{lemma:ircfr_bound}.
Hence, we need to show that choosing this $T_{lim}$ guarantees that 

{\small
\begin{equation}
p_r(T_{lim}, B_i(T)) \geq 1 - \delta.\label{eq:prob_bound}
\end{equation}}

The proof is conducted in the following way.
First, in Lemma \ref{lemma:taur} we show the size of $\uptau^{T_{lim}}_{B_i(T)}$ sufficient to guarantee that inequality \eqref{eq:prob_bound} holds. In Lemma \ref{lemma:tbt} we derive the lower bound $T_{B_i(T)}$ on each element of $\uptau^{T_{lim}}_{B_i(T)}$. $T_{B_i(T)}$ is the number of iterations sufficient to guarantee that some information set preventing the abstraction from allowing representation of average strategy with average external regret $B_i(T)$ is split during the regret bound abstraction update. Finally, in Lemma \ref{lemma:Tlim} we derive the $T_{lim}$ that implies sufficient number of elements in $\uptau^{T_{lim}}_{B_i(T)}$.
\begin{lemma}\label{lemma:taur}
{\small
\begin{equation}
|\uptau^{T_{lim}}_{B_i(T)}| > \sqrt[|\mathcal{I}|]{|\mathcal{I}|!\log_{1 - \left(\frac{k_b}{\mathcal{I}} \right)^{|\mathcal{I}|}} (\delta)} + |\mathcal{I}|
\end{equation}}

guarantees that $p_r(T_{lim}, B_i(T)) \geq 1 - \delta$.
\end{lemma}
\begin{proof}
{\small
 \begin{align}
&&|\uptau^{T_{lim}}_{B_i(T)}| &>  \sqrt[|\mathcal{I}|]{|\mathcal{I}|!\log_{1 - \left(\frac{k_b}{\mathcal{I}} \right)^{|\mathcal{I}|}} (\delta)} + |\mathcal{I}|\\
&\implies&\left(|\uptau^{T_{lim}}_{B_i(T)}| - |\mathcal{I}|\right)^{|\mathcal{I}|} &> |\mathcal{I}|!\log_{1 - \left(\frac{k_b}{\mathcal{I}} \right)^{|\mathcal{I}|}} (\delta)\\
&\implies&\prod_{j = |\uptau^{T_{lim}}_{B_i(T)}|-|\mathcal{I}| + 1}^{|\uptau^{T_{lim}}_{B_i(T)}|}j &>  |\mathcal{I}|!\log_{1 - \left(\frac{k_b}{\mathcal{I}} \right)^{|\mathcal{I}|}} (\delta)\\
&\implies&\frac{\prod_{j = |\uptau^{T_{lim}}_{B_i(T)}|-|\mathcal{I}| + 1}^{|\uptau^{T_{lim}}_{B_i(T)}|}j}{|\mathcal{I}|!} &>  \log_{1 - \left(\frac{k_b}{\mathcal{I}} \right)^{|\mathcal{I}|}} (\delta)\\
&\implies&\frac{|\uptau^{T_{lim}}_{B_i(T)}|!}{(|\uptau^{T_{lim}}_{B_i(T)}| - |\mathcal{I}|)!|\mathcal{I}|!} &> \log_{1 - \left(\frac{k_b}{\mathcal{I}} \right)^{|\mathcal{I}|}} (\delta)\\
&\implies&\binom{|\uptau^{T_{lim}}_{B_i(T)}|}{|\mathcal{I}|} &> \log_{1 - \left(\frac{k_b}{\mathcal{I}} \right)^{|\mathcal{I}|}} (\delta)\\
&\implies&\left(1 - \left(\frac{k_b}{\mathcal{I}} \right)^{|\mathcal{I}|}\right)^{\binom{|\uptau^{T_{lim}}_{B_i(T)}|}{|\mathcal{I}|}} &< \delta\\
&\implies&1 - \left(1 - \left(\frac{k_b}{\mathcal{I}} \right)^{|\mathcal{I}|}\right)^{\binom{|\uptau^{T_{lim}}_{B_i(T)}|}{|\mathcal{I}|}} &> 1 - \delta
\end{align}}
\end{proof}

\begin{lemma}\label{lemma:tbt}
{\small
\begin{equation}
T_{B_i(T)} = \left\lceil \left[\frac{T\sqrt{A_{max}}}{T_{lim}+\sqrt{A_{max}}\sqrt{T}}\right]^2\right\rceil\nonumber
\end{equation}}

is a sufficient number of iterations to guarantee that the regret bound abstraction update splits some information set preventing the abstraction from allowing representation of average strategy with average external regret $B_i(T)$.
\end{lemma}
\begin{proof}
The regret bound abstraction update removes an information set $I \in \mathcal{I}$ from its abstracted information set during the sequence of iterations $(T_j, ..., T_{j+1})$, when there exists $T' \in (T_j, ..., T_{j+1})$ such that $R^{T'}_{T_j, imm}(I) > L^{T_j}_I(T')$. As discussed in the proof of Lemma \ref{lemma:p_r}, to guarantee that $I$ is removed from its abstracted information set if it prevents the regret bellow $B_i(T)$, we need to make sure that 

{\small 
\begin{equation}
L_I^{T_{j}}(T_{j+1}) < \frac{B_i(T)}{|\mathcal{I}_i|}.\label{eq:T_R_cond}
\end{equation}}

 Since

 {\small
 $$L^{T_j}_I(T') \leq\frac{\Delta_I\sqrt{|\mathcal{A}(I)|}}{\sqrt{T' - T_j}}, \quad\forall T' \in \{T_j+1, ..., T_{j+1}\}, \forall I \in \mathcal{I}_i,$$}
 it follows that

{\small
\begin{equation}
\max_{I \in \mathcal{I}_i}L_I^{T_{j}}(T_{j+1}) \leq \frac{\Delta\sqrt{A_{max}}}{\sqrt{T_{j+1} - T_j}}.\label{eq:T_R_cond1}
\end{equation}}
And so from eqs. \eqref{eq:T_R_cond} and \eqref{eq:T_R_cond1}, it is sufficient for $T_{B_i(T)}$ to satisfy

{\small 
$$\frac{\Delta\sqrt{A_{max}}}{\sqrt{T_{B_i(T)}}} < \frac{B_i(T)}{|\mathcal{I}_i|}.$$}

Smallest $T_{B_i(T)}$ satisfying this inequality is

{\small
\begin{equation}
T_{B_i(T)} = \left\lceil \left[\frac{T\sqrt{A_{max}}}{T_{lim}+\sqrt{A_{max}}\sqrt{T}}\right]^2\right\rceil.
\end{equation}}
\end{proof}

Hence, all elements in $\uptau^{T_{lim}}_{B_i(T)}$ must be greater or equal to $T_{B_i(T)}$ to make sure that the structure of the abstraction in iteration $T_{lim}$ allows representation of average strategy with average external regret $B_i(T)$ for player $i$ with sufficient probability.

The number of elements in $\uptau^{T_{lim}}$ is at least $\log_2(T_{lim} + 1) - 1$, since $\forall j \in \{1, ..., | \uptau^{T_{lim}}|\}\  \uptau_j = 2^{j-1}$. From Lemma \ref{lemma:taur} we know that we need the last 

{\small
$$
\left\lceil \sqrt[|\mathcal{I}|]{|\mathcal{I}|!\log_{1 - \left(\frac{k_b}{\mathcal{I}} \right)^{|\mathcal{I}|}} (\delta)} + |\mathcal{I}| \right\rceil
$$}

elements of $\uptau^{T_{lim}}$ (which form $\uptau^{T_{lim}}_{B_i(T)}$) to be higher or equall to $T_{B_i(T)}$.

\begin{lemma}\label{lemma:Tlim}
When using 

{\small
$$
T_{lim} =  \raisebox{1.8ex}{\Bigg\lceil}\sqrt[3]{2^{\sqrt[|\mathcal{I}|]{|\mathcal{I}|!\log_{1-\left(\frac{k_b}{|\mathcal{I}|}\right) ^{|\mathcal{I}|}}(\delta)}+ |\mathcal{I}| + 3} A_{max}}\cdot T^{\frac{2}{3}}\raisebox{1.8ex}{\Bigg\rceil},
$$}

the last 

{\small
$$
\left\lceil\sqrt[|\mathcal{I}|]{|\mathcal{I}|!\log_{1 - \left(\frac{k_b}{\mathcal{I}} \right)^{|\mathcal{I}|}} (\delta)} + |\mathcal{I}|\right\rceil
$$}

elements in $\uptau^{T_{lim}}$ are higher or equal to $T_{B_i(T)}$. 
\end{lemma}
\begin{proof}
{\small
\begin{align}
&\quad &&T_{lim} = \raisebox{1.8ex}{\Bigg\lceil}\sqrt[3]{2^{\sqrt[|\mathcal{I}|]{|\mathcal{I}|!\log_{1-\left(\frac{k_b}{|\mathcal{I}|}\right) ^{|\mathcal{I}|}}(\delta)}+ |\mathcal{I}| + 3} A_{max}}\cdot T^{\frac{2}{3}}\raisebox{1.8ex}{\Bigg\rceil}\\
&\implies &&T_{lim} \geq \sqrt[3]{2^{\sqrt[|\mathcal{I}|]{|\mathcal{I}|!\log_{1-\left(\frac{k_b}{|\mathcal{I}|}\right) ^{|\mathcal{I}|}}(\delta)}+ |\mathcal{I}| + 3}  A_{max}} \cdot T^{\frac{2}{3}}\\
&\implies &&T_{lim}^3 \geq 2\cdot2^{\sqrt[|\mathcal{I}|]{|\mathcal{I}|!\log_{1 - \left(\frac{k_b}{\mathcal{I}} \right)^{|\mathcal{I}| + 1}} (\delta)} + |\mathcal{I}| + 2}A_{max}T^2\\
&\implies &&(T_{lim} +1)\left(T_{lim} + \sqrt{A_{max}}\sqrt{T}\right)^2  \geq\\
&\quad &&\qquad\geq2^{\sqrt[|\mathcal{I}|]{|\mathcal{I}|!\log_{1 - \left(\frac{k_b}{\mathcal{I}} \right)^{|\mathcal{I}|}} (\delta)} + |\mathcal{I}| +2} 2A_{max}T^2\\
&\implies &&\frac{T_{lim} + 1}{2^{\sqrt[|\mathcal{I}|]{|\mathcal{I}|!\log_{1 - \left(\frac{k_b}{\mathcal{I}} \right)^{|\mathcal{I}|}} (\delta)} + |\mathcal{I}| + 2}} \geq \frac{2A_{max}T^2}{\left(T_{lim} + \sqrt{A_{max}}\sqrt{T}\right)^2} \\
&\implies &&\frac{T_{lim} + 1}{2^{\sqrt[|\mathcal{I}|]{|\mathcal{I}|!\log_{1 - \left(\frac{k_b}{\mathcal{I}} \right)^{|\mathcal{I}|}} (\delta)} + |\mathcal{I}| + 2}} \geq \left\lceil\frac{A_{max}T^2}{\left(T_{lim} + \sqrt{A_{max}}\sqrt{T}\right)^2}\right\rceil\\
&\implies &&2^{\log_2(T_{lim} +1) -2 - \sqrt[|\mathcal{I}|]{|\mathcal{I}|!\log_{1 - \left(\frac{k_b}{\mathcal{I}} \right)^{|\mathcal{I}|}} (\delta)} - |\mathcal{I}|} \geq T_{B_i(T)}\\
&\implies &&2^{\log_2(T_{lim} +1) -1 - \left\lceil\sqrt[|\mathcal{I}|]{|\mathcal{I}|!\log_{1 - \left(\frac{k_b}{\mathcal{I}} \right)^{|\mathcal{I}|}} (\delta)} + |\mathcal{I}|\right\rceil} \geq T_{B_i(T)}\label{eq:proof_finish}
\end{align}}

Eq. \eqref{eq:proof_finish} states that last

{\small $$
\left\lceil \sqrt[|\mathcal{I}|]{|\mathcal{I}|!\log_{1 - \left(\frac{k_b}{\mathcal{I}} \right)^{|\mathcal{I}|}} (\delta)} + |\mathcal{I}|\right\rceil
$$}

elements of $\uptau^{T_{lim}}$ are at least $T_{B_i(T)}$, since the elements in $\uptau^{T_{lim}}$ are increasing.
\end{proof}

Hence, using $T_{lim}$ from Lemma \ref{lemma:Tlim} guarantees that $p_r(T_{lim}, B_i(T)) > 1-\delta$ from Lemma \ref{lemma:taur}.
When substituting this $T_{lim}$ to the bound from Lemma \ref{lemma:ircfr_bound}, i.e., to 

{\small
\begin{equation}
\frac{R_i^T}{T}  \leq \frac{1}{T}\left( \Delta |\mathcal{I}_i|\sqrt{A_{max}}\sqrt{T}  +\Delta |\mathcal{I}_i|T_{lim}\right),
\end{equation}}

we get

{\small
\begin{align}
\frac{R_i^T}{T} &\leq \frac{1}{T}\left( \Delta |\mathcal{I}_i|\sqrt{A_{max}}\sqrt{T}  +\Delta |\mathcal{I}_i|\left\lceil\alpha(\delta)T^{\frac{2}{3}} \right\rceil \right)\\ 
&\leq \Delta |\mathcal{I}_i|\frac{\sqrt{A_{max}}}{\sqrt{T}} + \Delta|\mathcal{I}_i|\frac{\alpha(\delta)}{\sqrt[3]{T}} + \Delta|\mathcal{I}_i|\frac{1}{T} \\
&\in \mathcal{O}\left(\frac{1}{\sqrt[3]{T}}\right).
\end{align}}

Finally, we need to show that using $T \geq (\alpha(\delta) + 1)^3$ guarantees that $T\geq T_{lim}$:

{\small
\begin{align}
&&T &\geq (\alpha(\delta) + 1)^3\\
&\implies &\sqrt[3]{T} &\geq (\alpha(\delta) + 1)\\
&\implies &T &\geq (\alpha(\delta) + 1)T^{\frac{2}{3}}\\
&\implies &T &\geq \left\lceil\alpha(\delta)T^{\frac{2}{3}}\right\rceil\\
&\implies &T &\geq T_{lim}.
\end{align}}
\end{proof}

\noindent\textbf{Combining Regret Bound and Heuristic Abstraction Update.}

When using the heuristic abstraction update in combination with the regret bound update, Theorem \ref{thm:bound} still holds, since in the worst case the heuristic update does not perform any splits and all the information set splits need to be performed by the regret bound update. Adding the heuristic abstraction update can only reduce the number of iterations needed by the algorithm.

\noindent\textbf{Counterexample for Heuristic Abstraction Update.}
\begin{table}[t]
\centering
\caption{Expected utilities after actions in $I'$ and $I''$
\label{tab:utils}}{
\begin{tabular}{l|c|c}
 & $I'$ & $I''$ \\
\hline 
$c$ & 0 & 0\\
$d$ & 1  & 10\\
\end{tabular}
\hspace{0.5cm}
\begin{tabular}{l|c|c}
 & $I'$ & $I''$ \\
\hline 
$c$ & 10 & 1\\
$d$ & 0  & 0\\
\end{tabular}
}
\end{table}

Here we show the intuition why using only the heuristic abstraction update does not guarantee convergence of \IRCFR to the NE of the solved game $G$.

Let $I$ be an information set in $G^t$ such that $\phi^{-1}_t(I) = \{I', I''\}$. Let $\mathcal{A}^t(I) = \{c, d\}$ and $\mathcal{A}(I') = \{x, y\}$, $\mathcal{A}(I'') = \{v, w\}$. Let us assume that the expected values for actions $c$ and $d$ in $I'$ and $I''$ oscillate between values depicted in Table \ref{tab:utils} (left) and (right). Hence, when computing the regret $r^h_{\hat{I}}(a) = v_i(b^t_{\hat{I} \rightarrow a}, \hat{I}) - v_i(b^t, \hat{I})$ for $\hat{I} \in \{I', I''\}$ and $a \in \mathcal{A}(\hat{I})$ during the heuristic abstraction update in iteration $t$ where the expected utilities from Table \ref{tab:utils}(left) occur, we get

{\small
\begin{align}
r^h_{I'}(c) &= - b^t_i(I, d)\\
r^h_{I'}(d) &= 1 - b^t_i(I, d)\\
r^h_{I''}(c) &= - 10b^t_i(I, d)\\
r^h_{I''}(d) &= 10-10b^t_i(I, d).
\end{align}}

When computing the regret $r^h$ for heuristic update in iteration $t'$ where the expected utilities from Table \ref{tab:utils} (right) occur, we get

{\small
\begin{align}
r^h_{I'}(c) &= 10-10b^{t'}_i(I, c)\\
r^h_{I'}(d) &= -10b^{t'}_i(I, c)\\
r^h_{I''}(c) &= 1 - b^{t'}_i(I, c)\\
r^h_{I''}(d) &= -b^{t'}_i(I, c).
\end{align}}

The actions corresponding to $d$ are always preferred in both $I'$ and $I''$ when the utilities are given by Table \ref{tab:utils} (left) by the same margin as the actions corresponding to $c$ are always prefer in both $I'$ and $I''$ when the utilities are given by Table \ref{tab:utils} (right). Hence, $\gamma_{I'} = \gamma_{I''}$ (i.e., the indices of actions with the largest regret are always equal) in all iterations $t''$ for any $b_i^{t''}$. Therefore, $I$ is never split during the heuristic update. However, \IRCFR in $I$ converges to a uniform strategy with the average expected value $2.75$ in both $I'$ and $I''$ since the regret updates lead to equal regrets for both $c$ and $d$, while CFR+ converges to strategy $b_i(I', x) = 1$ and $b_i(I'', w) = 1$ with the average expected value $5$ in both $I'$ and $I''$.

\subsubsection{Storing the Information Set Map}\label{sec:ircfr_ismap}
In this section we discuss the details of storing the mapping of information sets $\mathcal{I}$ of $G$ to $\mathcal{I}^t$ of $G^t$ in \IRCFR. 

\textbf{Initial Abstraction Storage.} The initial abstraction is identical to the initial abstraction used by FPIRA, hence the mapping is stored without using any additional memory as described in Section \ref{sec:fpira_abstraction_storage}.

\textbf{Regret Bound Update.} The regret bound update resulting in $G^t$ only removes information sets from $G$ from the abstracted information sets in $G^{t-1}$. Hence similarly to Section \ref{sec:fpira_abstraction_storage} Case 2, we use the mapping used in $G^{t-1}$ in all information sets of $G$ mapped to abstracted information sets in $G^t$, and the information set structure provided by the domain description in the rest. Therefore, no additional memory is needed to store the mapping for $G^t$ compared to the mapping required in $G^{t-1}$.

\textbf{Heuristic Update.} The heuristic update resulting in $G^t$, can split existing abstracted information set $I$ in $G^{t-1}$ to a set of information sets $\mathcal{I}'$ such that some $I' \in \mathcal{I}'$ are still abstracted information sets. Hence in this case we need to store the new mapping as described in Section \ref{sec:fpira_abstraction_storage} Case 1.

\subsubsection{Memory and Time Efficiency}\label{sec:ircfr_memory}
The CFR+ requires storing the regret and average strategy for each $I \in \mathcal{I}$ and $a \in \mathcal{A}(I)$ in $G$. Hence, when storing the regrets and average strategy in $G^t$ in iteration $t$ in \IRCFR, the algorithm achieves memory savings directly proportional to the reduction in the number of information sets in $G^t$ compared to the number of information sets in $G$. Additional memory used to store the current abstraction mapping is discussed in Section \ref{sec:ircfr_ismap}. Finally, \IRCFR stores additional regrets in $k_b + k_h$ information sets of $G$ in every iteration for the abstraction update. Since $k_b$ and $k_h$ are parameters of the \IRCFR, this memory can be adjusted as necessary. In Section \ref{sec:experiments} we empirically demonstrate that memory used to store the information set mapping is small and so it does not substantially affect the memory efficiency of \IRCFR. Additionally, we show how the different $k_h$ and $k_b$ affect the convergence of \IRCFR.

The iteration of \IRCFR consists of a tree traversal in function $\mathsf{ComputeRegrets}$ (Algorithm \ref{alg:CFR+_func}) and the two updates of the abstraction. The update of regrets $r$ as a part of $\mathsf{ComputeRegrets}$ takes the same time as standard CFR+ iteration applied to $G$. Additionally, in $\mathsf{ComputeRegrets}$, \IRCFR updates $r_h$ and $r_b$. Since there is no additional tree traversal necessary to obtain the values for the update of $r_h$ and $r_b$, the update takes time proportional to the number of states in $\hat{\mathcal{I}}_b \cup \hat{\mathcal{I}}_h$. In the worst case, the abstraction updates take time proportional to $|\mathcal{I}|$. 
In Section \ref{sec:experiments} we provide experimental evaluation of the runtime of \IRCFR. 

{
\paragraph{Scalability limitations}
The computation time of \IRCFR is given by the time required for traversing all histories in the unabstracted game tree, in order to compute the regrets. This computation is slightly simpler than the computation of best response required by FPIRA, but it is still linear in the number of histories in the game. The extra overhead per iteration, compared to CFR+, is negligible, so we can estimate the scalability limitations of a domain-specific implementation of the algorithm based on existing results with CFR+. Solving Limit Texas Hold'em with approximately $10^{18}$ reported in \cite{bowling2015heads} required 900 core years on a large cluster. Even though an iteration of \IRCFR is not substantially more expensive than plain CFR+, \IRCFR needs to perform more iterations, since the structure of the game is changing, which can waste a part of the previous computation. Therefore, we expect running \IRCFR on the same game would take at least an order of magnitude more computation than running plain CFR+. This is also supported by the experiments in Section~\ref{sec:experiments}.
}

\section{Experiments}\label{sec:experiments}

In this section, we present the experimental evaluation of FPIRA and \IRCFR.
{
Since both algorithms do full traversals of the game tree of the original game, they are not expected to solve the games faster than CFR+. However, they should require substantially less memory in the solution process and produce smaller strategies in the end. We focus mainly on this aspect of the algorithms in the evaluation.
}

First, we briefly describe the Double Oracle algorithm (DOEFG)~\cite{bosansky2014} which will be used as a baseline in the experiments and introduce domains used for the experimental evaluation. 
Next, we demonstrate the convergence of \IRCFR compared to CFR+ and explain our choice of values of the $k_h$ and $k_b$ parameters in the rest of the experimental evaluation (these parameters control the memory used by \IRCFR to update the abstraction, see Sections \ref{sec:regr_bound_update} and \ref{sec:heur_update} for more details).  
We follow with the comparison of the memory requirements of FPIRA, \IRCFR and the DOEFG as a function of the exploitability of the resulting strategies. 
Finally, we provide a comparison of the runtime of FPIRA, \IRCFR, and DOEFG. 
When reporting the results for \IRCFR, B$x$H$y$\IRCFR stands for \IRCFR where $k_b = x$ and $k_h = y$. 

\subsection{Experimental Settings}
The experiments were performed using domain independent implementation of all algorithms in Java\footnote{The implementation is available at http://jones.felk.cvut.cz/repo/gtlibrary. \IRCFR and FPIRA are available in package /src/cz/agents/gtlibrary/experimental/imperfectrecall/\\automatedabstractions/memeff.}. DOEFG uses IBM CPLEX 12.6. to solve the underlying linear programs. Both FPIRA and \IRCFR used the initial abstraction built as described in Section \ref{sec:init_abstr}. Both CFR+ and \IRCFR used the delay of 100 iterations during the average strategy update, i.e., the average strategies are not computed for the first 100 iterations, and the average strategies in iteration $T > 100$ are computed only from current strategies in iterations $\{101, ..., T\}$. And finally, the $L_I^{T_j}$ functions in \IRCFR, used during the regret bound update (see Section \ref{sec:regr_bound_update}), were set to 

{\small
$$L_I^{T_j}(T') = \frac{\Delta_I \sqrt{\mathcal{A}(I)}}{100\sqrt{T'}}, \forall I \in \mathcal{I}, \forall T_j, T_{j+1}, \forall T' \in \{T_j + 1, ..., T_{j+1}\}.$$}

\subsection{Double Oracle Algorithm}
The Double oracle algorithm for solving perfect recall EFGs (DOEFG, \cite{bosansky2014-jair}) is an adaptation of column/constraint generation techniques for EFGs.
The main idea of DOEFG is to create a restricted game where only a subset of actions is allowed to be played by players. The algorithm then incrementally expands this restricted game by allowing new actions.
The restricted game is solved as a standard zero-sum extensive-form game using the sequence-form LP. The expansion of the restricted game is performed using best response algorithms that search the original unrestricted game to find new sequences to add to the restricted game for each player. 
The algorithm terminates when the best response calculated in the unrestricted game provides no improvement to the solution of the restricted game for either of the players.

DOEFG uses two main ideas: (1) the algorithm assumes that players play some pure default strategy outside of the restricted game (e.g., playing the first action in each information set given some ordering), (2) temporary utility values are assigned to leaves in the restricted game that correspond to inner nodes in the original unrestricted game (so-called temporary leaves), which form an upper bound on the expected utility.

We chose the DOEFG as a baseline for comparison since it is the state-of-the-art domain-independent algorithm for solving large EFGs with prohibitive memory requirements. Furthermore, the conceptual idea of DOEFG is similar to FPIRA and \IRCFR since DOEFG also creates a smaller version of the original game and repeatedly refines it until the desired approximation of the Nash equilibrium of the original game is found. Our algorithms, however, exploit a completely different type of sparseness than DOEFG. 

\subsection{Domains}
Here we introduce the zero-sum domains used in the experimental evaluation. These domains were chosen for their diverse structure both in the utility and the source of imperfect information. 
The players perfectly observe the actions of their opponent in poker, only the actions of the chance player at the start of the game (corresponding to card deal) are hidden. In II-Goofspiel and Graph pursuit, on the other hand, the imperfect information is created by partial observability of the moves of the opponent through the whole game. Furthermore, in II-Goofspiel and poker, the utility of players cumulates during the play (the chips and the value of the cards won), while in Graph pursuit the utility depends solely on whether the attacker is intercepted or not or whether he reaches his goal.
\subsubsection{Poker}
As a first domain, we use a two-player poker, which is commonly used as a benchmark in imperfect-information game solving~\cite{rubin2011computer}. We use a version of poker with a deck of cards with 4  card types 3 cards per type. There are two rounds. In the first round, each player places an ante of 1 chip in the pot and receives a single private card. A round of betting follows. Every player can bet from a limited set of allowed values or check. After a bet, the other player can raise, again choosing the value from a limited set, call or forfeit the game by folding. The number of consecutive raises is limited. A shared card is dealt after one of the players calls or after both players check. Another round of betting takes place with identical rules. The player with the highest pair wins. If none of the players has a pair, the player with the highest card wins. We create different poker domains by varying the number of bets $b$, the number of raises $r$ and the number of consecutive raises allowed $c$. We refer to these instances as P$brc$, e.g., P234 stand for a poker which uses two possible values of bets, 3 values of raises and allows 4 consecutive raises.
\subsubsection{II-Goofspiel}
II-Goofspiel is a modification of the Goofspiel game \cite{Ross71Goofspiel} which is commonly used as a benchmark domain (see, e.g., \cite{lisy2015online,lanctot2013monte}). Similarly to Goofspiel, II-Goofspiel is a card game with three identical packs of cards, two for players and one randomly shuffled and placed in the middle. In our variant, both players know the order of the cards in the middle pack. The game proceeds in rounds. Every round starts by revealing the top card of the middle pack. Both players proceed to bet simultaneously on this card using their own cards. The cards used to bet are discarded, and the player with the higher value of the card used to bet wins the middle card. After the end of the game, each player gets utility equal to the difference between the points collected by him and the number of points collected by his opponent. The players do not observe the bet of their opponent. Instead, they learn whether they have won, lost, or if there was a tie caused by both players using cards with equal value. We change the number of cards in all 3 decks, by GS$x$, we refer to the II-Goofspiel where each deck has $x$ cards.

\subsubsection{Graph Pursuit}
\begin{figure}[t]
\centering
\includegraphics[width=4cm]{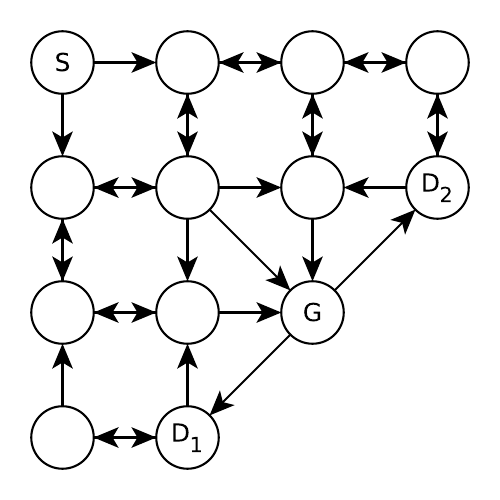}
\caption{Graph used in the Graph pursuit domain.}
\label{fig:graph_pursuit}
\end{figure}

Graph pursuit is a game played between the defender and the attacker on the graph depicted in Figure \ref{fig:graph_pursuit}. 
The attacker starts in the node labeled $S$ and tries to reach the node labeled $G$. The defender controls two units which start in nodes $D_1$ and $D_2$. The players move simultaneously and are forced to move their units each round. Both the attacker and defender only observe the content of the nodes with distance less or equal to 2 from the current node occupied by any of their units. The attacker gets utility 2 for reaching the goal $G$. If the attacker is caught by crossing the same edge as any of the units of the defender or by moving to a node occupied by the defender, he obtains the utility -1 and the game ends. If a given number of moves occurs without any of the previous events, the game is a tie, and both players get 0. We create different versions of Graph pursuit by changing the limit on the number of moves. By GP$x$ we denote Graph pursuit where there are $x$ moves of each player allowed.

\subsection{Convergence of \IRCFR}
\label{sec:experiments:parameters}
In this section, we provide the experiments showing the convergence of \IRCFR with varying $k_h$ and $k_b$ parameters compared to CFR+ applied directly to the unabstracted game. Additionally, we justify our choice of values of the $k_h$ and $k_b$ parameters in the rest of the experimental evaluation.

\begin{figure}[t]
\begin{center}\includegraphics[height=0.175cm]{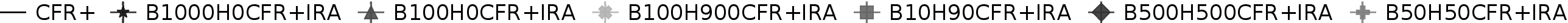}\end{center}
\vspace{-0.3cm}
\includegraphics[width=0.33\textwidth]{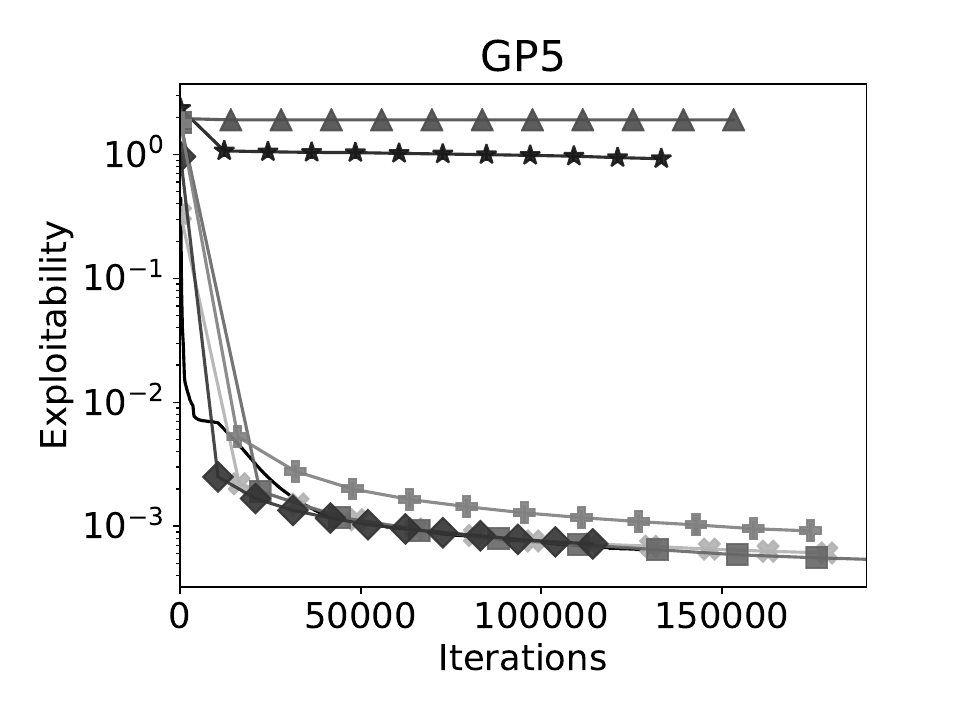}\includegraphics[width=0.33\textwidth]{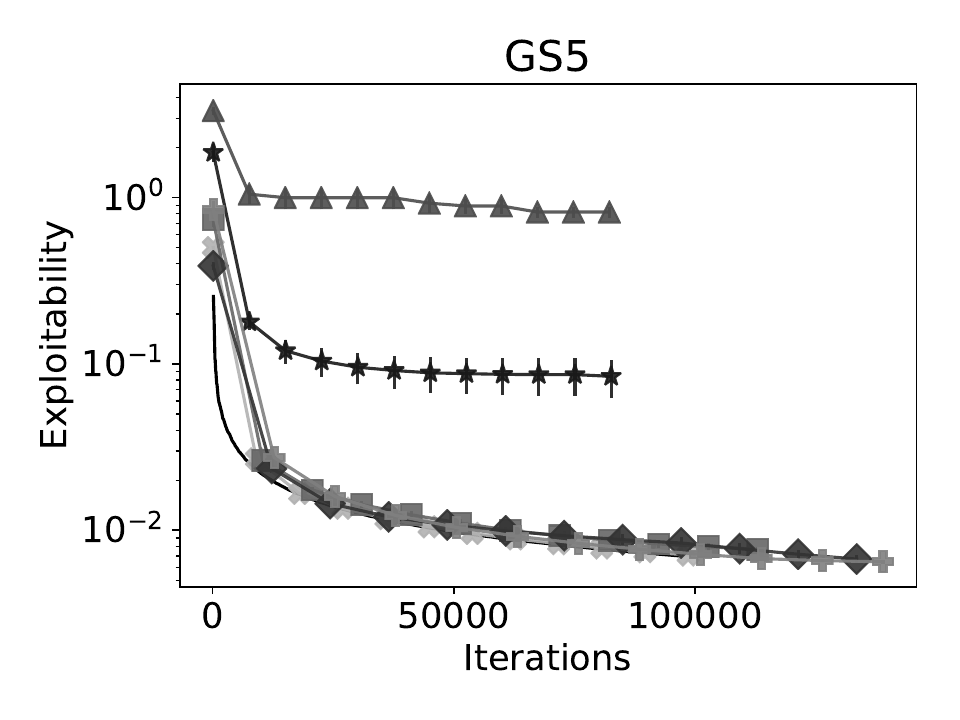}\includegraphics[width=0.33\textwidth]{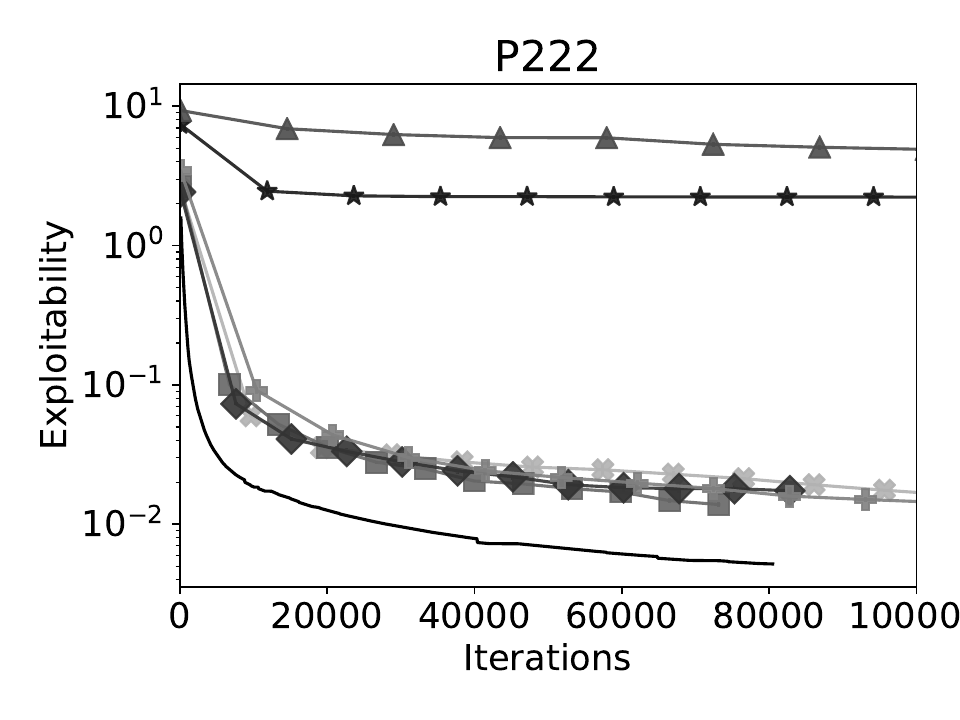}
\caption{The plots showing the sum of exploitabilities of the resulting strategies for player 1 and 2 (log y-axis) as a function of iterations (x-axis) for GP5, GS5, P222.}
\label{fig:medium_exploitability}
\end{figure}

In Figure \ref{fig:medium_exploitability} we present the sum of exploitabilities of the resulting strategies of player 1 and 2 computed by \IRCFR with various settings compared to CFR+ for GP5, GS5, and P222 as a function of the number of iterations. We depict the results for \IRCFR as averages and standard error over 10 runs of the algorithm (the standard error is usually too small to be visible). Each run uses a different seed to randomly sample the information sets for regret bound and heuristic abstraction updates. We show 3 settings of \IRCFR where $k_h+k_b = 100$ and 3 settings where $k_h+k_b = 1000$. The convergence of \IRCFR with $k_h=0$ is significantly worse compared to the rest of the settings across all domains even when we increase $k_h+k_b$ from 100 to 1000. For this reason, we focus only on settings where $k_h > 0$ in the following experimental evaluation. In GP5 and GS5 all \IRCFR parametrizations with $k_h > 0$ converge similarly to CFR+. Note that \IRCFR can converge faster than CFR+ since the final abstracted game which allows convergence can have significantly fewer information sets and hence tighter bound on the average external regret. On the other hand, in P222 CFR+ converges faster than \IRCFR. Additionally, there is very little difference between the convergence speed of the different settings of \IRCFR where $k_h > 0$.

\begin{figure}[t]
\begin{center}\includegraphics[height=0.175cm]{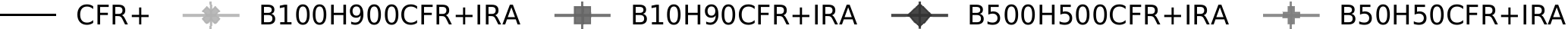}\end{center}
\vspace{-0.3cm}
\includegraphics[width=0.33\textwidth]{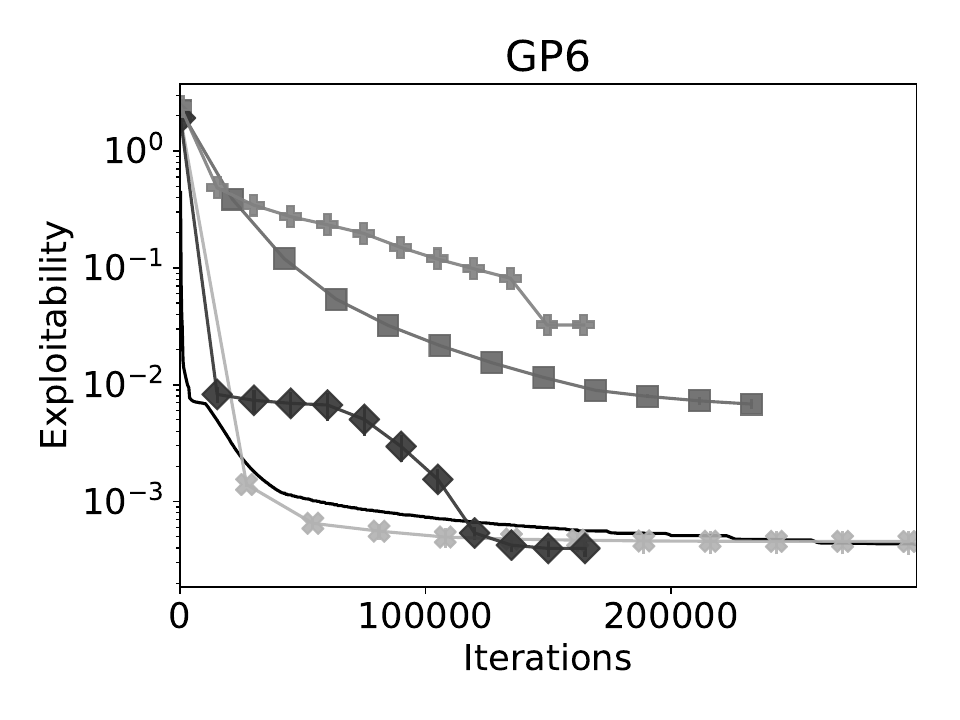}\includegraphics[width=0.33\textwidth]{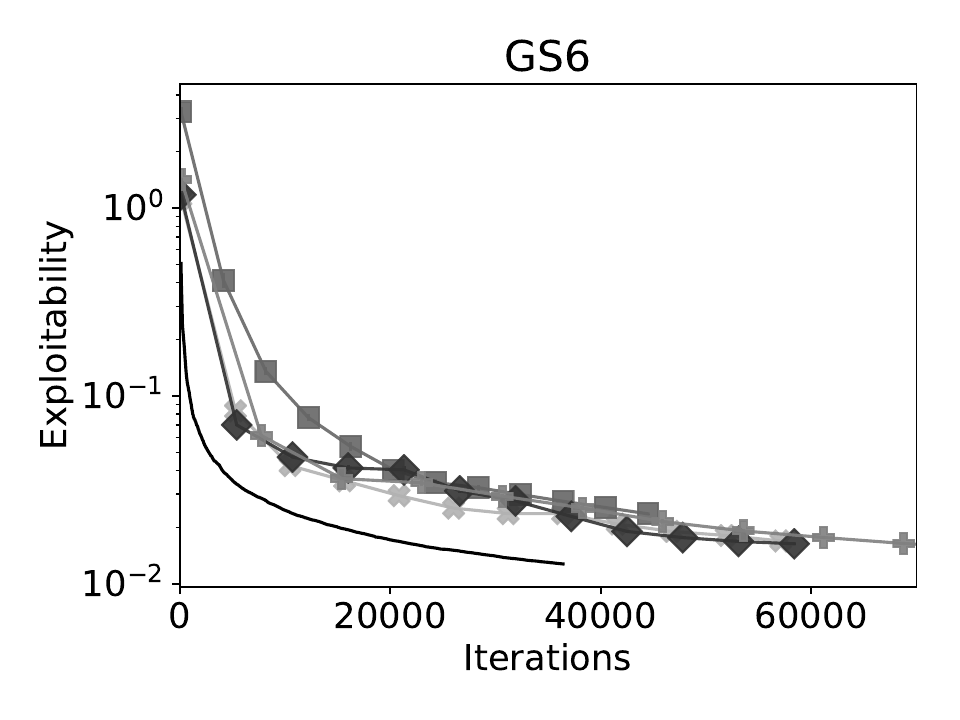}\includegraphics[width=0.33\textwidth]{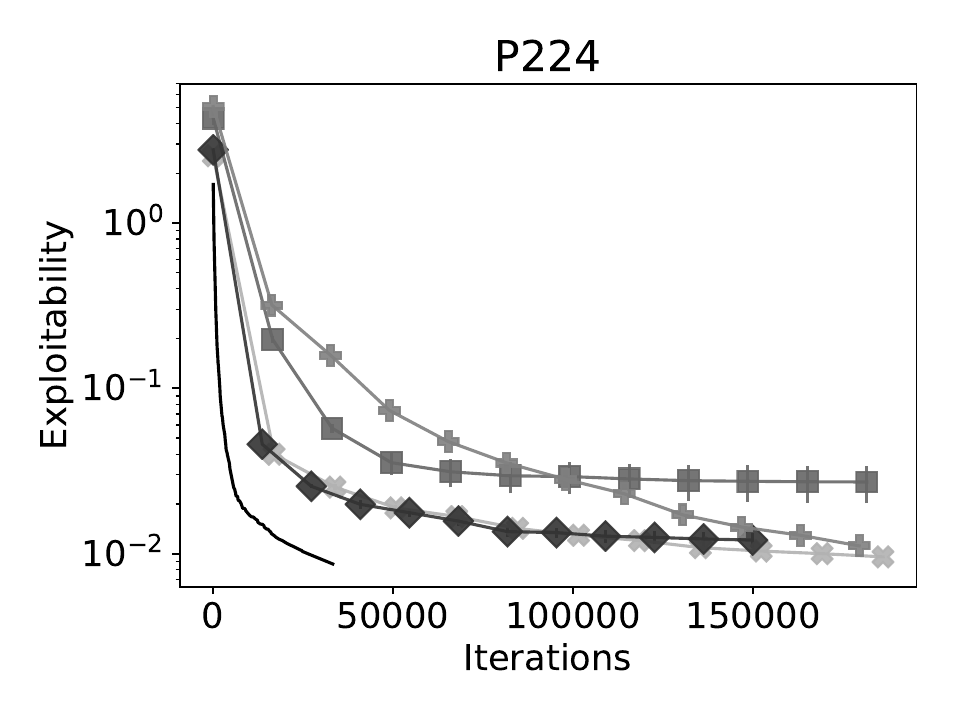}
\caption{The plots showing the sum of exploitabilities of the resulting strategies of player 1 and 2 (log y-axis) as a function of iterations (x-axis) for GP6, GS6, P224.}
\label{fig:large_exploitability}
\end{figure}

In Figure \ref{fig:large_exploitability} we show the same results for GP6, GS6 and P224. The plots show slower convergence of \IRCFR with $k_h+k_b = 100$ compared to $k_h+k_b = 1000$. Additionally, except for GP6, all versions of \IRCFR converge slower than CFR+. This is expected, since using $k_h+k_b = 100$ and $k_h+k_b = 1000$ means that the algorithm uses less than $0.1\%$ and $1\%$ of information sets for abstraction update in all 3 domains. Hence it takes longer to refine the abstraction to allow strategies with a smaller exploitability.

The results above suggest that all the evaluated settings of \IRCFR with fixed $k_h + k_b$ and $k_h > 0$ perform similarly. Additionally, we have observed that this similarity also holds for the memory required by \IRCFR. On the other hand, increasing the sum of $k_h + k_b$ from 100 to 1000 improved the convergence speed of \IRCFR in the larger domains. Hence, for clarity, in the following experiments we report only two settings of \IRCFR, namely B10H90\IRCFR and B100H900\IRCFR. 

\subsection{Memory Requirements of Algorithms}
In this section, we discuss the memory requirements of FPIRA, \IRCFR, and DOEFG as a function of the exploitability of the resulting strategies on two sets of domains: GP5, GS5, P222 with approx. $10^4$ information sets and GP6, GS6, P224 with approx. $10^5$ information sets. Since the actual memory usage is implementation dependent, we analyze the size of the abstractions built by \IRCFR and FPIRA and the size of the rest of the data structures required by these algorithms. We compare these results with the size of the restricted game built by DOEFG and other data structures stored by this algorithm. Next, we compare the number of 32-bit words stored during the run of FPIRA, \IRCFR, and DOEFG.
We conclude this section by more detailed analysis of the scalability of \IRCFR in the poker domain and compare this scalability with DOEFG and FPIRA.

\subsubsection{GP5, GS5 and P222}
\begin{figure}[t]
\begin{center}\includegraphics[height=0.175cm]{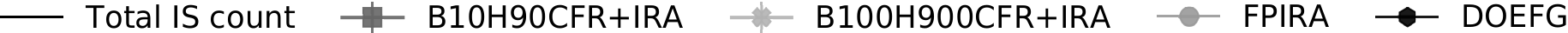}\end{center}
\vspace{-0.3cm}
\includegraphics[width=0.33\textwidth]{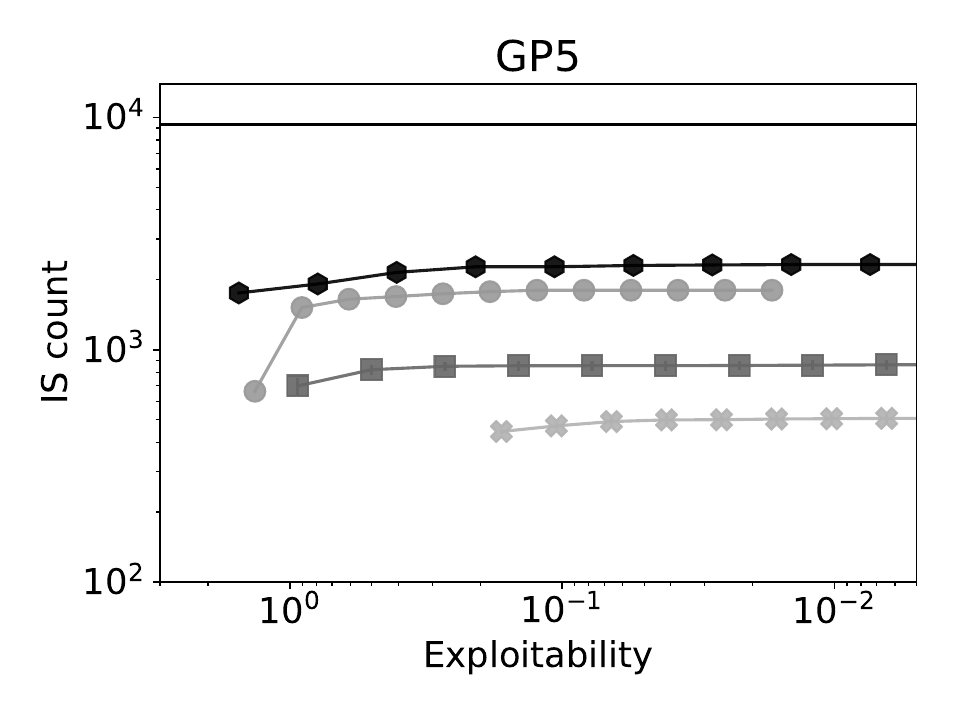}\includegraphics[width=0.33\textwidth]{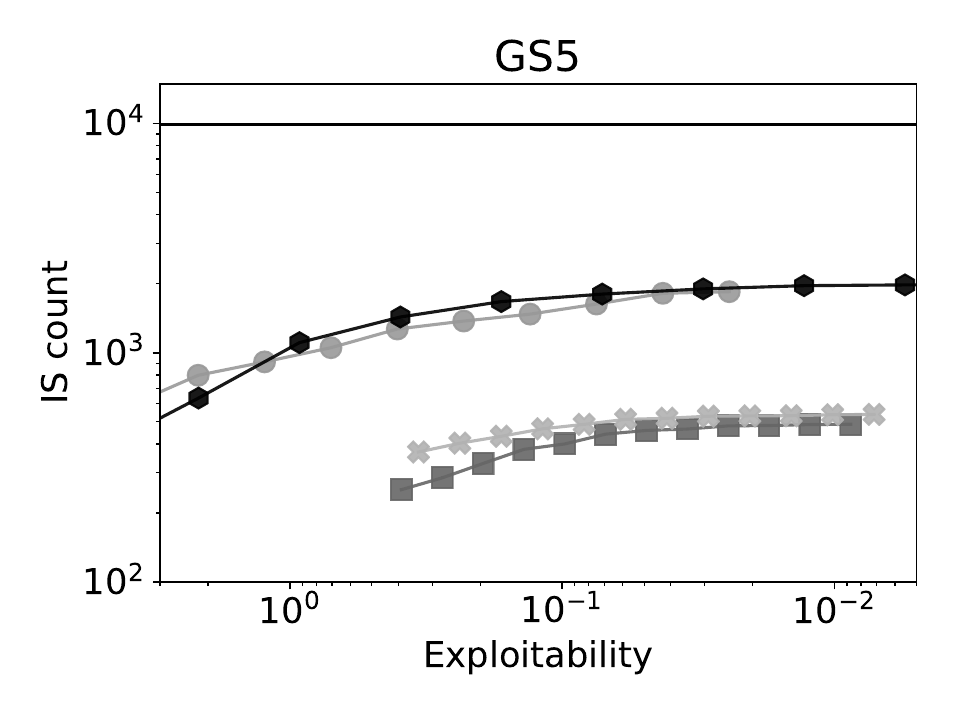}\includegraphics[width=0.33\textwidth]{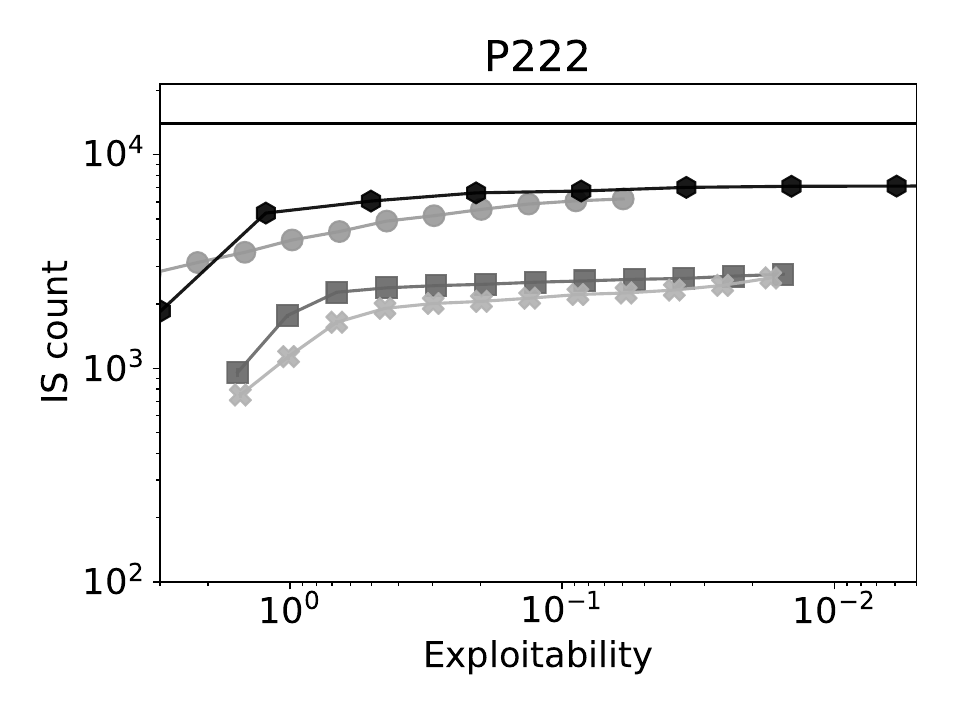}
\caption{The plots depicting the number of information sets (log y-axis) used by algorithms to compute strategies with the exploitability depicted on the log x-axis for GP5, GS5, P222.}
\label{fig:medium_abstr_size}
\end{figure}

In Figure \ref{fig:medium_abstr_size} we show the number of information sets in the abstraction (or restricted game) built by the algorithms as a function of the sum of exploitabilities of the resulting strategies of player 1 and 2 for GP5, GS5, and P222. We report these results for DOEFG, FPIRA and two settings of \IRCFR. The lines for \IRCFR depict the average number of information sets over 10 runs of \IRCFR. Each run uses a different seed to randomly sample the information sets for regret bound and heuristic abstraction updates. Note that the standard error is too small to be visible in these plots.
FPIRA and DOEFG require a similar number of information sets to reach strategies with equal exploitability. The \IRCFR, on the other hand, requires significantly fewer information sets than FPIRA and DOEFG for these domains. 
For exploitability $0.05$, the B10H90\IRCFR uses on average $9.5\%$, $3.1\%$, $14.9\%$ of information sets of the total information set count of GP5, GS5 and P222 respectively, while FPIRA uses $19.6\%$, $15.6\%$, $42.9\%$ and DOEFG $24.0\%$, $17.2\%$, $47.8\%$. 

\begin{figure}[t]
\begin{center}\includegraphics[height=0.165cm]{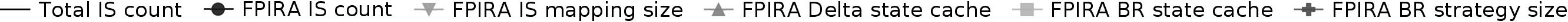}\end{center}
\vspace{-0.3cm}
\includegraphics[width=0.33\textwidth]{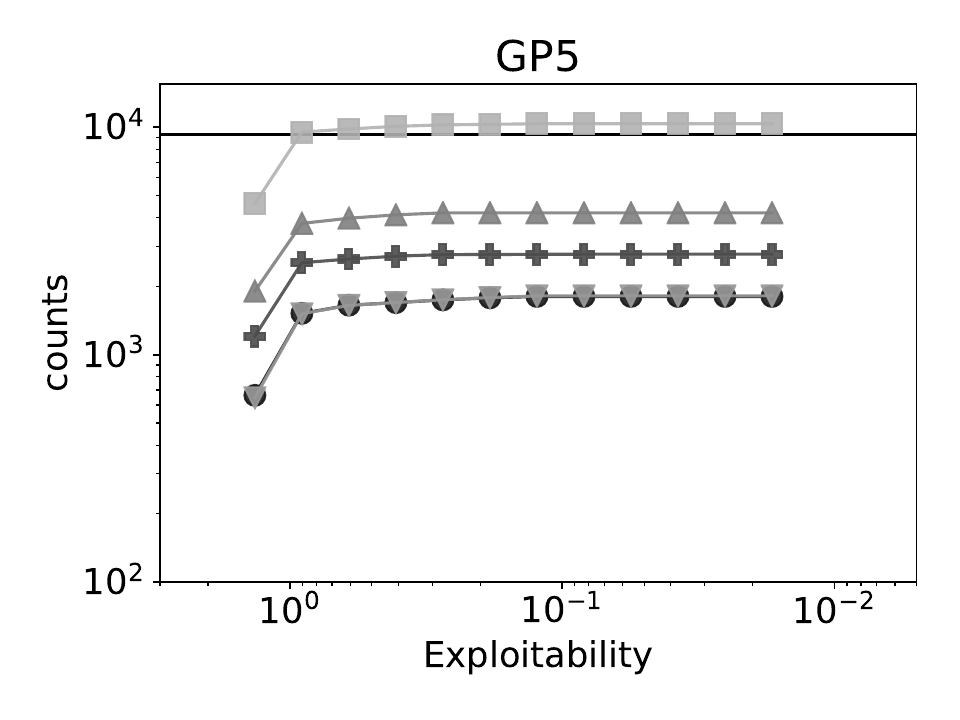}\includegraphics[width=0.33\textwidth]{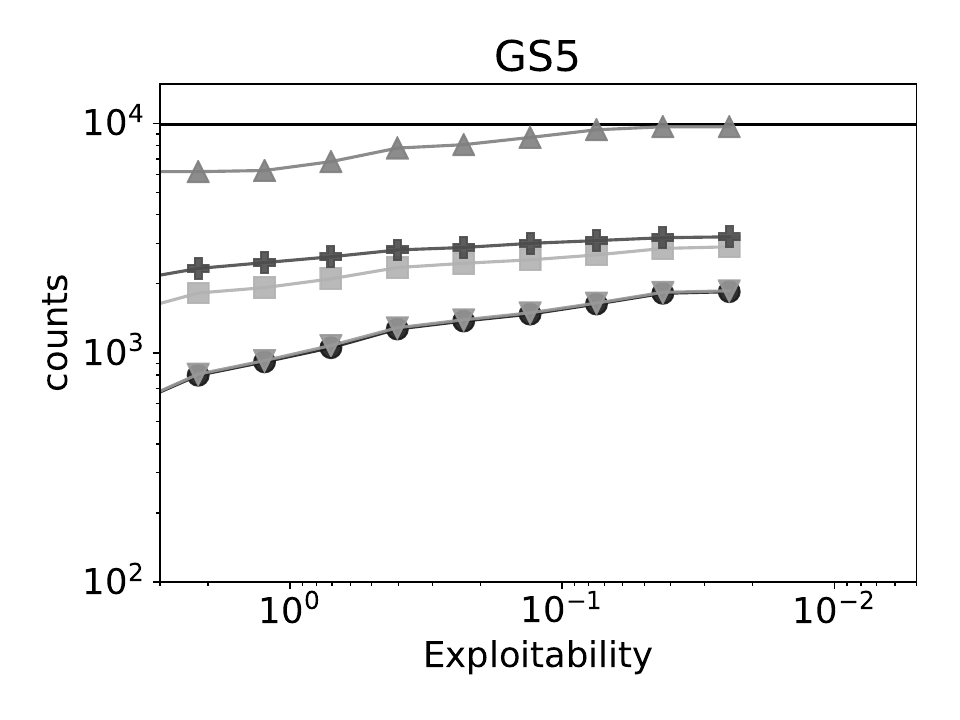}\includegraphics[width=0.33\textwidth]{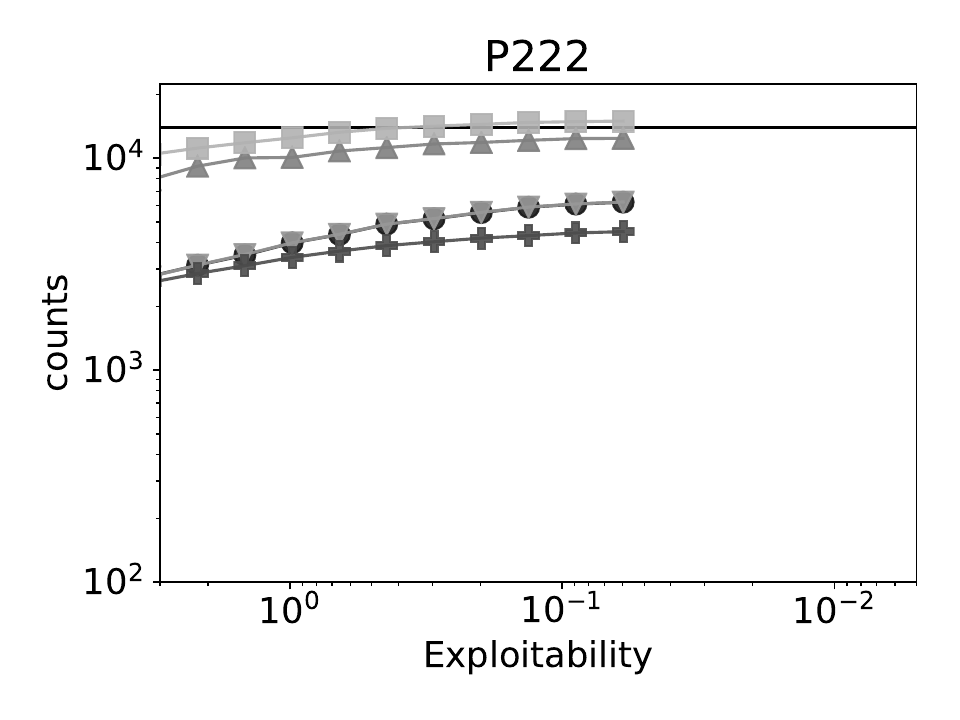}
\vspace{-0.4cm}
\begin{center}\includegraphics[height=0.175cm]{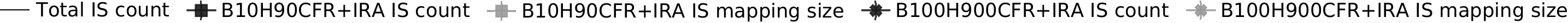}\end{center}
\vspace{-0.3cm}
\includegraphics[width=0.33\textwidth]{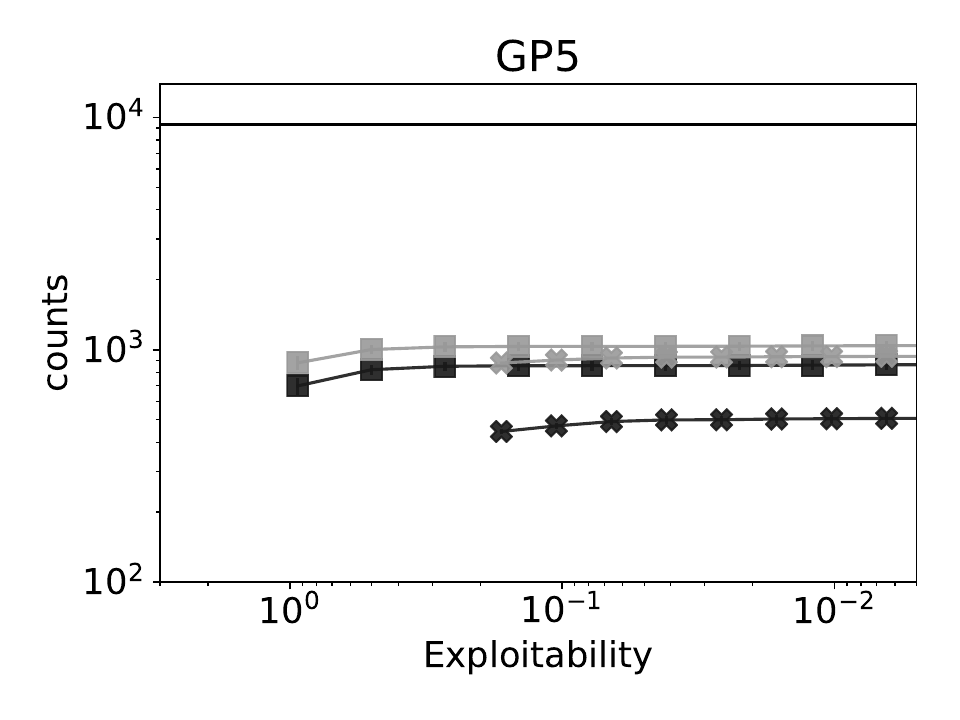}\includegraphics[width=0.33\textwidth]{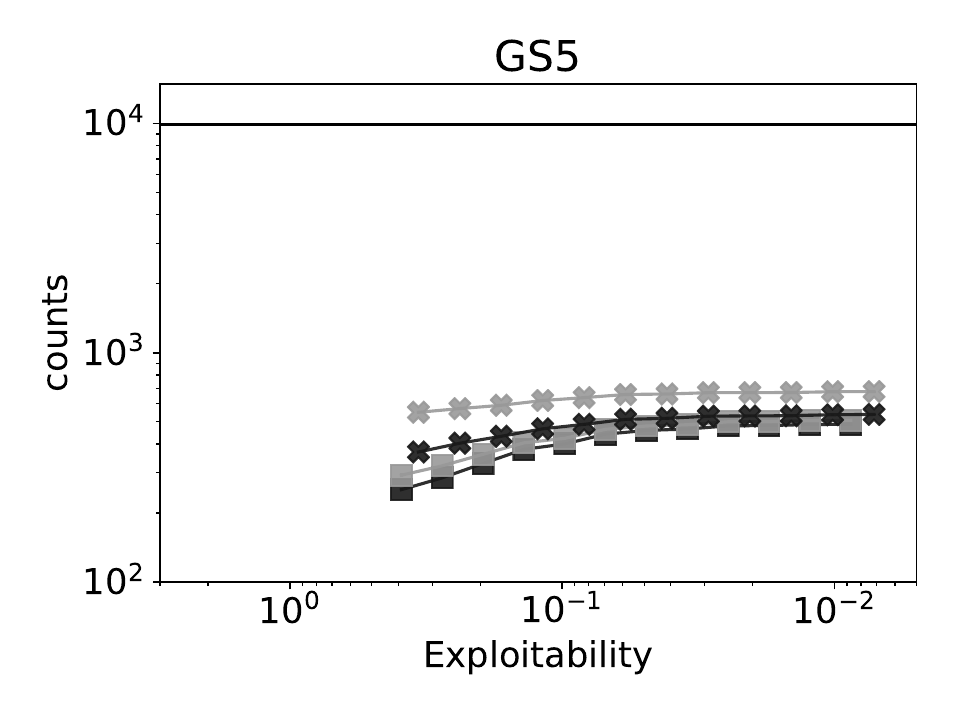}\includegraphics[width=0.33\textwidth]{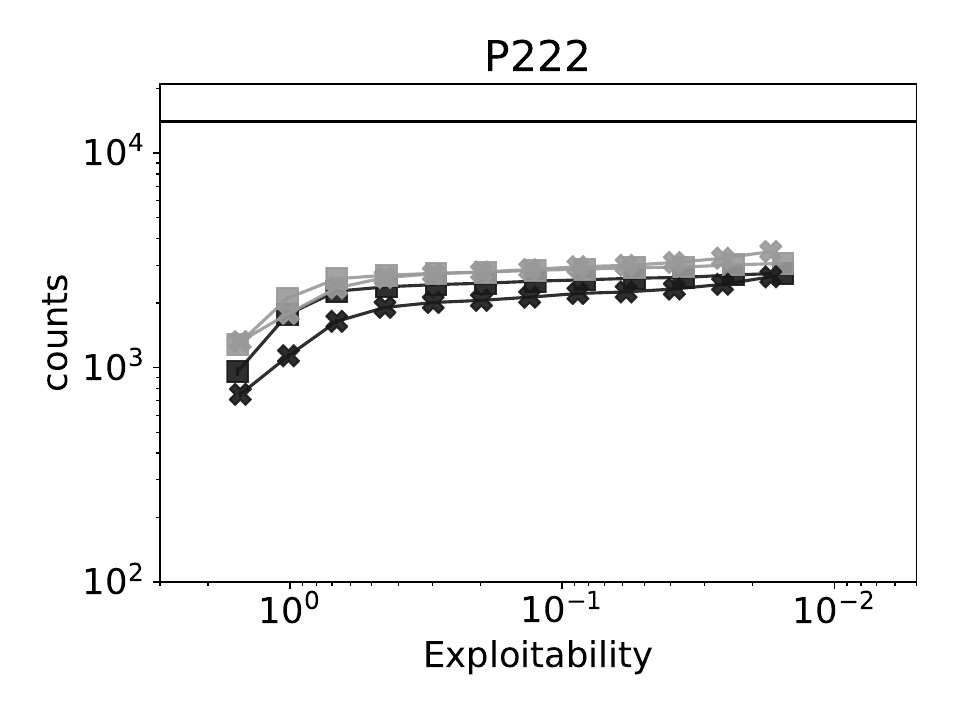}
\vspace{-0.4cm}
\begin{center}\includegraphics[height=0.165cm]{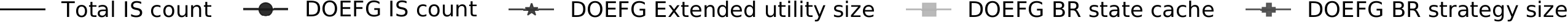}\end{center}
\vspace{-0.3cm}
\includegraphics[width=0.33\textwidth]{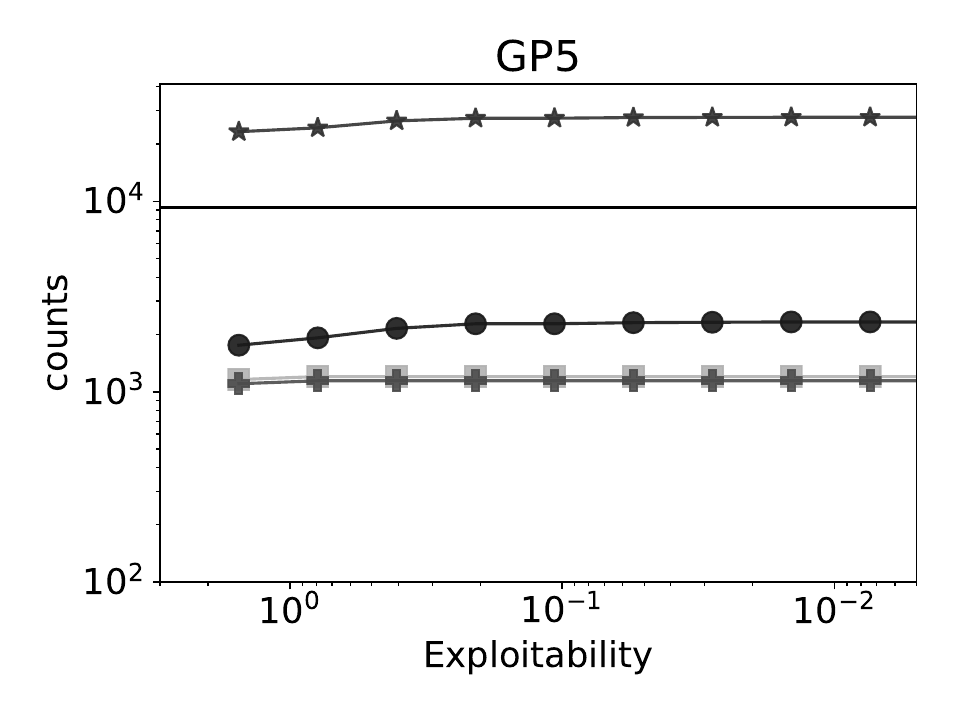}\includegraphics[width=0.33\textwidth]{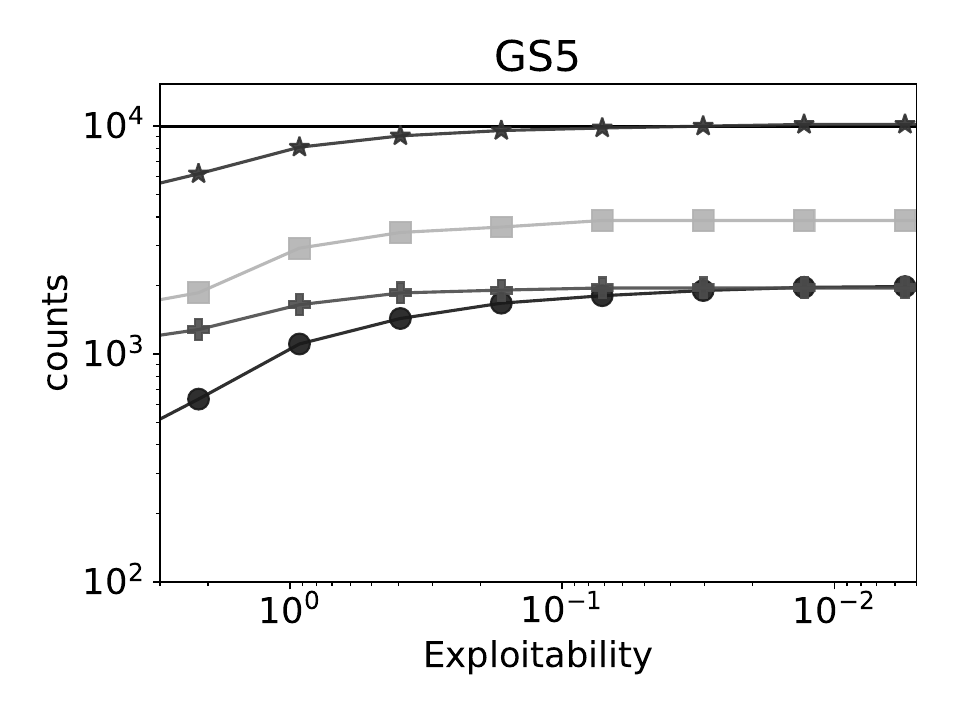}\includegraphics[width=0.33\textwidth]{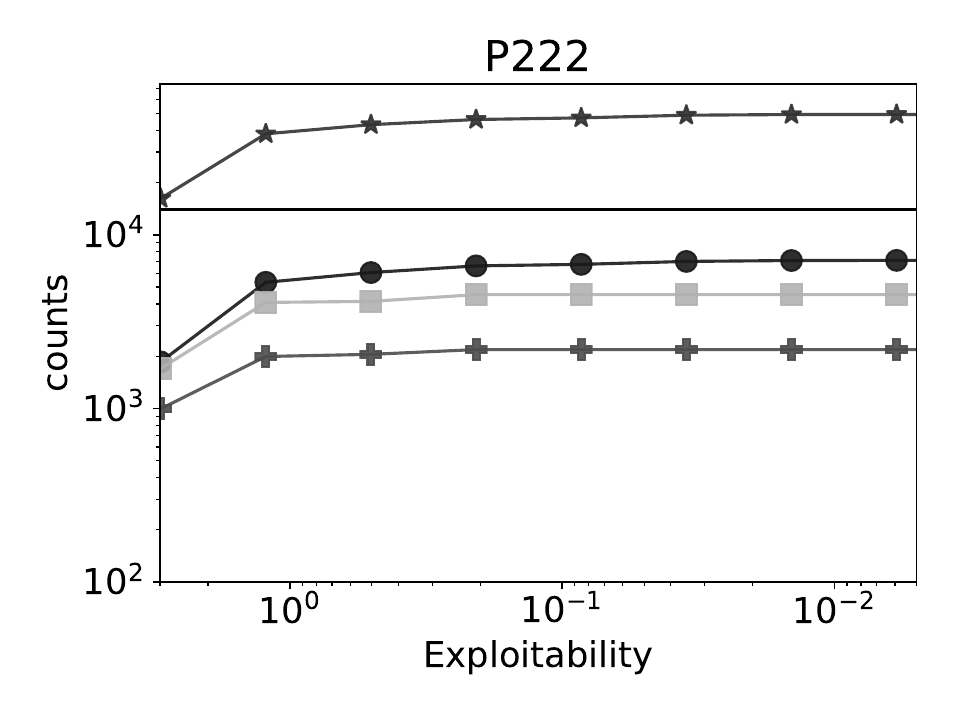}
\caption{Plots showing size of data stored during the run of FPIRA in the first row, \IRCFR in the second row and DOEFG in the third row (log y-axis) as a function of the sum of exploitabilities of the resulting strategies of player 1 and 2 (log x-axis) for GP5, GS5 and P222.}
\label{fig:medium_memory}
\end{figure}

In Figure \ref{fig:medium_memory} we show the size of data structures stored during the run of FPIRA, \IRCFR, and DOEFG for GP5, GS5, and P222. 

\textbf{FPIRA:} The results for FPIRA are presented in the plots in the first row of the Figure. We report the number of information sets of the abstraction, the size of the information set mapping, the size of the cache used during the best response computation, the size of the best response and the size of the cache used during the $\Delta$ computation as a function of the sum of exploitabilities of the resulting strategies of player 1 and 2. 
Note that the cache sizes and best response strategy size are reported as the maximum size encountered until FPIRA reached strategies with the corresponding exploitability. 
The results show that the size of the abstraction mapping remains small. On the other hand, the cache used during the best response computation and the $\Delta$ computation can require prohibitive memory when applied to large games (the cache stores value for each state encountered during the best response computation, and hence its size can be significantly larger than the number of information sets). For example, for exploitability $0.05$ the size of the information set mapping was $20.5\%$, $18.2\%$, $43.4\%$ of the total information set count of GP5, GS5, and P222, while the size of the cache in the best response computation was $109.6\%$, $23.8\%$, $106.1\%$. 

\textbf{\IRCFR:} The plots in the second row show that both the average mapping size and the average number of information sets in the abstraction in \IRCFR remain small: B10H90\IRCFR stored the mapping on average only for $10.7\%$, $3.9\%$, $22.1\%$ of the total information sets for GP5, GS5, and P222 respectively. 
Additionally, in case of GP5, GS5 and P222 $k_h + k_b = 100$ corresponds to storing regrets required for the abstraction update only in $1.1\%, 1.0\%, 0.7\%$ of information sets of the whole game respectively. 

\textbf{DOEFG:} Finally in the third row we show the data structures required by the DOEFG. Similarly to FPIRA, we report the maximum size of the cache used in the best response and the maximum size of the best response strategy since DOEFG uses the best response computation with the cache as oracle extending the restricted game (see \cite{bosansky2014}, Section 4.2). Additionally, we report the size of the extended utility, which is required to construct the sequence-form LP for the current restricted game. The extended utility stores one value for each combination of sequences leading to some leaf or temporary leaf (see \cite{bosansky2014} for more details). Finally, we report the number of information sets in the restricted game. Note that it is not clear whether these are all the data required by the DOEFG since DOEFG was never implemented or described with emphasis on memory efficiency. However, we believe that these data form a necessary subset of the data required by the DOEFG. The results show that storing the extended utility requires prohibitive memory as its size can be significantly larger than the number of information sets. For example, for exploitability $0.05$ the size of the extended utility was $297\%$, $102\%$ and $352\%$ of the total information set count of GP5, GS5, and P222.

\begin{figure}[t]
\begin{center}\includegraphics[height=0.175cm]{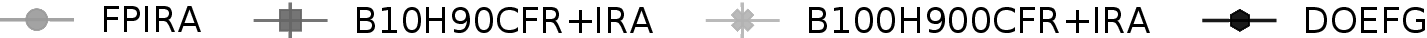}\end{center}
\vspace{-0.3cm}
\includegraphics[width=0.33\textwidth]{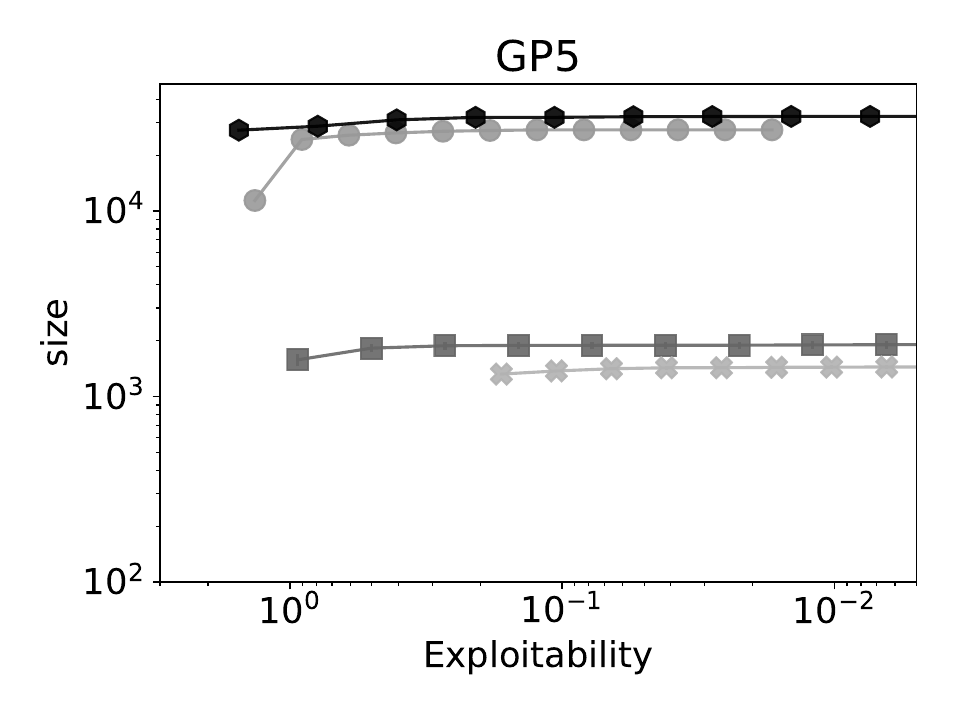}\includegraphics[width=0.33\textwidth]{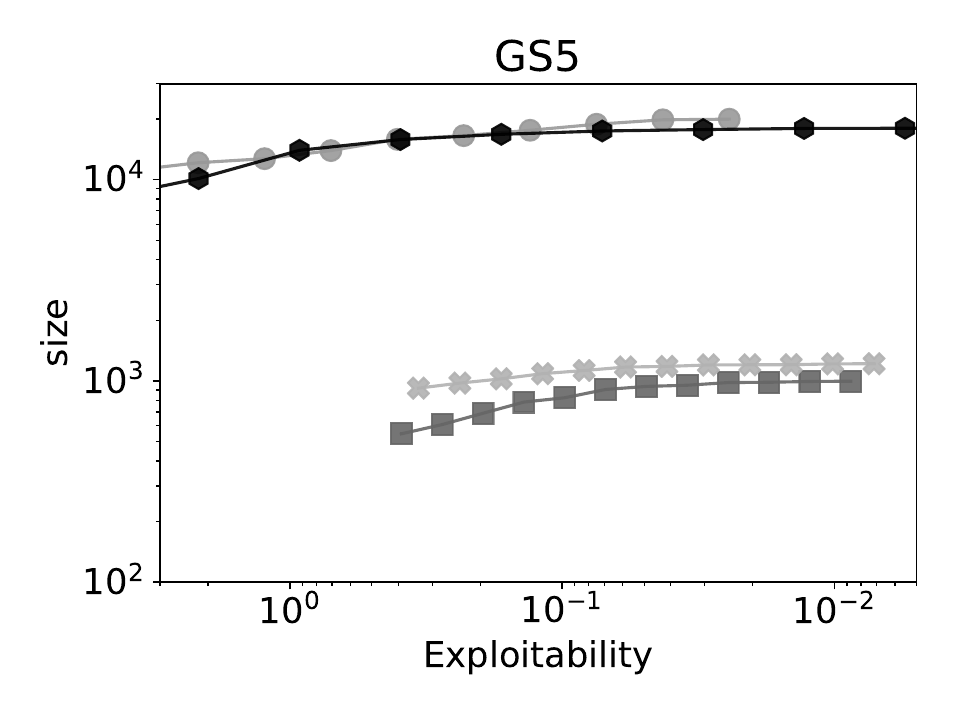}\includegraphics[width=0.33\textwidth]{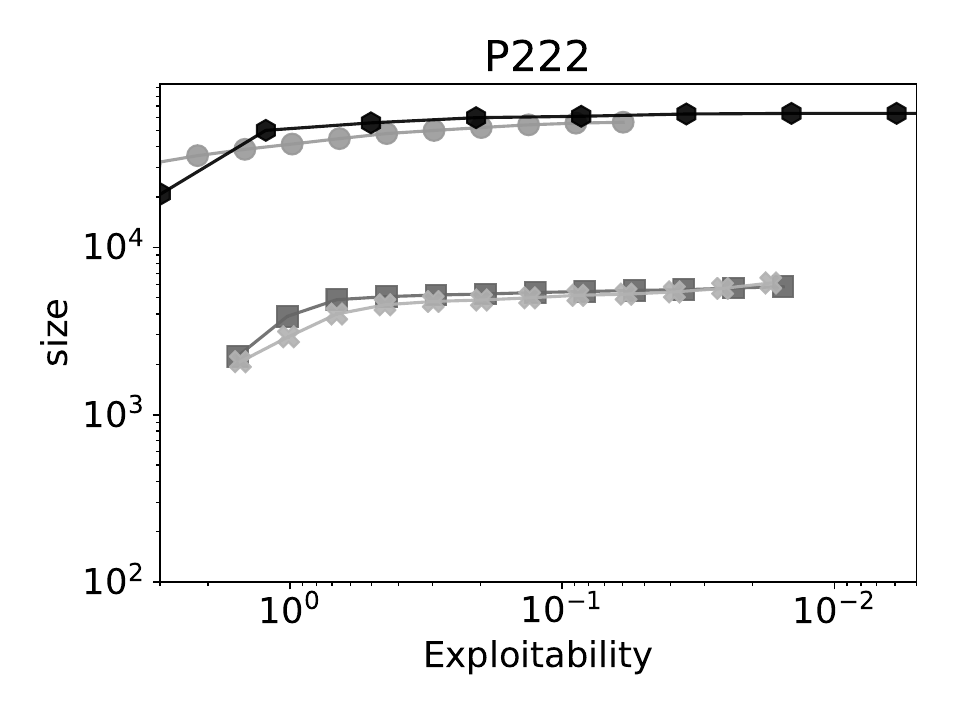}
\caption{The number of 32 bit words the algorithms store (log y-axis) as a function of the exploitability of the resulting strategies (log x-axis) for GP5, GS5 and P222.}
\label{fig:medium_number_count}
\end{figure}

Finally, in Figure \ref{fig:medium_number_count} we present the total number of 32-bit words (integers and floats) the algorithms need to store in GP5, GS5, and P222  as a function of the exploitability of the resulting strategies. These values were computed from the data depicted in Figure \ref{fig:medium_memory}. Furthermore, to properly reflect the data stored by FPIRA in the abstraction, we replace the abstraction size by the number of floats that are used to represent the average strategy stored in the current abstraction. Similarly, for \IRCFR we replace the average abstraction size by the average number of floats that are required to represent the regrets and average strategies stored in the current abstraction. These results show that the number of words stored by FPIRA and DOEFG is comparable. On the other hand, \IRCFR in both settings requires an order of magnitude smaller memory than FPIRA and DOEFG.

\subsubsection{GP6, GS6 and P224}
\begin{figure}[t]
\begin{center}\includegraphics[height=0.175cm]{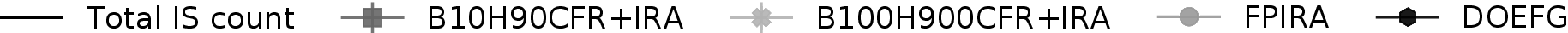}\end{center}
\vspace{-0.3cm}
\includegraphics[width=0.33\textwidth]{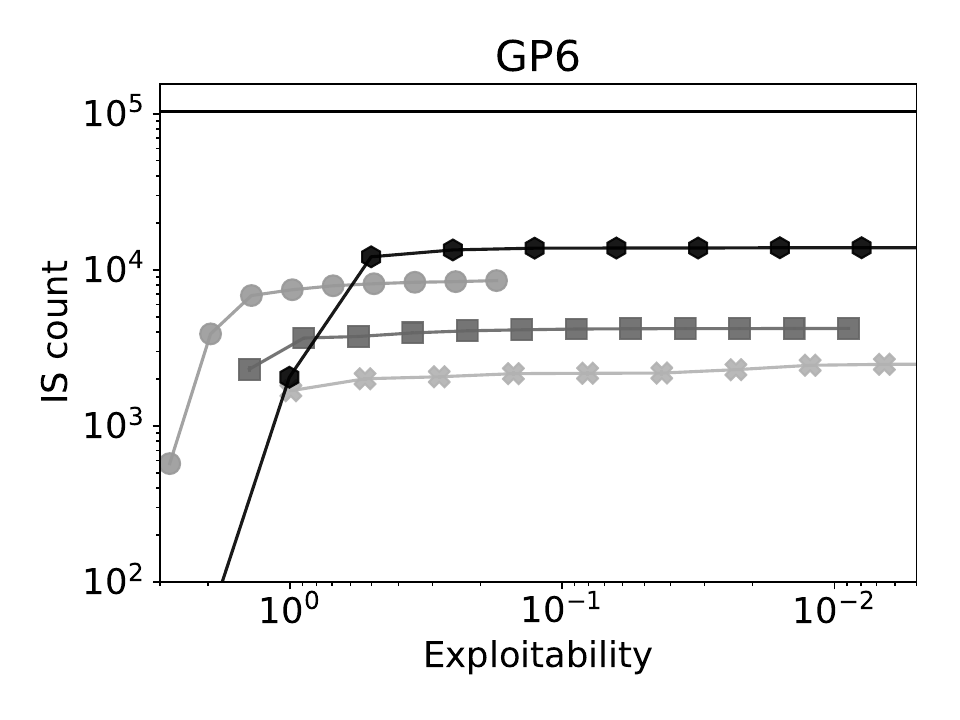}\includegraphics[width=0.33\textwidth]{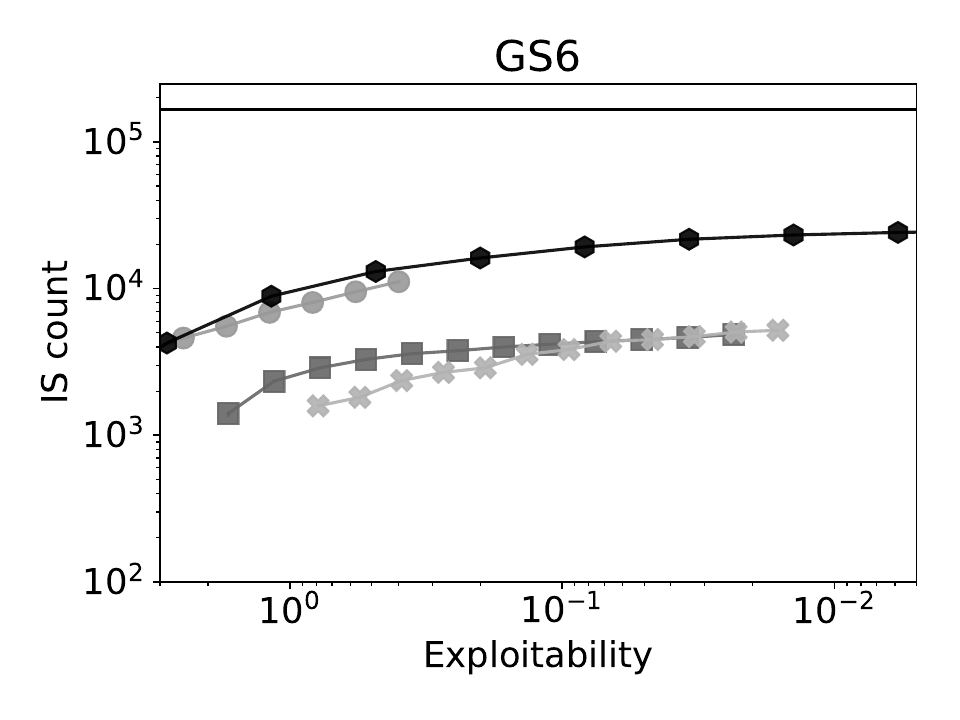}\includegraphics[width=0.33\textwidth]{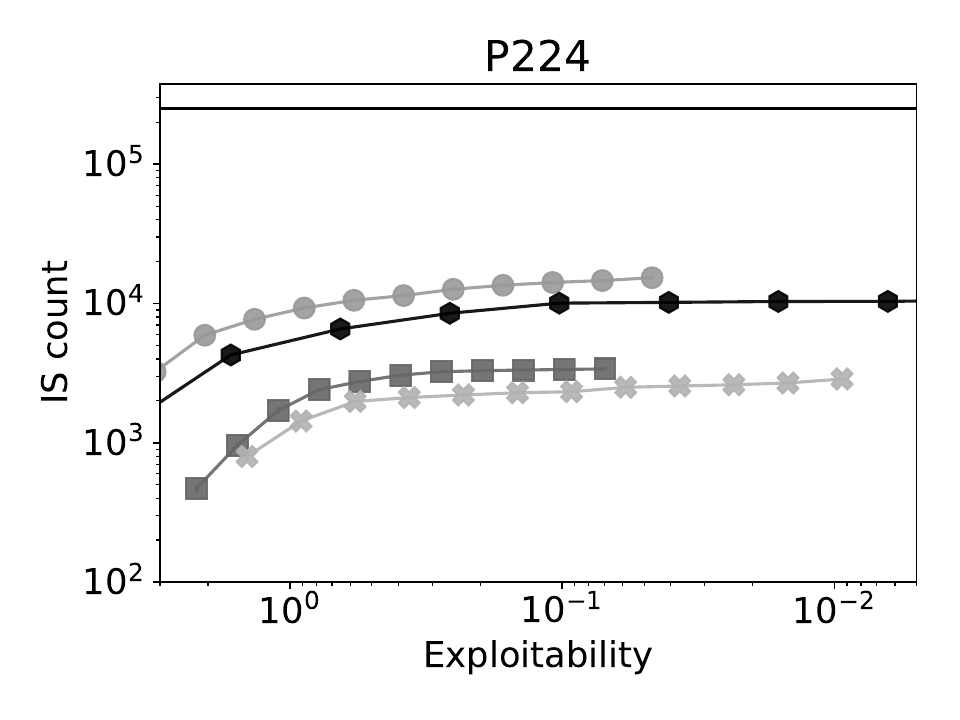}
\caption{The plots depicting the number of information sets (log y-axis) used by algorithms to compute strategies with the sum of their exploitabilities depicted on the log x-axis for GP6, GS6 and P224.}
\label{fig:large_abstr_size}
\end{figure}
In Figure \ref{fig:large_abstr_size} we present the results showing the abstraction size for GP6, GS6 and P224. We depict the results for \IRCFR as averages with the standard error over 5 runs with different seeds (the standard error is again too small to be visible). The \IRCFR is capable of solving the games using abstractions with significantly less information sets than the rest of the algorithms. 
For exploitability of the resulting strategies $0.05$, the B100H900\IRCFR  uses on average $2.0\%$, $2.4\%$, $0.9\%$ of information sets of GP6, GS6 and P224, while DOEFG uses $13.3\%$, $12.0\%$, $4.0\%$. A slow runtime prevented FPIRA from convergence to strategies with exploitability 0.05 in the given time (see Section \ref{sec:experiments:runtime} for runtime analysis). Additionally, in case of GP6, GS6 and P224, $k_h + k_b = 100$ corresponds to storing regrets required for the abstraction update only in $0.09\%, 0.06\%, 0.03\%$ of information sets of the whole game respectively.

\begin{figure}[t]
\begin{center}\includegraphics[height=0.175cm]{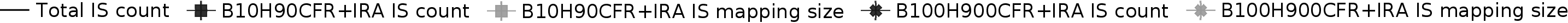}\end{center}
\vspace{-0.3cm}
\includegraphics[width=0.33\textwidth]{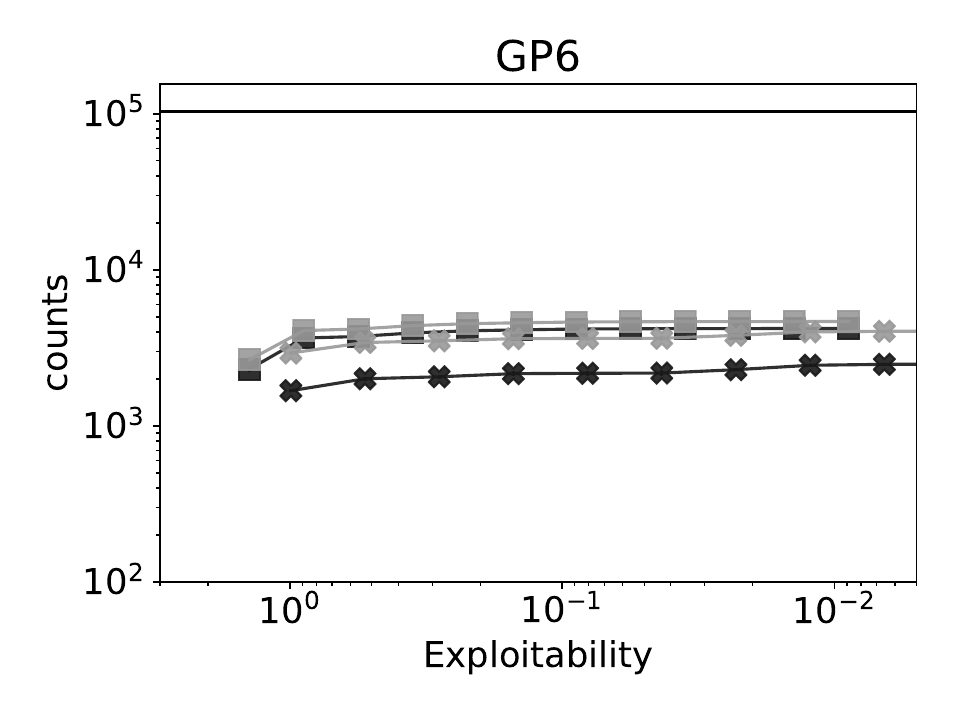}\includegraphics[width=0.33\textwidth]{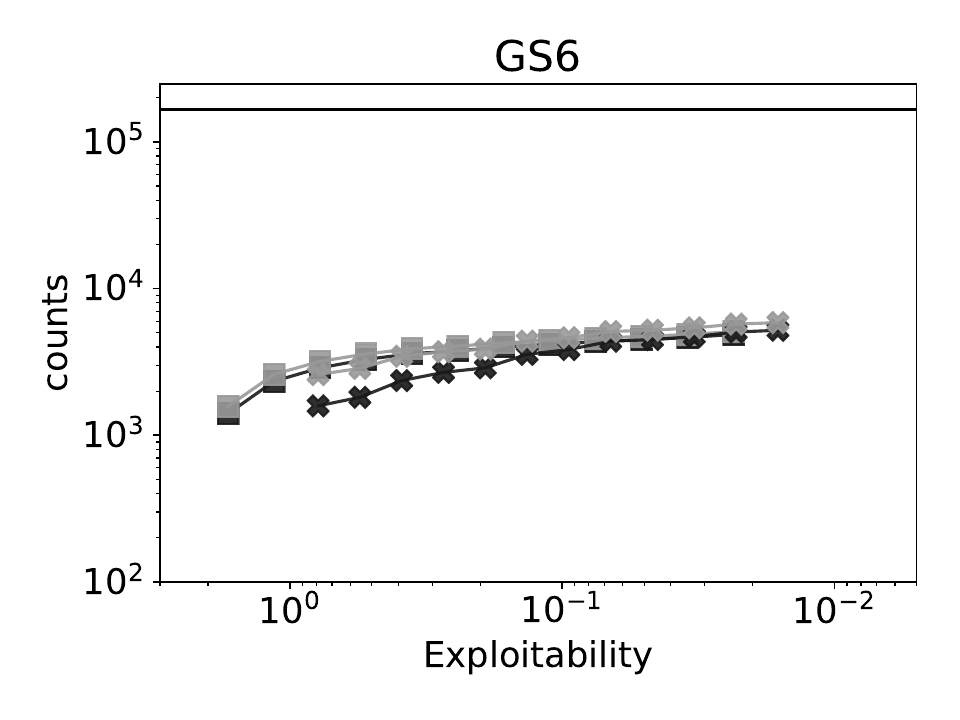}\includegraphics[width=0.33\textwidth]{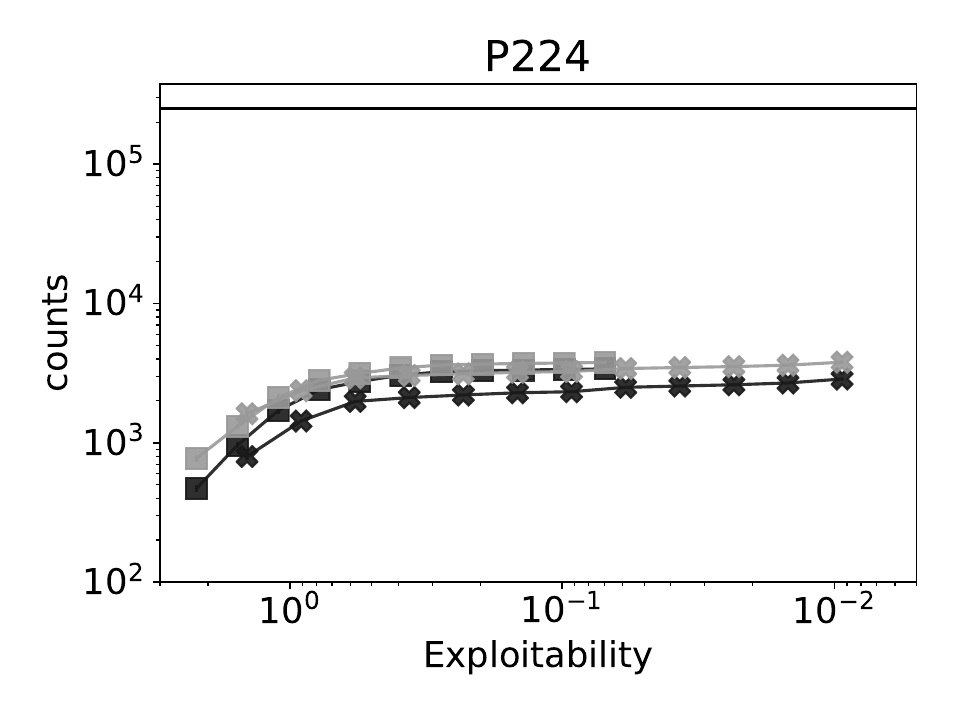}\\
\vspace{-0.6cm}
\begin{center}\includegraphics[height=0.165cm]{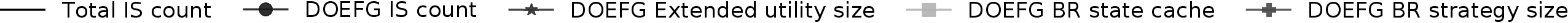}\end{center}
\vspace{-0.3cm}
\includegraphics[width=0.33\textwidth]{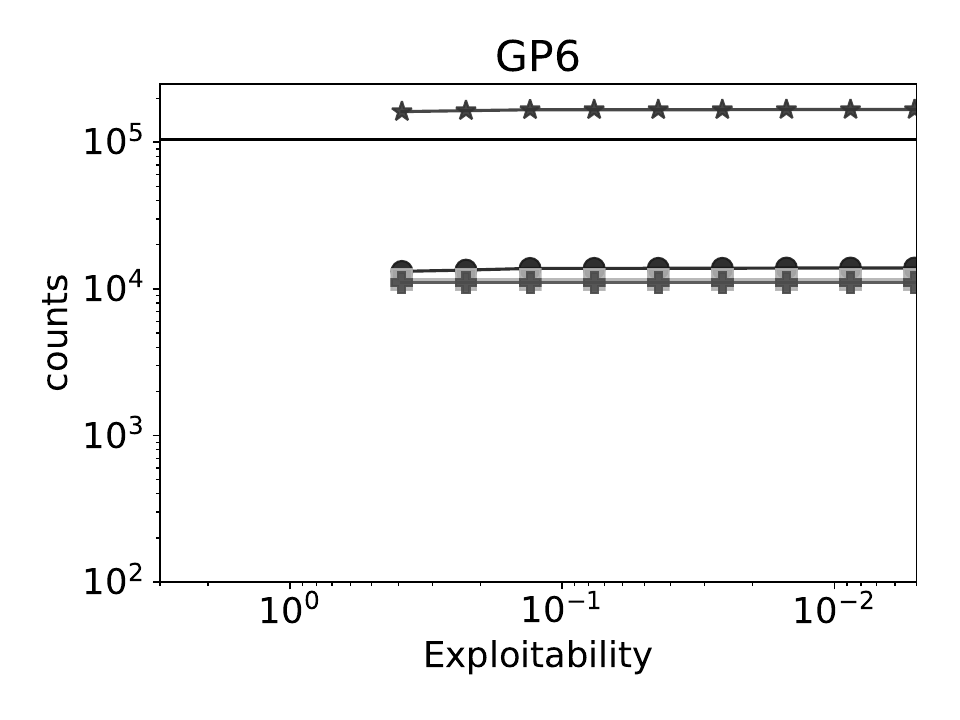}\includegraphics[width=0.33\textwidth]{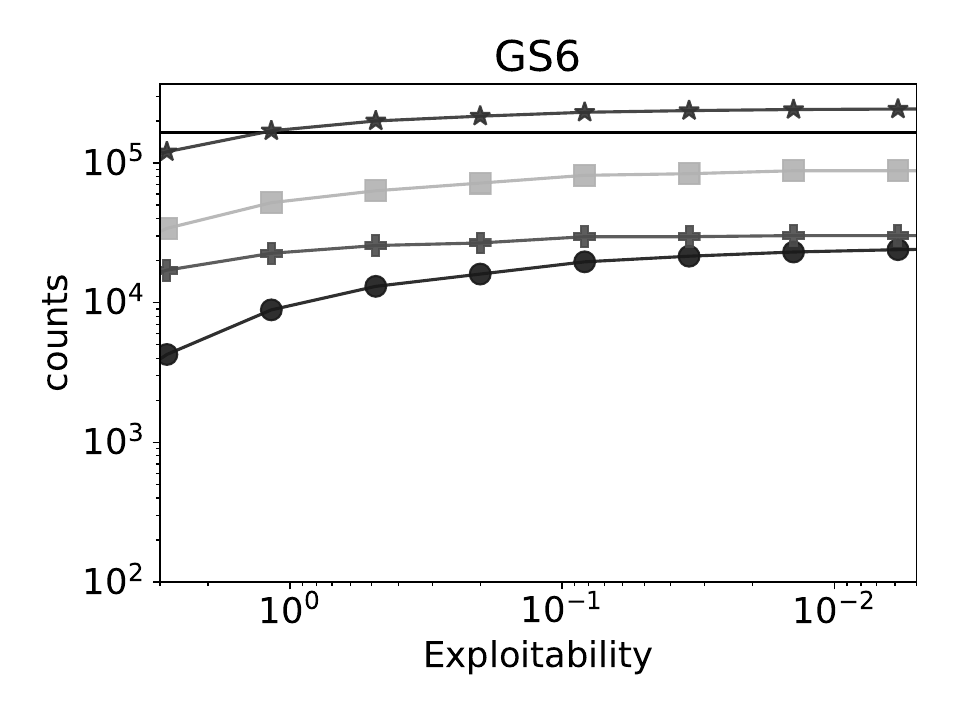}\includegraphics[width=0.33\textwidth]{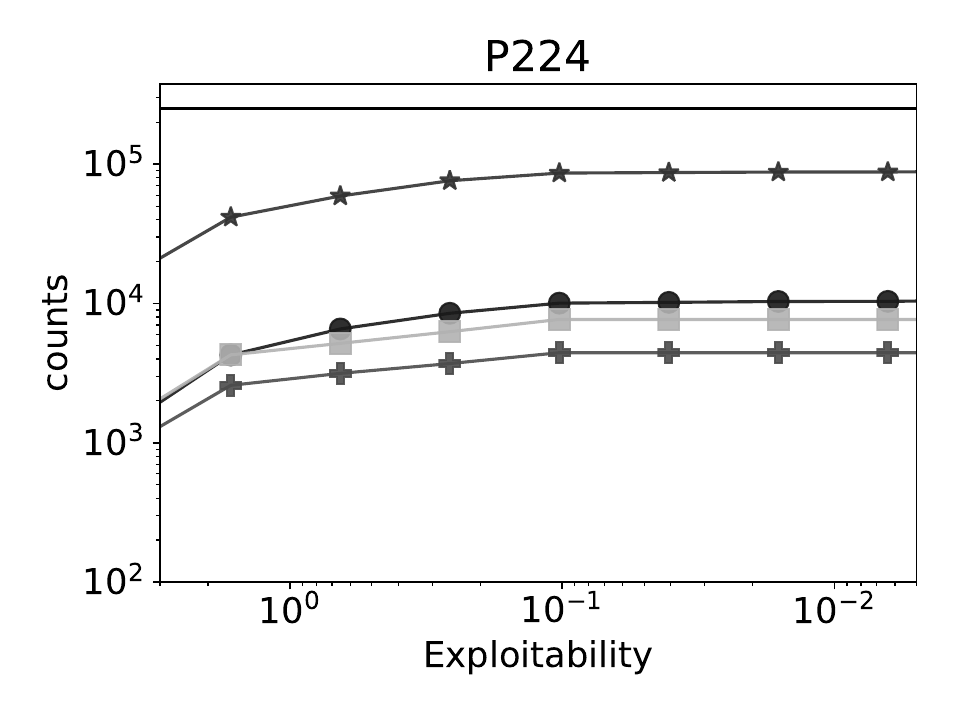}
\caption{Plots showing size of data stored during the run of \IRCFR in the first row and DOEFG in the second row (log y-axis) as a function of the sum of exploitabilities of the resulting strategies of player 1 and 2 (log x-axis) for GP6, GS6 and P224.}
\label{fig:large_memory}
\end{figure}

In Figure \ref{fig:large_memory} we present the size of data structures stored during the run of \IRCFR (first row) and DOEFG (second row) for GP6, GS6, and P224.
These results confirm that \IRCFR requires small memory even in larger domains, while the data stored by DOEFG remain large. %

\begin{figure}[t]
\begin{center}\includegraphics[height=0.175cm]{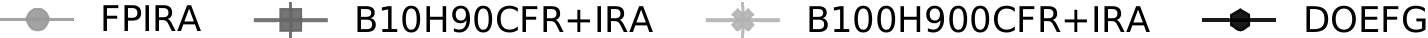}\end{center}
\vspace{-0.3cm}
\includegraphics[width=0.33\textwidth]{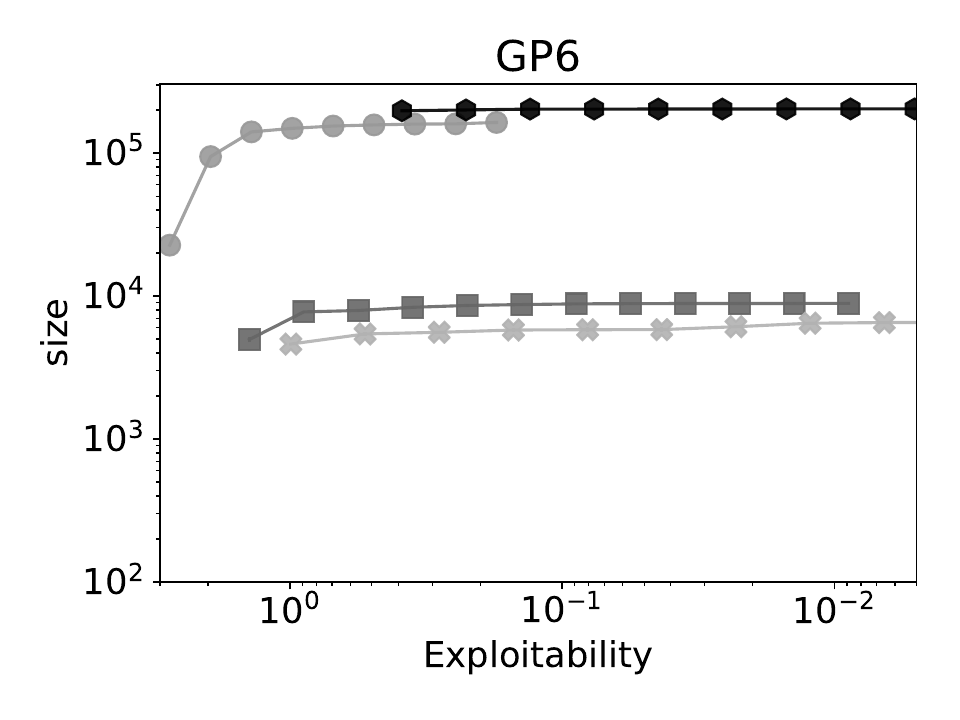}\includegraphics[width=0.33\textwidth]{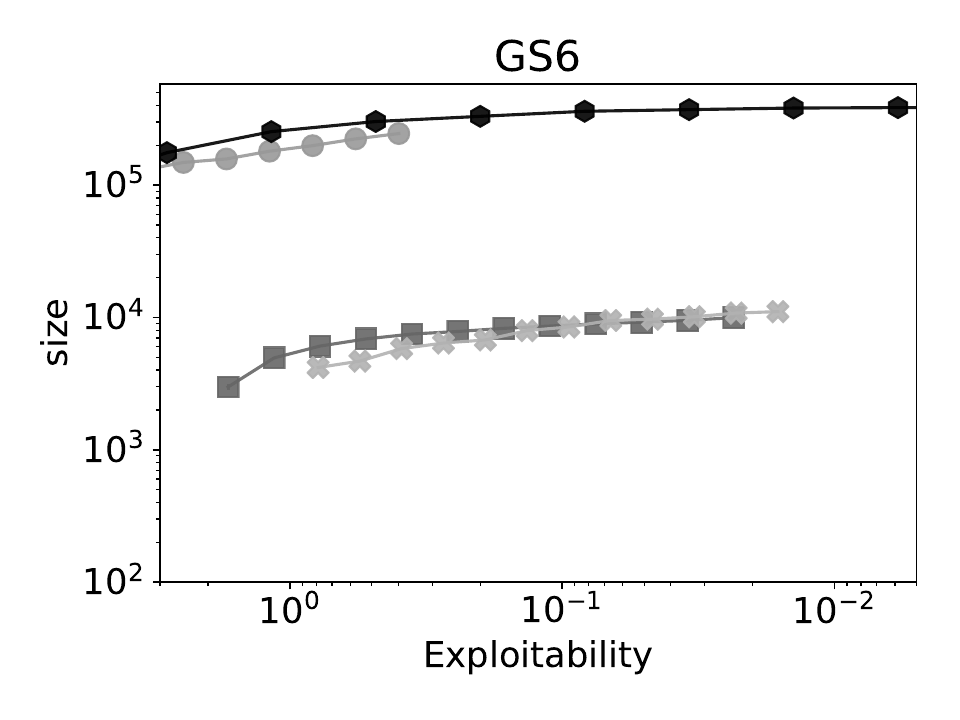}\includegraphics[width=0.33\textwidth]{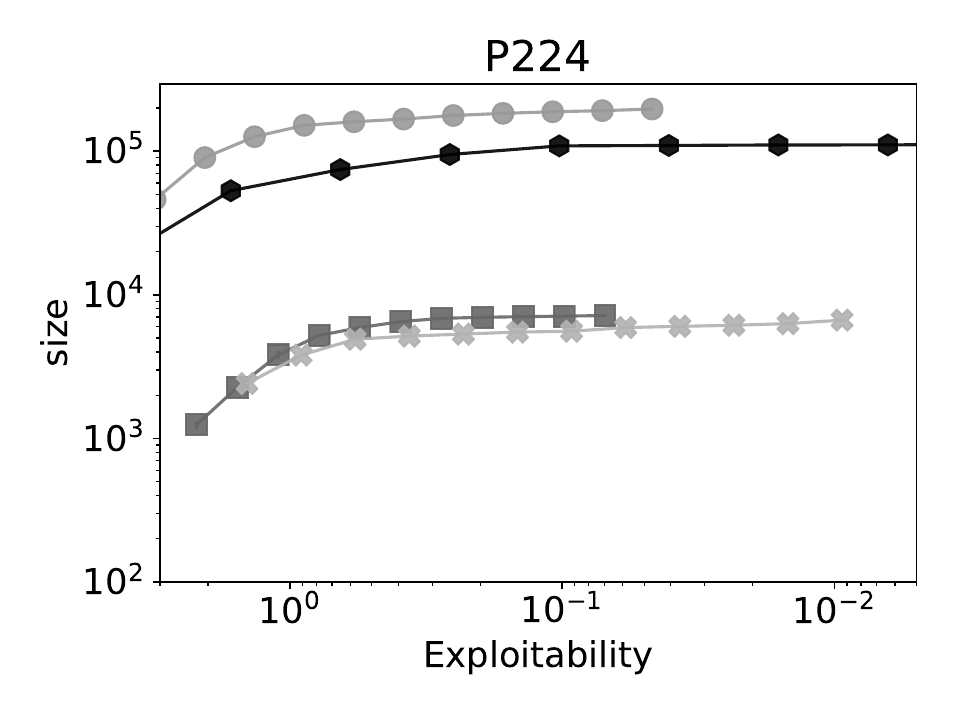}
\caption{The number of 32 bit words the algorithms store (log y-axis) as a function of the exploitability of the resulting strategies (log x-axis) for GP6, GS6 and P224.}
\label{fig:large_number_count}
\end{figure}

Finally, in Figure \ref{fig:large_number_count} we again depict the number of 32-bit words stored by the algorithms for GP6, GS6, and P224. These results further confirm that \IRCFR requires at least an order of magnitude less memory than FPIRA and DOEFG.

\subsubsection{Relative Size of \IRCFR Abstractions}
\begin{figure}[t]
\includegraphics[width=0.33\textwidth]{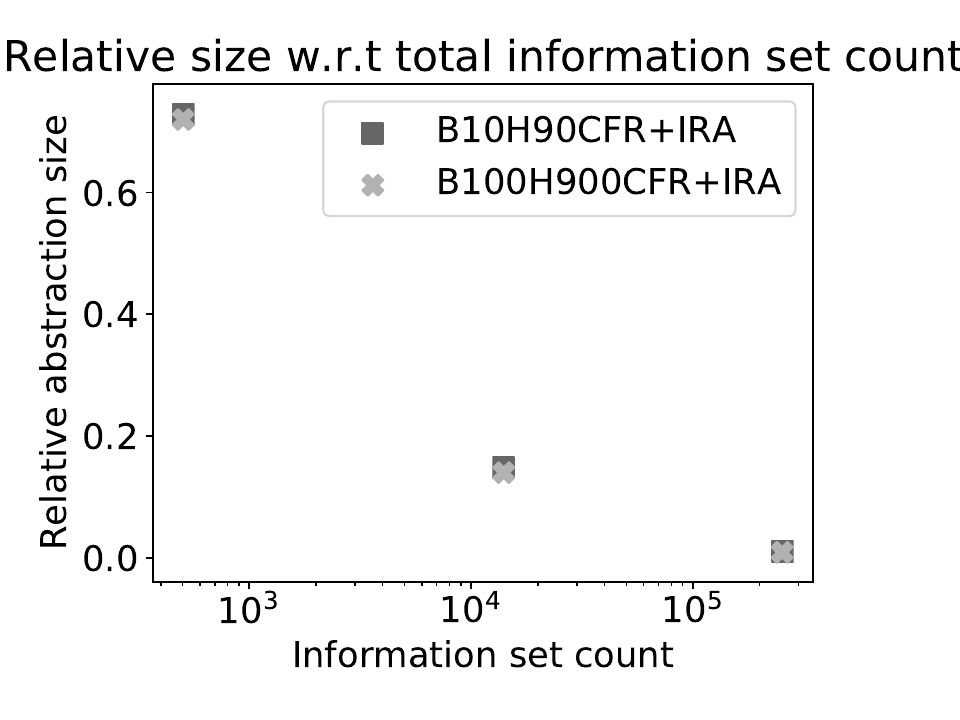}\includegraphics[width=0.33\textwidth]{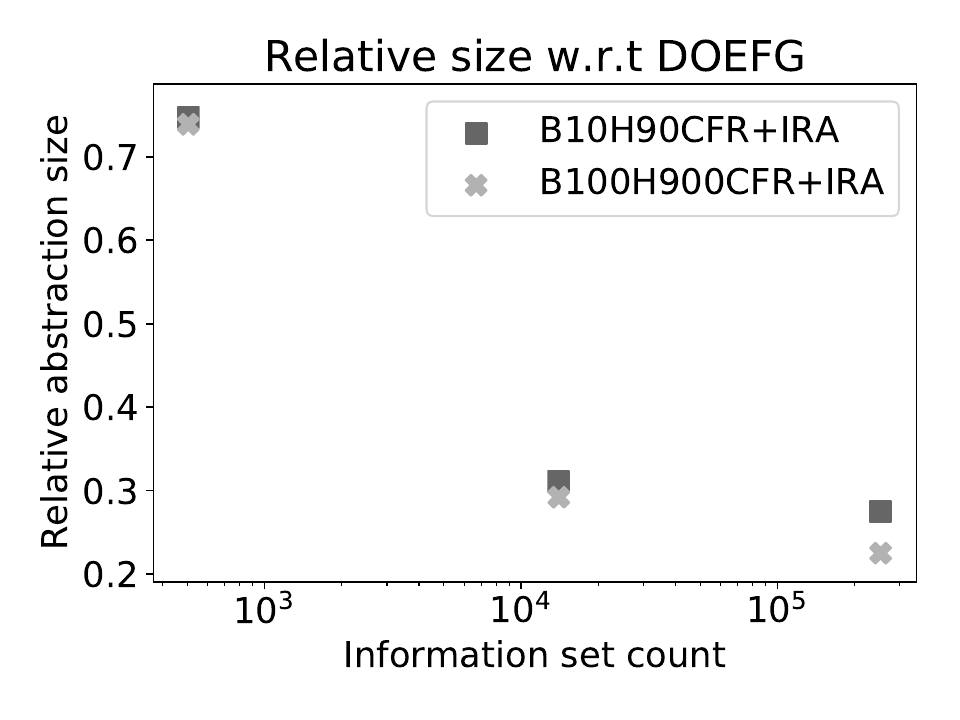}\includegraphics[width=0.33\textwidth]{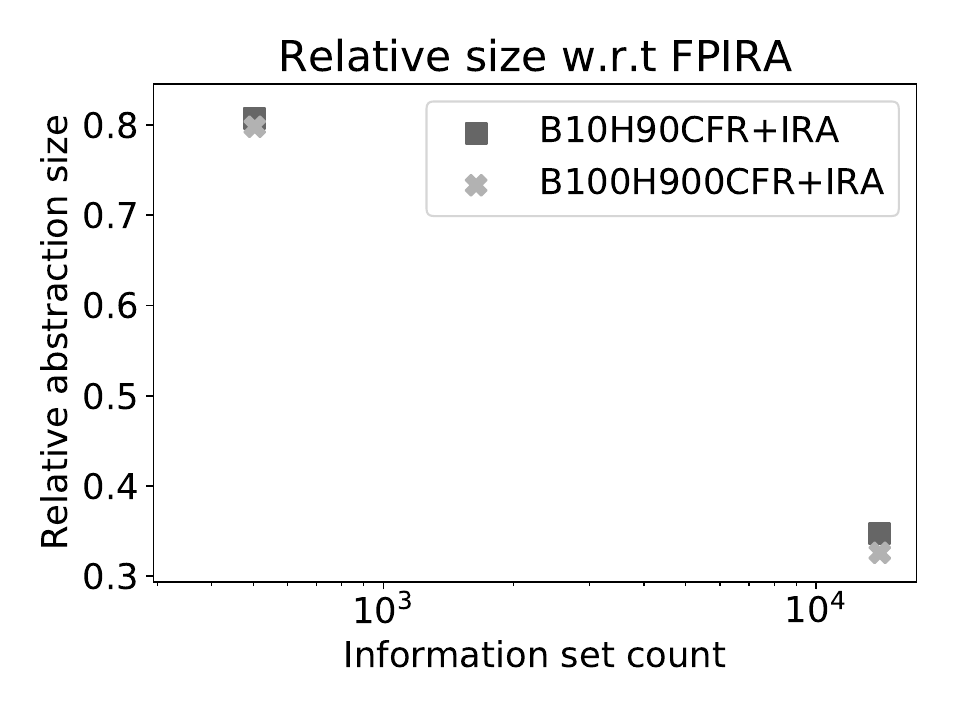}
\caption{Plots showing the relative size of the abstractions used by \IRCFR (y-axis) compared to the total information set count, the size of the restricted game used by DOEFG and the size of the abstraction used by FPIRA to compute resulting strategies with the sum of exploitabilities 0.05 as a function of the information set count of different poker instances (log x-axis).}
\label{fig:poker_sizes}
\end{figure}
 Next, we compare how the size of the abstraction of the \IRCFR scales with the size of the solved domain. In Figure \ref{fig:poker_sizes} we provide the relative size of the abstraction required by \IRCFR to compute strategies with exploitability 0.05 on P111, P222, and P224 as a function of the total information set count of the solved domains. The relative sizes reported are computed with respect to the total number of information set of the solved domains (left plot), the size of the restricted game required by DOEFG (middle plot) and the size of the abstraction required by FPIRA (right plot). The plots show that the relative size of the abstraction required by \IRCFR decreases in all three settings. This suggests that for larger domains the relative size of the abstraction built by \IRCFR will further decrease not only compared to the total information set count but also compared to the restricted game required by DOEFG and the abstraction size required by FPIRA.

\subsection{Runtime}\label{sec:experiments:runtime}
In this section, we provide a comparison of the runtime of the algorithms.
\begin{figure}[t]
\begin{center}\includegraphics[height=0.175cm]{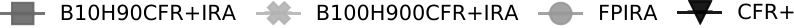}\end{center}
\vspace{-0.3cm}
\includegraphics[width=0.33\textwidth]{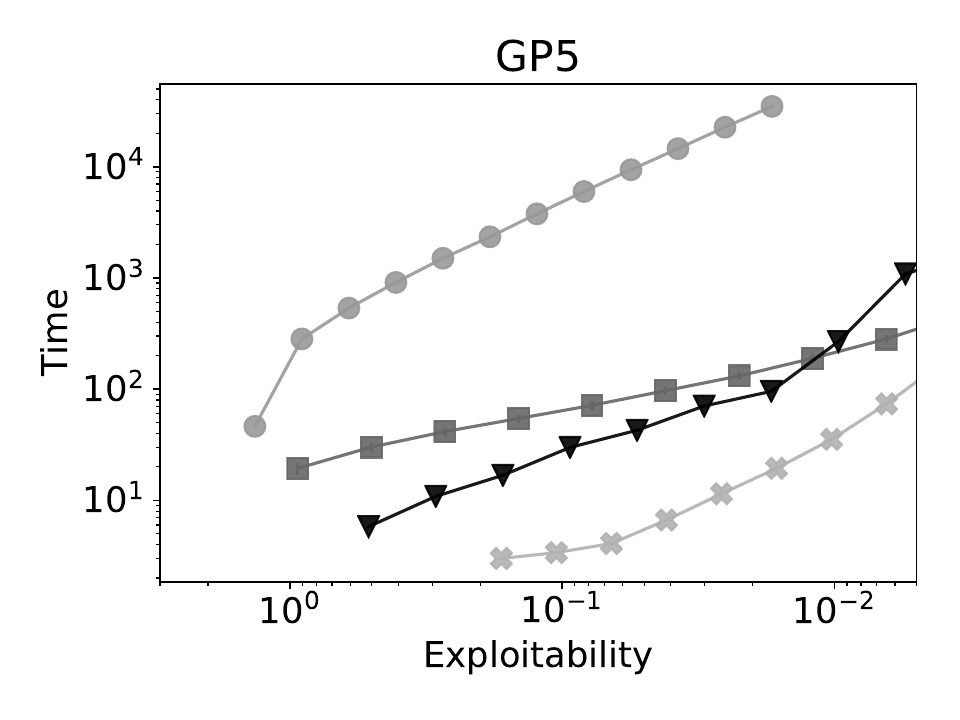}\includegraphics[width=0.33\textwidth]{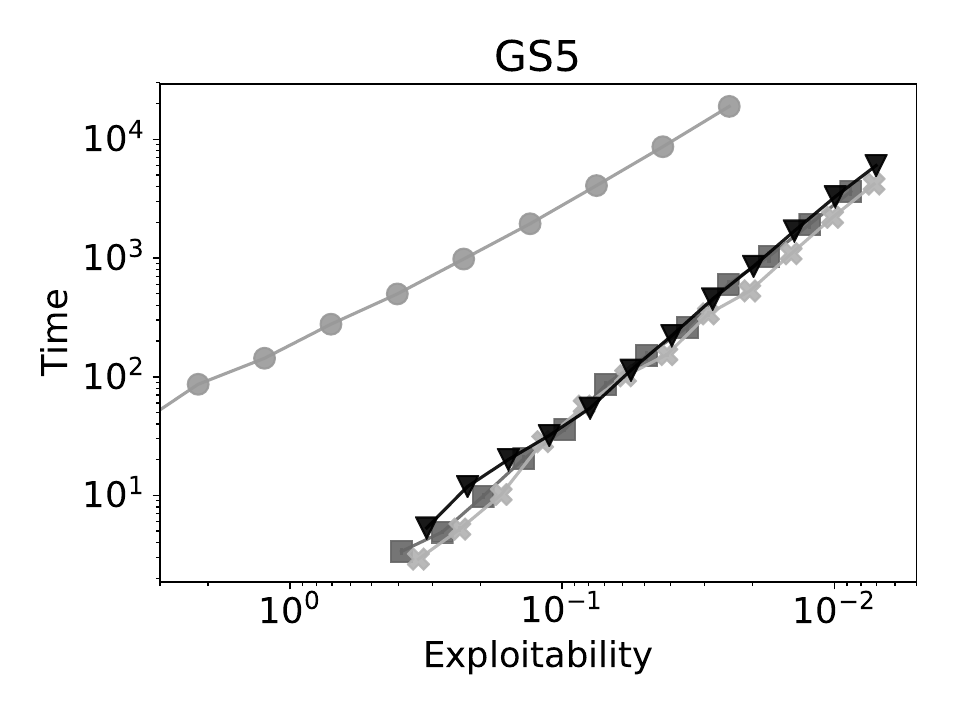}\includegraphics[width=0.33\textwidth]{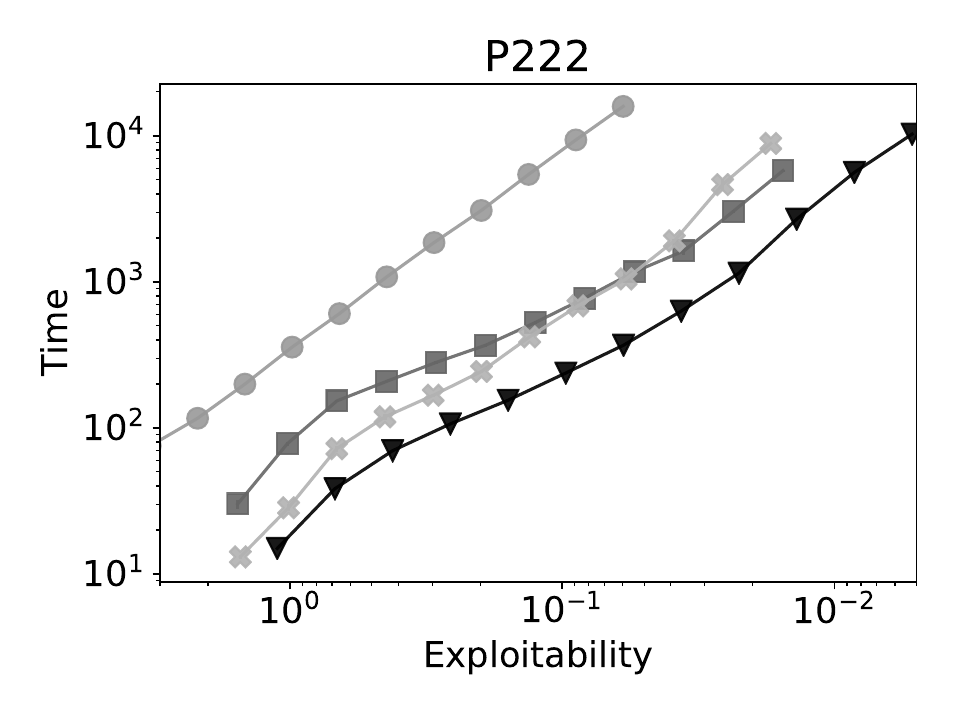}
\caption{The plots showing runtime of FPIRA and \IRCFR in seconds (log y-axis) required to reach the given sum of exploitabilities of resulting strategies of player 1 and 2 (log x-axis) for GP5, GS5, and P222 respectively.}
\label{fig:medium_runtime}
\end{figure}
The plots in Figure \ref{fig:medium_runtime} show the runtime comparison in seconds of \IRCFR, FPIRA and CFR+ for GP5, GS5, and P222. We again depict the results for \IRCFR as averages with the standard error over 10 runs with different seeds. The runtime of \IRCFR in all versions is consistently better than FPIRA and except for P222 is comparable to the runtime of CFR+. We omitted the DOEFG from this comparison since all \IRCFR, FPIRA and CFR+ use a domain-independent implementation in Java, while our implementation of DOEFG uses efficient IBM CPLEX LP solver. Furthermore, as discussed above, DOEFG was never implemented with emphasis on memory efficiency. And so it heavily exploits additional caches for the restricted game not reported in the previous section which significantly improves its runtime while increasing its memory requirements. Hence, e.g., in the case of GP5, the DOEFG took $19$ seconds to find strategies with exploitability 0.01, while B100H900\IRCFR took $38$ seconds and the B10H90\IRCFR $200$ seconds. On the other hand, DOEFG required 1 GB of memory, while both settings of \IRCFR used 38 MB. 

\begin{figure}[t]
\begin{center}\includegraphics[height=0.175cm]{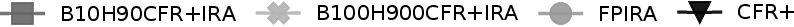}\end{center}
\vspace{-0.3cm}
\includegraphics[width=0.33\textwidth]{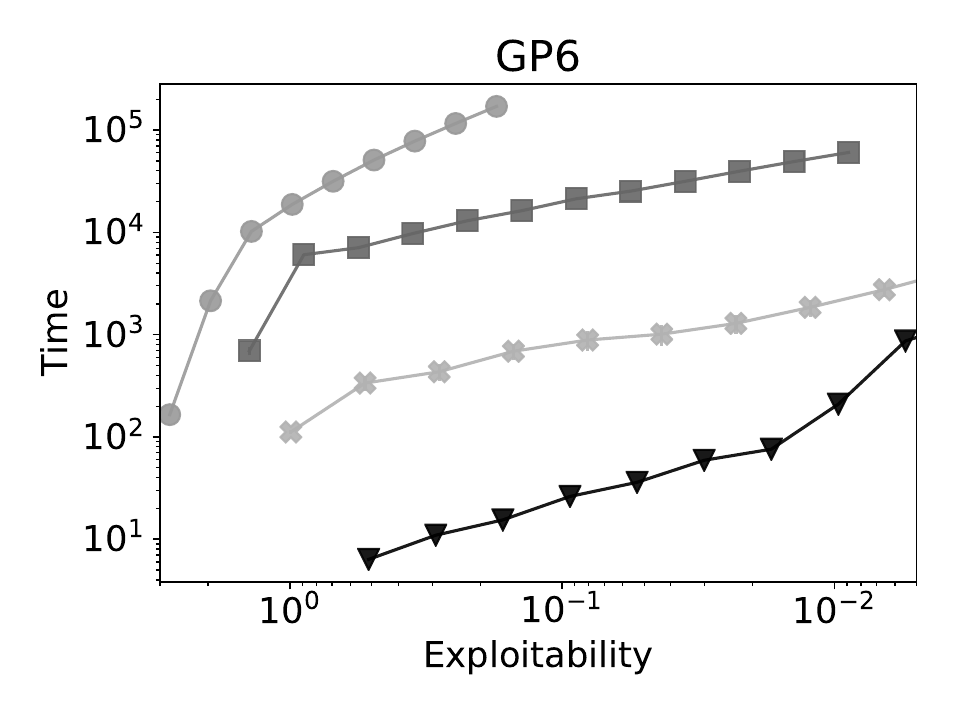}\includegraphics[width=0.33\textwidth]{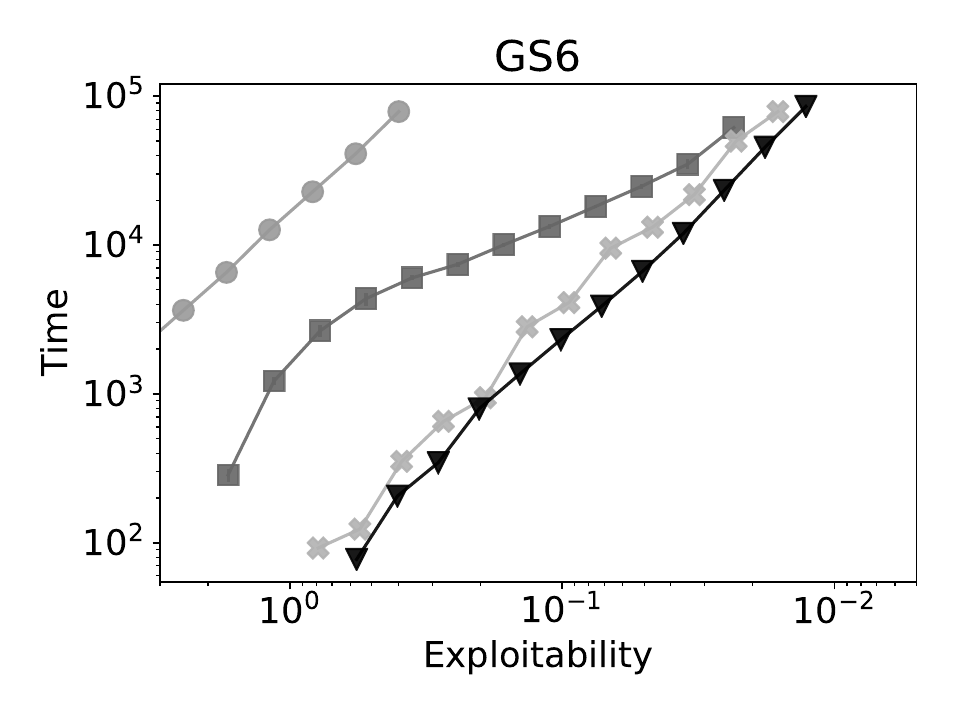}\includegraphics[width=0.33\textwidth]{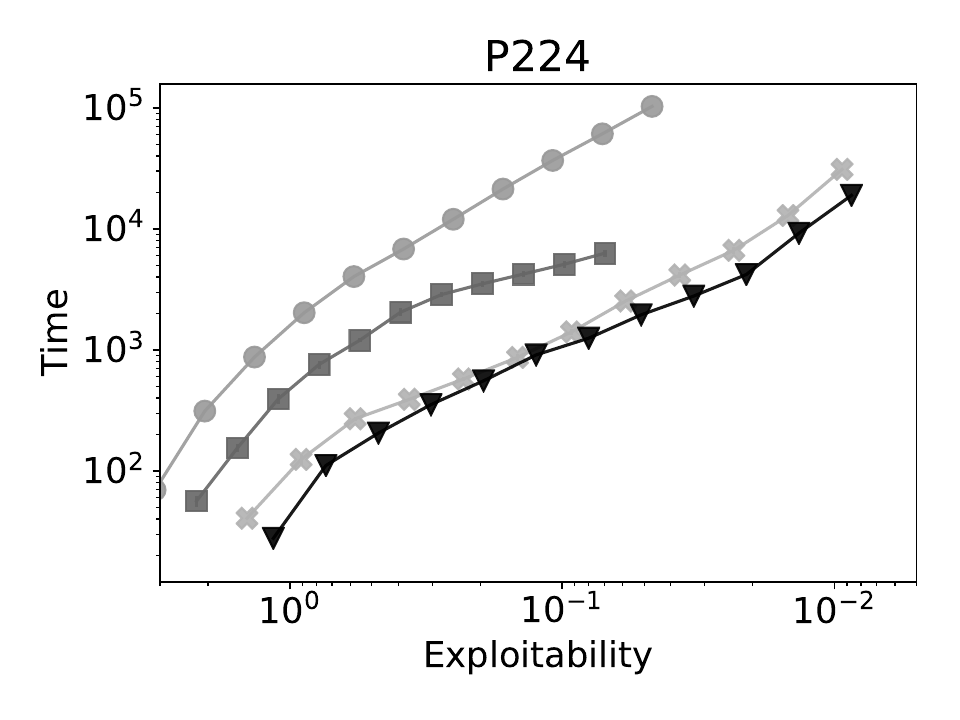}
\caption{The plots showing runtime of FPIRA and \IRCFR in seconds (log y-axis) required to reach the given sum of exploitabilities of resulting strategies of player 1 and 2 (log x-axis) for GP6, GS6 and P224 respectivelly.}
\label{fig:large_runtime}
\end{figure}

The plots in Figure \ref{fig:large_runtime} show the runtime comparison in seconds of \IRCFR, FPIRA and CFR+ for GP6, GS6 and P224. These plots further confirm the runtime dominance of \IRCFR over FPIRA. The high runtime of FPIRA is the cause for omitting the results of FPIRA for smaller exploitabilities of the resulting strategies, as the time required to compute them becomes prohibitive.  Furthermore, the results show that the runtime is worse for \IRCFR where $k_b + k_h = 100$ compared to the case with $k_b + k_h = 1000$. Additionally, there is a more profound difference between the \IRCFR runtime and the runtime of CFR+. Both observations are expected since using $k_h+k_b = 100$ and $k_h+k_b = 1000$ means that the algorithm uses less than $0.1\%$ and $1\%$ of information sets for the abstraction update in all 3 domains. Hence it takes longer to refine the abstraction to allow strategies with a smaller exploitability.

{
\subsection{Comparison of Computed Abstractions}
Further, we qualitatively compare the abstractions created by FPIRA and \IRCFR. These abstractions are automatically tuned to accommodate the data required by the individual algorithms. Hence, we compare only the final abstraction after converging close to the optimal strategy.

Let us start with single player games. In these games, optimal strategies do not need to be randomized, and the games can be solved using a single best response computation. FPIRA will find this optimal strategy in the first iteration and refine the initial abstraction to be able to represent it. If there are no chance decisions in the game, only one information set for each level of the game is necessary. With chance, only the information that leads to playing a different action after different chance outcomes is used to refine the abstraction. Using the full best response computation allows the very efficient global view to refining the abstraction. On the other hand, \IRCFR uses only a local view of counterfactual regrets in individual information sets. The abstractions updates based on the bound are rather conservative and harder to analyze, so we focus on the heuristic. The initial iterations are strongly influenced by the pure strategy used to compute the regrets in the first iteration. %
The heuristic can unnecessarily refine the abstraction closer to the root of the game before the further information sets converge to the optimal solution. Therefore, FPIRA will find smaller abstractions in single player problems than \IRCFR, and we assume the same may hold even for games, where single player problem solving is a significant component.

Next, we analyze the final abstractions computed by the algorithms in Imperfect Information Goofspiel with only three cards (3,2,1). The only Nash equilibrium of this small game is to always match the card value. There are 108 information sets in the game. Starting from the initial abstraction with only one information set per level of a game ($6$ in total), FPIRA finds the equilibrium with $20$ and IRCFR with $16$ information sets. The Nash equilibrium can be represented with only the initial 6 information sets. Hence, none of the algorithms finds the minimal possible abstraction to represent it. A notable difference between the final abstractions is in the final two levels of the game, where each player has only the last card left to play. The players clearly do not have to remember anything to do this and \IRCFR correctly identifies this. These 56 information sets from the original games are kept merged to just two by \IRCFR. 
FPIRA, however, ends with 10 information sets on these levels. The reason is that even though a strategy cannot be represented by FPIRA in the current abstraction because some information is missing close to the root, this information is not forgotten until the end of the game. Some of the information sets, leading to higher winnings for either of the players are detached from the rest. For playing the second card (with two cards left), FPIRA leaves the 14 information sets merged to 8 while \IRCFR to 12. The only information sets kept merged by \IRCFR represent forgetting whether the opponent played action 2 or 3 (player lost or tied the trick) after the player played $2$ in the first round. This is because in both cases, continuing with a 3 dominates playing 1, regardless of what the opponent does. FPIRA keeps more information sets merged, because it quickly realizes that only playing $3$ in the first round is meaningful and does not refine the information sets after the other actions.

Kuhn Poker is a small poker game with 12 information sets \cite{Kuhn1950}. FPIRA completely refines the game while computing the equilibrium. \IRCFR finishes with 11 information sets. The only information forgotten is for the second player, if she holds the strongest card, she forgets whether the first player check or bets. Regardless of this, the second player wants to take the more aggressive action (call or bet). A version of Kuhn Poker with 5 cards instead of three has 20 information sets. The final abstraction with \IRCFR has 16 information sets and with FPIRA 17 information sets. The first player has 10 information sets. Both algorithms forget whether the first player holds the strongest or the second strongest card after the second player bets. Always calling is optimal in both these cases. \IRCFR also ignores this information for the first decision of the first player, while FPIRA does not and split the information sets. The first player must randomize among the options in this information sets, which is more difficult to do without splitting the information set for FPIRA. For the second player, the algorithms find different abstractions, but both consist of 8 information sets.
\IRCFR keeps unrefined for the second player the situations where she has the strongest card regardless of the action of the opponent together with having the second strongest card and the opponent checking. These three situations belong to distinct information sets in the original game, but they are the same in the abstraction. With FPIRA, the second player forgets whether she has the strongest or the second strongest card after the first player bets. Also, it forgets the same information after the first player checks. However, she remembers whether the opponent checks or bets.

These example strategies suggest that FPIRA may perform better if the original game contains dominated actions and the best responses will never visit large parts of the game. On the other hand, \IRCFR may perform better in balancing the actions within the support without unnecessary splits of the information sets.
}

\subsection{Experiment Summary}
We have shown that \IRCFR requires at least an order of magnitude less memory to find strategies with a given exploitability compared to the memory required by FPIRA and DOEFG. Furthermore, the results suggest that when increasing the size of the solved domains, the relative memory requirements of \IRCFR will further decrease not only compared to the total information set count of the solved domain but also compared to the memory requirements of FPIRA and DOEFG.
Additionally, we have shown that the heuristic abstraction update provides a good indication of the parts of the abstraction that need to be updated and hence significantly enhances the convergence speed of \IRCFR. 

From the runtime perspective, the DOEFG is the most efficient algorithm. However, we show that its good performance comes at the cost of high memory requirements. Hence, \IRCFR proves useful, as there are domains where the memory requirements of DOEFG are too large.

\section{Conclusion}

The imperfect recall abstraction methodology can significantly reduce the memory required to solve large extensive-form games {and the size of the final solution}. However, solving the resulting imperfect recall abstract games is a {computationally hard problem and the standard algorithms for solving extensive form games are not applicable or lose their convergence guarantees.} Hence, there is only a limited amount of work that focuses on using imperfect recall abstractions.
Previous works use either very restrictive subclasses of imperfect recall abstractions~\cite{lanctot2012,kroer2016imperfect,bosansky2015combining}, heuristic approaches~\cite{waugh2009}, or use computationally complex algorithms to solve the imperfect recall abstracted game~\cite{cermak2018approximating}.

In this work, we propose a novel approach to imperfect recall information abstraction, which does not require any specific structure of the imperfect recall abstraction of a game nor does it use computationally complex algorithms to solve it. Instead, we introduce two domain-independent algorithms FPIRA and \IRCFR which can start with an arbitrary imperfect recall abstraction of the solved two-player zero-sum perfect recall extensive-form game. The algorithms simultaneously solve the abstracted game, detect the missing information causing problems and refine the abstraction to include it. This process is repeated until provable convergence to the desired approximation of the Nash equilibrium of the original game. As a consequence, the strategy from the abstracted game can be used directly in the original game without the need to use translation techniques, and the quality of such strategy is not affected by choice of the initial abstraction. {The cost for all these advantages is that unlike in standard abstraction methodology, the proposed algorithms traverse the full unabstracted game in its iterations. If the size of a game prohibits complete traversal, the standard abstraction methodology is still usable, while the proposed algorithms in their current form are not applicable anymore.
}

The experimental evaluation shows that \IRCFR requires at least an order of magnitude less memory than FPIRA and the Double Oracle algorithm (DOEFG, \cite{bosansky2014}) to solve given domains. Even when using trivial automatically build initial imperfect recall abstraction, \IRCFR is capable of closely approximating Nash equilibrium of large extensive-form games using abstraction with as little as $0.9\%$ of information sets of the original game. Furthermore, the results suggest that the relative size of the abstraction used by \IRCFR will further decrease as the size of the solved game increases. 

\subsection{Future Work}
There are several directions for future work. 

Currently, we use a mapping of actions defined by the domain description when merging the information sets as a part of the construction of the initial abstraction for FPIRA and \IRCFR. A more sophisticated mapping, which would take into account the similarity of the actions, could greatly reduce the number of refinements of the initial abstraction \IRCFR and FPIRA need to perform to find the desired approximation of the Nash equilibrium of the original game.

CFR based approaches are known to perform extremely well in poker domains \cite{moravcik2017deepstack,tammelin2015solving,brown2017superhuman}.  Since as a part of \IRCFR we use a domain-independent implementation of CFR+, we were not able to reach the full performance potential of \IRCFR. Hence, domain-specific implementation of \IRCFR, even outside of poker domain, would greatly improve its performance and would allow further significant scale-up of the algorithm.  

During the evaluation of \IRCFR and FPIRA, we focused on automatically built initial abstractions. We believe that applying \IRCFR to existing abstractions commonly used in poker would significantly improve the convergence speed of \IRCFR as these abstractions are built by domain experts and are known to perform well in practice. Applying \IRCFR to such abstractions would have two benefits. (1) \IRCFR could suggest further improvements of the abstractions in places where they are too coarse. (2) \IRCFR would use these abstractions to provide a better approximation of the Nash equilibrium of the original domain than the current solvers, as it would improve the abstraction if necessary. This, combined with the domain-specific implementation could lead to significant improvement of the quality of the strategies computed in the existing abstractions, as there are typically no guarantees that the abstraction allows computation of sufficient approximation of the Nash equilibrium in the poker domain.

{Finally, a notable limitation of the presented methods is the need to traverse the complete game tree of the unabstracted game in each iteration. Therefore, it would be interesting to investigate a combination of \IRCFR with sampling techniques, which would allow keeping the small memory requirements of the algorithm, but make it applicable to much larger games, in which it is impossible to traverse the complete game tree. Similarly, it might be possible to adapt FPIRA to use a fast approximation of best response, such as Local Best Response \cite{lisy2017lbr}, to avoid traversing the complete game tree.}

\section*{Acknowledgments}
This research was supported by the Czech Science Foundation (grant no. 15-23235S and 18-27483Y),  and by the Grant Agency of the Czech Technical University in Prague, grant No. SGS16/235/OHK3/3T/13. Computational resources were provided by the CESNET LM2015042 and the CERIT Scientific Cloud LM2015085, provided under the programme "Projects of Large Research, Development, and Innovations Infrastructures".
\section*{References}
\bibliography{refs}

\begin{thebibliography}{10}
\expandafter\ifx\csname url\endcsname\relax
  \def\url#1{\texttt{#1}}\fi
\expandafter\ifx\csname urlprefix\endcsname\relax\def\urlprefix{URL }\fi
\expandafter\ifx\csname href\endcsname\relax
  \def\href#1#2{#2} \def\path#1{#1}\fi

\bibitem{rubin2011computer}
J.~Rubin, I.~Watson, Computer poker: A review, Artificial intelligence
  175~(5-6) (2011) 958--987.

\bibitem{lisy2016-aaai}
V.~Lis{\'y}, T.~Davis, M.~Bowling, Counterfactual regret minimization in
  sequential security games, in: Proceedings of the Thirtieth AAAI Conference
  on Artificial Intelligence, AAAI Press, 2016, pp. 544--550.

\bibitem{christodoulou2008bayesian}
G.~Christodoulou, A.~Kov{\'a}cs, M.~Schapira, Bayesian combinatorial auctions,
  Automata, Languages and Programming (2008) 820--832.

\bibitem{sandholm2015steering}
T.~Sandholm, Steering evolution strategically: computational game theory and
  opponent exploitation for treatment planning, drug design, and synthetic
  biology, in: Proceedings of the Twenty-Ninth AAAI Conference on Artificial
  Intelligence, AAAI Press, 2015, pp. 4057--4061.

\bibitem{bowling2015heads}
M.~Bowling, N.~Burch, M.~Johanson, O.~Tammelin, Heads-up limit hold’em poker
  is solved, Science 347~(6218) (2015) 145--149.

\bibitem{lisy2012apr}
V.~Lis\'{y}, R.~P\'{\i}bil, J.~Stiborek, B.~Bo\v{s}ansk\'{y},
  M.~P\v{e}chou\v{c}ek, Game-theoretic approach to adversarial plan
  recognition, in: Proceedings of the 20th European Conference on Artificial
  Intelligence, ECAI'12, IOS Press, Amsterdam, The Netherlands, The
  Netherlands, 2012, pp. 546--551.

\bibitem{cowling2012}
P.~I. Cowling, E.~J. Powley, D.~Whitehouse, {Information Set Monte Carlo Tree
  Search}, Computational Intelligence and AI in Games, IEEE Transactions on 4
  (2012) 120--143.

\bibitem{long2010understanding}
J.~Long, N.~R. Sturtevant, M.~Buro, T.~Furtak, Understanding the success of
  perfect information monte carlo sampling in game tree search, in: Proceedings
  of the Twenty-Fourth AAAI Conference on Artificial Intelligence, AAAI Press,
  2010, pp. 134--140.

\bibitem{moravcik2017deepstack}
M.~Morav{\v c}{\'\i}k, M.~Schmid, N.~Burch, V.~Lis{\'y}, D.~Morrill, N.~Bard,
  T.~Davis, K.~Waugh, M.~Johanson, M.~Bowling, Deepstack: Expert-level
  artificial intelligence in heads-up no-limit poker, Science.

\bibitem{lisy2015online}
V.~Lis{\'y}, M.~Lanctot, M.~Bowling, Online monte carlo counterfactual regret
  minimization for search in imperfect information games, in: Proceedings of
  the 2015 International Conference on Autonomous Agents and Multiagent
  Systems, International Foundation for Autonomous Agents and Multiagent
  Systems, 2015, pp. 27--36.

\bibitem{brown2017superhuman}
N.~Brown, T.~Sandholm, {Superhuman AI for heads-up no-limit poker: Libratus
  beats top professionals}, Science (2017) eaao1733.

\bibitem{gilpin2007}
A.~Gilpin, T.~Sandholm, Lossless abstraction of imperfect information games,
  Journal of the ACM (JACM) 54~(5) (2007) 25.

\bibitem{kroer2014extensive}
C.~Kroer, T.~Sandholm, {Extensive-Form Game Abstraction with Bounds}, in:
  Proceedings of the fifteenth ACM conference on Economics and computation,
  ACM, 2014, pp. 621--638.

\bibitem{brown2015simultaneous}
N.~Brown, T.~Sandholm, Simultaneous abstraction and equilibrium finding in
  games, in: Proceedings of the 24th International Conference on Artificial
  Intelligence, AAAI Press, 2015, pp. 489--496.

\bibitem{Wired08}
J.~Rehmeyer, N.~Fox, R.~Rico, Ante up, human: The adventures of polaris the
  poker-playing robot, Wired 16.12 (2008) 186--191.

\bibitem{acpc17}
Annual computer poker competition: 2017 crosstable,
  \url{http://www.computerpokercompetition.org/downloads/competitions/2017/xtable/},
  accessed: 2018-11-01.

\bibitem{bard2016online}
N.~D. Bard, Online agent modelling in human-scale problems, Ph.D. thesis,
  University of Alberta (2016).

\bibitem{fang2017paws}
F.~Fang, T.~H. Nguyen, R.~Pickles, W.~Y. Lam, G.~R. Clements, B.~An, A.~Singh,
  B.~C. Schwedock, M.~Tambe, A.~Lemieux, Paws—a deployed game-theoretic
  application to combat poaching, AI Magazine.

\bibitem{brown2018depth}
N.~Brown, T.~Sandholm, B.~Amos, Depth-limited solving for imperfect-information
  games, arXiv preprint arXiv:1805.08195.

\bibitem{bard2013online}
N.~Bard, M.~Johanson, N.~Burch, M.~Bowling, Online implicit agent modelling,
  in: Proceedings of the 2013 international conference on Autonomous agents and
  multi-agent systems, International Foundation for Autonomous Agents and
  Multiagent Systems, 2013, pp. 255--262.

\bibitem{johanson2012}
M.~Johanson, N.~Bard, N.~Burch, M.~Bowling, Finding optimal abstract strategies
  in extensive-form games, in: Proceedings of the Twenty-Sixth AAAI Conference
  on Artificial Intelligence, AAAI Press, 2012, pp. 1371--1379.

\bibitem{burch2014solving}
N.~Burch, M.~Johanson, M.~Bowling, Solving imperfect information games using
  decomposition, in: Proceedings of the Twenty-Eighth AAAI Conference on
  Artificial Intelligence, AAAI Press, 2014, pp. 602--608.

\bibitem{bosansky2014-jair}
B.~Bošanský, C.~Kiekintveld, V.~Lisý, M.~Pěchouček, {An Exact
  Double-Oracle Algorithm for Zero-Sum Extensive-Form Games with Imperfect
  Information}, {Journal of Artificial Intelligence Research} 51 (2014)
  829--866.

\bibitem{shi2000abstraction}
J.~Shi, M.~L. Littman, Abstraction methods for game theoretic poker, Computers
  and Games 2063 (2000) 333--345.

\bibitem{billings2003approximating}
D.~Billings, N.~Burch, A.~Davidson, R.~Holte, J.~Schaeffer, T.~Schauenberg,
  D.~Szafron, Approximating game-theoretic optimal strategies for full-scale
  poker, in: IJCAI, 2003, pp. 661--668.

\bibitem{Gilpin07:Abstraction}
A.~Gilpin, T.~Sandholm, T.~B. S{\o}rensen, Potential-aware automated
  abstraction of sequential games, and holistic equilibrium analysis of texas
  hold'em poker, in: Proceedings of the National Conference on Artificial
  Intelligence, Vol.~22, 2007, p.~50.

\bibitem{vonStengel96}
B.~von Stengel, {Efficient Computation of Behavior Strategies}, Games and
  Economic Behavior 14 (1996) 220--246.

\bibitem{zinkevich2008regret}
M.~Zinkevich, M.~Johanson, M.~H. Bowling, C.~Piccione, Regret minimization in
  games with incomplete information., in: Advances in Neural Information
  Processing Systems, 2007, pp. 1729--1736.

\bibitem{Hoda2010}
S.~Hoda, A.~Gilpin, J.~Pe\~{n}a, T.~Sandholm, {Smoothing Techniques for
  Computing Nash Equilibria of Sequential Games}, Mathematics of Operations
  Research 35~(2) (2010) 494--512.

\bibitem{kroer2016imperfect}
C.~Kroer, T.~Sandholm, {Imperfect-Recall Abstractions with Bounds in Games},
  in: Proceedings of the seventeenth ACM conference on Economics and
  computation, ACM, 2016, pp. 459--476.

\bibitem{dalkey1953equivalence}
N.~Dalkey, Equivalence of information patterns and essentially determinate
  games, Contributions to the Theory of Games 2 (1953) 217.

\bibitem{cermak2017ijcai}
J.~{\v{C}}erm{\'a}k, B.~Bo{\v{s}}ansky, V.~Lisy, An algorithm for constructing
  and solving imperfect recall abstractions of large extensive-form games, in:
  Proceedings of the 26th International Joint Conference on Artificial
  Intelligence, AAAI Press, 2017, pp. 936--942.

\bibitem{brown1949some}
G.~W. Brown, Some notes on computation of games solutions, Tech. rep., DTIC
  Document (1949).

\bibitem{tammelin2014cfr+}
O.~Tammelin, Cfr+, CoRR, abs/1407.5042.

\bibitem{bosansky2014}
B.~Bosansky, C.~Kiekintveld, V.~Lisy, M.~Pechoucek, An exact double-oracle
  algorithm for zero-sum extensive-form games with imperfect information,
  Journal of Artificial Intelligence Research (2014) 829--866.

\bibitem{gilpin2006competitive}
A.~Gilpin, T.~Sandholm, A competitive texas hold'em poker player via automated
  abstraction and real-time equilibrium computation, in: Proceedings of the
  National Conference on Artificial Intelligence, Vol.~21, Menlo Park, CA;
  Cambridge, MA; London; AAAI Press; MIT Press; 1999, 2006, p. 1007.

\bibitem{lanctot2012}
M.~Lanctot, N.~Burch, M.~Zinkevich, M.~Bowling, R.~G. Gibson, No-regret
  learning in extensive-form games with imperfect recall, in: Proceedings of
  the 29th International Conference on Machine Learning (ICML-12), 2012, pp.
  65--72.

\bibitem{waugh2009}
K.~Waugh, M.~Zinkevich, M.~Johanson, M.~Kan, D.~Schnizlein, M.~H. Bowling, {A
  Practical Use of Imperfect Recall}, in: SARA, Citeseer, 2009.

\bibitem{cermak2018approximating}
J.~{\v{C}}erm{\'a}k, B.~Bo{\v{s}}ansk{\'y}, K.~Hor{\'a}k, V.~Lis{\'y},
  M.~P{\v{e}}chou{\v{c}}ek, Approximating maxmin strategies in imperfect recall
  games using a-loss recall property, International Journal of Approximate
  Reasoning 93 (2018) 290--326.

\bibitem{hawkin2011automated}
J.~Hawkin, R.~Holte, D.~Szafron, Automated action abstraction of imperfect
  information extensive-form games, in: Proceedings of the Twenty-Fifth AAAI
  Conference on Artificial Intelligence, AAAI Press, 2011, pp. 681--687.

\bibitem{hawkin2012using}
J.~Hawkin, R.~C. Holte, D.~Szafron, Using sliding windows to generate action
  abstractions in extensive-form games, in: Proceedings of the Twenty-Sixth
  AAAI Conference on Artificial Intelligence, AAAI Press, 2012, pp. 1924--1930.

\bibitem{gilpin2008heads}
A.~Gilpin, T.~Sandholm, T.~B. S{\o}rensen, A heads-up no-limit texas hold'em
  poker player: discretized betting models and automatically generated
  equilibrium-finding programs, in: Proceedings of the 7th international joint
  conference on Autonomous agents and multiagent systems-Volume 2,
  International Foundation for Autonomous Agents and Multiagent Systems, 2008,
  pp. 911--918.

\bibitem{basilico2011automated}
N.~Basilico, N.~Gatti, Automated abstractions for patrolling security games,
  in: Proceedings of the Twenty-Fifth AAAI Conference on Artificial
  Intelligence, AAAI Press, 2011, pp. 1096--1101.

\bibitem{sandholm2012lossy}
T.~Sandholm, S.~Singh, Lossy stochastic game abstraction with bounds, in:
  Proceedings of the 13th ACM Conference on Electronic Commerce, ACM, 2012, pp.
  880--897.

\bibitem{Kuhn1953}
H.~W. Kuhn, {Extensive Games and the Problem of Information}, Annals of
  Mathematics Studies.

\bibitem{tammelin2015solving}
O.~Tammelin, N.~Burch, M.~Johanson, M.~Bowling, Solving heads-up limit texas
  hold'em., in: IJCAI, 2015, pp. 645--652.

\bibitem{robinson1951iterative}
J.~Robinson, An iterative method of solving a game, Annals of mathematics
  (1951) 296--301.

\bibitem{Karlin2003mathematical}
S.~Karlin, Mathematical methods and theory in games, programming, and
  economics, Vol.~2, Courier Corporation, 2003.

\bibitem{Daskalakis2014counter}
C.~Daskalakis, Q.~Pan, A counter-example to karlin's strong conjecture for
  fictitious play, in: Annual Symposium on Foundations of Computer Science,
  IEEE, 2014, pp. 11--20.

\bibitem{heinrich2015fictitious}
J.~Heinrich, M.~Lanctot, D.~Silver, Fictitious self-play in extensive-form
  games., in: ICML, 2015, pp. 805--813.

\bibitem{johanson2011accelerating}
M.~Johanson, K.~Waugh, M.~Bowling, M.~Zinkevich, Accelerating best response
  calculation in large extensive games, in: IJCAI, Vol.~11, 2011, pp. 258--265.

\bibitem{Ross71Goofspiel}
S.~M. Ross, {G}oofspiel --- the game of pure strategy, Journal of Applied
  Probability 8~(3) (1971) 621--625.

\bibitem{lanctot2013monte}
M.~Lanctot, V.~Lis{\'y}, M.~H. Winands, Monte carlo tree search in simultaneous
  move games with applications to goofspiel, in: Workshop on Computer Games,
  Springer, 2013, pp. 28--43.

\bibitem{Kuhn1950}
H.~W. Kuhn, A simplified two-person poker, Contributions to the Theory of Games
  1.

\bibitem{bosansky2015combining}
B.~Bo{\v s}ansk{\' y}, A.~X. Jiang, M.~Tambe, C.~Kiekintveld, Combining compact
  representation and incremental generation in large games with sequential
  strategies., in: Proceedings of the Twenty-Ninth AAAI Conference on
  Artificial Intelligence, 2015, pp. 812--818.

\bibitem{lisy2017lbr}
V.~Lis{\'y}, M.~Bowling, Eqilibrium approximation quality of current no-limit
  poker bots, in: AAAI Workshop - Technical Report, Vol. WS-17-01 - WS-17-15,
  2017, pp. 361--366.

\end{thebibliography}
\end{document}